\documentclass[10pt,a4paper]{article}

\usepackage{setspace}
\newif\ifsiamart
\makeatletter%
\@ifclassloaded{siamart171218}%
  {\siamarttrue}%
  {\siamartfalse}%
\makeatother%

\usepackage[margin=1in]{geometry}
\usepackage{titlesec}
\usepackage{color}
\usepackage{graphicx}
\usepackage{float}
\usepackage{amsmath}
\usepackage{amssymb}
\usepackage{amsfonts}
\usepackage{mathtools}
\usepackage{dsfont}
\usepackage{todonotes}
\usepackage{enumitem}
\usepackage[
    style=trad-abbrv,
    maxcitenames=2,
    doi=false,
    url=false,
    isbn=false,
    backend=biber]{biblatex}
\usepackage{csquotes}
\usepackage[english]{babel}
\usepackage{bold-extra}

\ifsiamart\else
\usepackage{amsthm}
\usepackage{hyperref}
\hypersetup{%
    linktocpage=true,
    colorlinks=true,
    citecolor=magenta,
    linkcolor=blue,
}
\usepackage[nameinlink,capitalise]{cleveref}
\newcommand{\email}[1]{\href{mailto:#1}{#1}}
\fi

\usepackage{setspace}
\usepackage{mathrsfs}

\DeclarePairedDelimiter\abs{\lvert}{\rvert}
\DeclarePairedDelimiter\ip{\langle}{\rangle}
\DeclarePairedDelimiter\norm{\lVert}{\rVert}

\DeclareMathOperator{\sgn}{sgn}

\makeatletter
\renewcommand\section{\@startsection {section}{1}{\z@}%
                               {-3.5ex \@plus -1ex \@minus -.2ex}%
                               {2.3ex \@plus.2ex}%
                               {\normalfont\large\bfseries}}
\renewcommand\subsection{\@startsection{subsection}{2}{\z@}%
                                 {-3.25ex\@plus -1ex \@minus -.2ex}%
                                 {1.5ex \@plus .2ex}%
                                 {\normalfont\bfseries}}
\makeatother

\ifsiamart
\newtheorem{remark}[theorem]{Remark}
\newtheorem{assumption}{Assumption}
\newtheorem{example}[theorem]{Example}
\else
\newtheorem{theorem}{Theorem}
\newtheorem{lemma}[theorem]{Lemma}
\newtheorem{corollary}[theorem]{Corollary}
\newtheorem{proposition}[theorem]{Proposition}
\newtheorem{assumption}[theorem]{Assumption}

\theoremstyle{remark}
\newtheorem{remarkx}[theorem]{Remark}
\theoremstyle{definition}

\newtheorem{examplex}[theorem]{Example}
\newenvironment{remark}
  {\pushQED{\qed}\remarkx}
  {\popQED\endremarkx}

\newenvironment{example}
  {\pushQED{\qed}\examplex}
  {\popQED\endexamplex}
\fi

\ifsiamart\else
\Crefname{figure}{Figure}{Figures}
\Crefname{examplex}{Example}{Examples}
\Crefname{example}{Example}{Examples}
\crefname{lemma}{Lemma}{Lemmas}
\crefname{remark}{Remark}{Remarks}
\crefname{assumption}{Assumption}{Assumptions}
\crefname{proposition}{Proposition}{Propositions}
\crefname{section}{Section}{Sections}
\crefname{subsection}{Subsection}{Subsections}
\crefname{equation}{}{}
\Crefname{equation}{Equation}{Equations}
\fi

\ifsiamart\else

\newcommand{\orcid}[1]{\href{https://orcid.org/#1}{\includegraphics[width=.4cm]{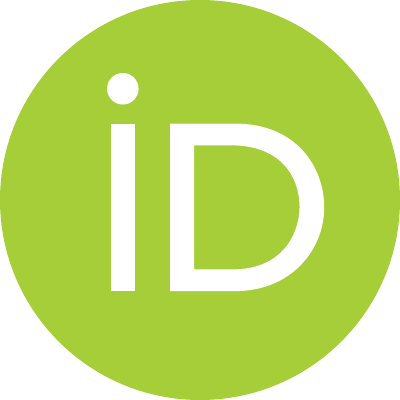}}}
\fi

\DeclareMathOperator{\Span}{Span}
\DeclareMathOperator{\supp}{supp}
\DeclareMathOperator{\closure}{cl}
\DeclareMathOperator{\range}{Ran}
\DeclareMathOperator{\kernel}{Ker}
\DeclareMathOperator{\diag}{diag}
\renewcommand{\d}{\mathrm d}
\newcommand{\e}{\mathrm e}
\newcommand{\dummy}{\,\cdot\,}
\newcommand{\placeholder}{\,\cdot\,}
\newcommand{\normal}{\mathcal N}
\newcommand{\vect}[1]{\boldsymbol{\mathbf #1}}
\newcommand{\mat}{\mathit}
\newcommand{\real}{\mathbb R}
\newcommand{\integer}{\mathbb Z}
\newcommand{\nat}{\mathbb N}
\newcommand{\torus}{\mathbb T}
\newcommand{\domain}{\mathbb D}
\newcommand{\grad}{\nabla}
\newcommand{\laplacian}{\Delta}
\newcommand{\expect}{\mathbb E}

\newcommand{\id}{\mathcal I}
\newcommand{\matid}{\mathrm{id}}
\renewcommand{\t}{\mathsf T}

\newcommand{\Z}[1]{Z[#1]}
\newcommand{\ZN}[1]{Z_N[#1]}
\newcommand{\step}{\delta}

\renewcommand{\leq}{\leqslant}
\renewcommand{\geq}{\geqslant}
\newcommand{\smoothcompact}{C^{\infty}_{\rm c}}

\DeclareFieldFormat{volume}{volume \textbf{#1}}
\DeclareFieldFormat[article]{volume}{\textbf{#1}}

\ifsiamart\else
\fi

\title{Optimal importance sampling for overdamped Langevin dynamics}
\author{%
    M. Chak\thanks{%
        LJLL, Sorbonne Université (\email{martin.chak@sorbonne-universite.fr})
        \orcid{0000-0001-5012-0172}
    }%
    \and T. Lelièvre\thanks{%
        CERMICS, \'Ecole des Ponts, France \& MATHERIALS project-team, Inria Paris (\email{tony.lelievre@enpc.fr})
        \orcid{0000-0002-3412-113X}
    }%
    \and G. Stoltz\thanks{%
        CERMICS, \'Ecole des Ponts, France \& MATHERIALS project-team, Inria Paris (\email{gabriel.stoltz@enpc.fr})
        \orcid{0000-0002-2797-5938}
    }%
    \and U. Vaes\thanks{%
        MATHERIALS project-team, Inria Paris \& CERMICS, \'Ecole des Ponts (\email{urbain.vaes@inria.fr})
        \orcid{0000-0002-7629-7184}
    }%
}
\date{\today}

\begin{document}
\maketitle

\begin{abstract}
    Calculating averages with respect to multimodal probability distributions is often necessary in applications.
    Markov chain Monte Carlo (MCMC) methods to this end,
    which are based on time averages along a realization of a Markov process ergodic with respect to the target probability distribution,
    are usually plagued by a large variance due to the metastability of the process.
    In this work,
    we mathematically analyze an importance sampling approach for MCMC methods that rely on the overdamped Langevin dynamics.
    Specifically, we study an estimator based on an ergodic average along a realization of an overdamped Langevin process for a modified potential.
    The estimator we consider incorporates a reweighting term in order to rectify the bias that would otherwise be introduced by this modification of the potential.
    We obtain an explicit expression in dimension 1 for the biasing potential that minimizes the asymptotic variance of the estimator for a given observable,
    and propose a general numerical approach for approximating the optimal potential in the multi-dimensional setting.
    We also investigate an alternative approach where,
    instead of the asymptotic variance for a given observable,
    a weighted average of the asymptotic variances corresponding to a class of observables is minimized.
    Finally, we demonstrate the capabilities of the proposed method by means of numerical experiments.
\end{abstract}

\section{Introduction}%
\label{sec:introduction}

\subsection{Context}
In many applications ranging from Bayesian inference to statistical physics and computational biology,
it is often necessary to calculate expectations with respect to high-dimensional probability distributions of the form
\begin{equation}
    \label{eq:target_measure}
    \mu = \frac{\e^{-V}}{Z}, \qquad Z = \int_{\domain^d} \e^{-V},
\end{equation}
where $\domain \in \{\real, \torus\}$,
with $\torus := \real / 2\pi \integer$ the one-dimensional torus,
and $V\colon \domain^d \rightarrow \real$ is a potential energy function (confining if $\domain = \real$ and periodic if $\domain = \torus$)
such that $\e^{-V}$ is Lebesgue integrable.
In the context of Bayesian inference,
the distribution~$\mu$ usually describes likelihoods of the possible values of an unknown parameter given some observed data~\cite{MR3363508,MR2652785},
while in statistical physics,
the distribution $\mu$ assigns probabilities to the possible configurations of a molecular system. In the latter setting,
averages with respect to~$\mu$ give access to macroscopic properties of the system,
such as the heat capacity or equations of state relating pressure, density and temperature~\cite{evans-moriss,Tuckerman,MR2681239,MR3362507,AT17}.

The first systematic approach to sampling probability distributions originates from the~1950s with the seminal work of \citeauthor{metropolis1953equation}~\cite{metropolis1953equation}.
In 1970, Hastings generalized this approach and proposed a sampling method~\cite{MR3363437}
which was later recognized as a particular case of what is now known as a Markov chain Monte Carlo (MCMC) method.
The MCMC approach to sampling is based on the use of a Markov process that admits the target probability distribution as unique invariant measure.
A simple yet widely used Markov process that is ergodic with respect to~$\mu$ under appropriate conditions on the potential~$V$
is the overdamped Langevin dynamics,
\begin{equation}
    \label{eq:overdamped}
    \d Y_t = -\nabla V(Y_t) \, \d t + \sqrt{2} \, \d W_t,
    %\qquad Y_0 = y_0\in\domain^d,
\end{equation}
where $W_t$ denotes a standard $d$-dimensional Wiener process.
Under appropriate assumptions on the potential~$V$,
the average with respect to~$\mu$ of an observable $f \in L^1(\mu)$ can be approximated by a time average along a realization of the solution to this equation:
\begin{equation}\label{eq:basicconv}
    \mu^T(f):=\frac{1}{T} \int_{0}^{T} f(Y_t) \, \d t \xrightarrow[T \to \infty]{\rm a.s.} \int_{\domain^d} f \, \d \mu  =: \mu(f)  =: I,
\end{equation}
see e.g.~\cite[Theorem 5.1]{MR112175},~\cite{MR885138},
 ~\cite{MR2248986} and references therein, ~\cite[Chapter~9 and proof of Theorem~1.6.2]{MR3616034} and~\cite[Proposition~2]{MR3905536}.
In practice, it is necessary to discretize the dynamics~\eqref{eq:overdamped},
and the resulting discrete-time Markov process is generally ergodic with respect to not $\mu$
but a probability measure $\mu_{\Delta t}$ differing from~$\mu$ at order $\Delta t^{\alpha}$,
for some exponent~$\alpha$ larger than or equal to the weak order of convergence of the scheme.
The bias introduced by the discretization can usually be bounded from above as a function of the time step;
such estimates were first obtained by Talay and Tubaro~\cite{TT90} for general SDEs.
They were later made precise for implicit schemes for Langevin and overdamped Langevin dynamics in~\cite{MR1931266},
and then refined in works such as \cite{MR2669996,DF12,ACZ15,LMS16}.
Alternatively, it may be possible to consider the numerical scheme as a proposal to be accepted or rejected in a Metropolis--Hastings scheme,
so that the resulting discrete-time process is also ergodic with respect to~$\mu$.
This is the Metropolis-adjusted Langevin algorithm (MALA)~\cite{rossky1978brownian,MR1440273,MR1625691}.

In this paper,
we focus on the continuous-time dynamics~\eqref{eq:overdamped}
and show that modifying the potential~$V$, in combination with importance sampling,
can be used for variance reduction.
Importance sampling is widely used to make the sampling of high dimensional probability measures easier; see for instance~\cite{APSS17} for a review. The idea of using importance sampling in the context of MCMC methods was already suggested in Hastings' 1970 paper,
see~\cite[Section 5]{MR3363437}.
If $(X_t)_{t \geq 0}$ is a Markov process that is ergodic with respect to a probability measure
\begin{equation}
    \label{eq:mu_u}
    \mu_{U} = \frac{\e^{-V - U}}{\Z{U}},
    \qquad \Z{U} = \int_{\domain^d} \e^{-V-U},
\end{equation}
where $U\colon \domain^d \to \real$ is a smooth function such that $\e^{-V-U}$ is Lebesgue integrable over $\domain^d$,
then~$\mu(f)$ may be approximated by
\begin{equation}
    \label{eq:estimator}
    \mu^T_U(f) = \frac{\displaystyle \int_0^T (f \e^U)(X_t) \, \d t}{\displaystyle \int_0^T(\e^U)(X_t) \, \d t}.
\end{equation}
Like $\mu^T(f)$, this estimator converges to~$I$ almost surely in the limit as $T \to \infty$,
and it does not require the knowledge of the normalization constants $Z$ and $\Z{U}$.
The main objective of this work is to study the properties of~$\mu^T_U(f)$
when~$(X_t)_{t \geq 0}$ is the solution to the overdamped Langevin dynamics with the potential~$V+U$:
\begin{equation}
    \label{eq:mix}
    \d X_t = -\nabla V(X_t) \, \d t -\nabla U(X_t) \d t + \sqrt{2} \, \d W_t.
    %\qquad X_0 = x_0\in\domain^d.
\end{equation}
In particular, we study whether it is possible to find a biasing potential~$U$ such that the asymptotic variance of $\mu_T^U(f)$,
which we define more precisely in~\cref{lemma:asymptotic_variance},
is minimized for either a single observable~(\cref{sec:minimizing_the_asymptotic_variance_for_one_observable}) or a class of observables~(\cref{sec:minimizing_the_expected_asymptotic_variance}). Let us also mention here the recent work~\cite{CHR22} where optimal importance sampling is performed over a class of distributions.

Before considering the MCMC estimator~\eqref{eq:estimator},
it is instructive to present background material on importance sampling in the~independent and identically distributed (i.i.d.)\ setting.
This is the aim of~\cref{sub:iid}.
In all the theoretical results presented in this work,
we assume that the following assumptions on~$V$ and~$f$ are satisfied,
even when this is not explicitly mentioned.
\begin{assumption}
    \label{assumption:assumptions_V}
    $~$
    \begin{itemize}
        \item
            The potential $V$ is smooth over $\domain^d$ (in particular $\supp(\e^{-V}) = \domain^d$).

        \item
            The function $\e^{-V}$ is Lebesgue integrable over $\domain^d$.

        \item
            Any observable $f$ considered (just one in \cref{sub:iid,sec:minimizing_the_asymptotic_variance_for_one_observable} and set of them in~\cref{sec:minimizing_the_expected_asymptotic_variance}) is smooth,
            integrable with respect to the probability measure~$\mu \propto \e^{-V}$,
            and not~$\mu$-almost everywhere constant.
    \end{itemize}
\end{assumption}

% It is desirable for numerical simulation of~\eqref{eq:mix}
% that the potential~$U$ is sufficiently regular,
% which leads us to consider a different optimization problem in~\cref{sec:minimizing_the_expected_asymptotic_variance}.

\subsection{Our contributions}
The contributions of this paper are the following:
\begin{itemize}

    \item
        In the one-dimensional setting ($d=1$) either with $\domain = \torus$ or $\domain = \real$,
        and for a given observable~$f$,
        we obtain an explicit expression for the biasing potential~$U$ in~\eqref{eq:estimator}--\eqref{eq:mix} that is optimal in terms of asymptotic variance.
        We also prove that when $\domain = \real$,
        the asymptotic variance~$\sigma^2_f[U]$ of $\mu_U^T(f)$, viewed as a functional of $U$,
        is convex.

    \item
        In the general multi-dimensional setting,
        we obtain an expression for the $L^2(\mu)$ functional derivative of $\sigma^2_f[U]$ with respect to~$U$,
        and we propose a gradient descent approach for finding a minimizer of~$\sigma^2_f[U]$.
        We also prove that any minimizer of $\sigma^2_f[U]$ is necessarily singular when $\domain = \torus$.

    \item
        We propose a method for minimizing the asymptotic variance over a class of observables.
        More precisely, we present an approach for minimizing the average asymptotic variance
        % when the observable is a simple Gaussian random field.
        when a simple Gaussian probability distribution is placed on the observable.
        We demonstrate through theoretical results and numerical experiments that this approach usually leads to a smooth optimizer,
        and may thus be more suitable for applications.

    \item
        We present examples and numerical experiments
        illustrating the properties of the optimal biasing potential and the performance of the method,
        both in the one-dimensional case and the multi-dimensional setting.
\end{itemize}

\paragraph{Plan of the paper.}
The remainder of the paper is organized as follows.
We begin in~\cref{sub:iid} by presenting background material on importance sampling in the i.i.d.\ setting.
In~\cref{sec:minimizing_the_asymptotic_variance_for_one_observable},
we investigate the problem of minimizing the asymptotic variance for a given observable,
first in the one-dimensional case and then in the multi-dimensional setting.
In~\cref{sec:minimizing_the_expected_asymptotic_variance},
we generalize the approach to the problem of minimizing the asymptotic variance over a class of observables.
Examples and numerical experiments are presented in \cref{sec:examples_and_numerical_experiments}.
\Cref{sec:conclusion} is reserved for conclusions and perspectives for future works.
This section is followed by four appendices:
\cref{sec:auxiliary_results,remark:constant_A} contain proofs of auxiliary results,
\cref{sec:second_variation_of_the_asymptotic_variance} presents a derivation of the second variation of the asymptotic variance,
and \cref{sec:numerical_discretization_of_the_Poisson_equation} provides a detailed analysis of the numerical scheme employed for approximating the functional derivative of the asymptotic variance.

\section{Background: the i.i.d.\ setting}
\label{sub:iid}

We recall in this section the expression of the optimal importance distribution~$\mu_{U}$ in the setting where~$I$
is estimated from i.i.d.\ samples from~$\mu_U$,
as opposed to the ergodic approach in~\eqref{eq:estimator}.
Specifically,
we consider the estimator
\begin{equation}
    \label{eq:iid_importance_sampling}
    \mu_U^N(f) :=
    \displaystyle \frac
    {\displaystyle \sum_{n=1}^{N} (f \e^U)(X^{n})}
    {\displaystyle \sum_{n=1}^{N} (\e^U)(X^{n})}
    = I + \displaystyle \frac
    {\displaystyle \sum_{n=1}^{N} \left((f-I) \e^U\right)(X^{n})}
    {\displaystyle \sum_{n=1}^{N} (\e^U)(X^{n})},
\end{equation}
where~$(X^n)_{n \in \nat}$ are i.i.d.\ samples from~$\mu_{U}$.
The right-most expression is not useful in practice because the value of~$I$ is unknown,
but this expression is convenient for the theoretical analysis.
We first comment briefly on the connection between this estimator and the estimator~\eqref{eq:estimator} in \cref{sub:connection},
then obtain an expression for the asymptotic variance of the estimator~\eqref{eq:iid_importance_sampling} in~\cref{sub:iid_asymvar}.
and finally prove bounds on the asymptotic variance in~\cref{sub:iid_optimal}.

\subsection{\texorpdfstring{Connection between the estimators~\eqref{eq:estimator} and~\eqref{eq:iid_importance_sampling}}{Connection between the estimators}}
\label{sub:connection}
The estimators~\eqref{eq:estimator} and~\eqref{eq:iid_importance_sampling} can be viewed as two limiting cases,
corresponding to the limits $\tau \to 0$ and~$\tau \to \infty$ respectively,
of the estimator given in the second row and right-most column of~\cref{table:importance_sampling_estimators}.
The latter estimator is based not on the full solution to~\eqref{eq:mix}
but on discrete periodic evaluations of it with a period~$\tau$.

The analysis presented in this paper can be repeated for the estimators in the middle column of~\cref{table:importance_sampling_estimators},
which can be employed when the normalization constants~$Z$ and~$\Z{U}$ are known.
In this case, different optimal potentials are obtained.
However, since the normalization constants are usually unknown in high-dimensional settings,
we focus in most of this paper on the self-normalized estimators in the right-most column of~\cref{table:importance_sampling_estimators}.

\begin{table}[ht]
    \centering
    \begin{tabular}{|c|c|c|}
         \hline
         \phantom{$\Big($}
         & $Z$ and $\Z{U}$ are \textbf{known} & $Z$ and $\Z{U}$ are \textbf{unknown}
         \\ \hline
         Continuous ($\tau = 0$)
         & \phantom{$\frac{\Big(}{\Big(}$}\hspace{-4mm}
            \(
                \displaystyle
                \frac{1}{T} \left( \frac{\Z{U}}{Z} \right)
                \int_0^T \left(f\e^U\right)(X_t) \, \d t
            \)
         &
            \(
                \displaystyle \frac
                {\int_0^T \left(f\e^U\right)(X_t) \, \d t}
                {\int_0^T(\e^U)(X_t) \, \d t}
            \)
         \\ \hline
         $0 < \tau < \infty$
         &  \phantom{$\frac{\Big(}{\Big(}$}\hspace{-4mm}
            \(
                \displaystyle
                \frac{1}{N} \left( \frac{\Z{U}}{Z} \right)
                \sum_{n=0}^{N-1} \left(f\e^U\right)(X_{n \tau})
            \)
         &
            \(
                \displaystyle \frac
                {\sum_{n=0}^{N-1} \left(f\e^U\right)(X_{n \tau})}
                {\sum_{n=0}^{N-1} (\e^U)(X_{n \tau})}
            \)
         \\ \hline
         i.i.d ($\tau = \infty$)
         & \phantom{$\frac{\Big(}{\Big(}$}\hspace{-4mm}
            \(
                \displaystyle
                \frac{1}{N} \left( \frac{\Z{U}}{Z} \right)
                \sum_{n=1}^{N} \left(f\e^U\right)(X^{n})
            \)
         &
            \(
                \displaystyle \frac
                {\sum_{n=1}^{N} \left(f\e^U\right)(X^{n})}
                {\sum_{n=1}^{N} (\e^U)(X^{n})}
            \)
         \\ \hline
    \end{tabular}
    \caption{%
        Classification of importance sampling estimators according to two criteria:
        the correlation between samples,
        and use of the normalization constants~$Z$ and~$\Z{U}$.
        Here~$(X_t)_{t \in \real_{\geq 0}}$ is a solution to~\eqref{eq:mix}
        and~$(X^n)_{n \in \nat}$ are i.i.d.\ samples from~$\mu_{U}$.
    }
    \label{table:importance_sampling_estimators}
\end{table}

\subsection{Asymptotic variance}
\label{sub:iid_asymvar}

In the whole \cref{sub:iid},
we suppose that~\cref{assumption:assumptions_V} holds,
in particular that the potential $V$ is smooth,
but we do not assume that the function~$U$ is smooth;
we assume only
that~$U\colon \domain^d \to \real \cup \{\infty\}$ is measurable and such that
the following minimal requirements are satisfied:
\begin{itemize}
    \item
        the function $\e^{-V-U}$ is Lebesgue integrable with $\Z{U} > 0$;

    \item
        the estimator~\eqref{eq:iid_importance_sampling} converges to~$I$ almost surely, i.e.\ it holds that
        \begin{equation}
            \label{eq:unbiased}
        \int_{\domain^d} (f-I) \e^U \d \mu_U = 0
        \qquad \Leftrightarrow \qquad
            \int_{\supp(\mu_U)} (f-I) \d \mu = 0.
        \end{equation}
        % Notice that the integral may be rewritten as $\int_{\supp(\mu_U)} (f-I) \d \mu$  = 0,
        Here
        \(
            \supp(\mu_U) = \closure\{x \in \domain^d : U(x) + V(x) < \infty\}
        \)
        is the support of the measure~$\mu_U$,
        with $\closure$ the closure.
        In this work,
        we adopt the usual convention in measure theory that $0 \cdot + \infty = + \infty \cdot 0 = 0$.
    \item
         the function $(f-I)\e^U$ is square-integrable with respect to~$\mu_U$,
         so that the central limit theorem can be applied.
\end{itemize}
We denote by $\mathcal U$ the set of functions~$U$ that satisfy these conditions.
The set~$\mathcal U$ depends on~$V$ and~$f$,
but since these are considered to be fixed data of the problem,
we do not explicitly indicate this dependence in the notation.
If $U \in \mathcal U$ and $(X^n)_{n\in \nat}$ are i.i.d. samples with law $\mu_U$,
then it holds by~\eqref{eq:unbiased} and the central limit theorem that
\begin{equation}
    \label{eq:iid_convergence_numerator}
    \frac{1}{\sqrt{N}} \sum_{n=1}^{N} \left((f-I) \e^U\right) \left(X^{n}\right)
    \xrightarrow[N\to \infty]{\rm Law} \mathcal N\left(0, \int_{\domain^d} \abs*{(f-I) \e^U}^2 \, \d \mu_{U}\right),
\end{equation}
where $\mathcal N(m, C)$ denotes the Gaussian distribution with mean~$m$ and (co)variance $C$.
On the other hand,
the law of large numbers gives that
\begin{equation}
    \label{eq:iid_convergence_denominator}
    \frac{1}{N} \sum_{n=1}^{N} \left(\e^U\right)\left(X^{n}\right)
    \xrightarrow[N \to \infty]{\rm a.s.}
    \int_{\domain^d} \e^{U} \d \mu_U = \frac{1}{Z[U]}\int_{\supp(\mu_U)} \e^{-V}
    = \frac{\mathscr Z[U]}{\Z{U}},
\end{equation}
where we introduced
\[
    \mathscr Z[U] := \int_{\supp(\mu_U)} \e^{-V}.
\]
Note that $\mathscr Z[U] \leq Z$ in general,
with equality if and only if $\supp(\mu_U) = \supp(\mu)$.
Combining~\eqref{eq:iid_convergence_numerator} and~\eqref{eq:iid_convergence_denominator} and using Slutsky's lemma,
we conclude that
\begin{align}
    \label{eq:asym_var_iid}
    \sqrt{N} \bigl( \mu^N_U(f) - I\bigr)
    \xrightarrow[N \to \infty]{\rm Law} \normal\bigl(0, s^2_f[U]\bigr),
    \qquad
    s^2_f[U] := \frac{\Z{U}^2}{\mathscr Z[U]^2} \displaystyle\int_{\domain^d} \abs*{(f-I) \e^U}^2 \, \d \mu_{U}.
\end{align}
The variance~$s^2_f[U]$ of the asymptotic normal distribution is the quantity we wish to minimize.
In the next section, we obtain sharp bounds from below on the asymptotic variance~$s^2_f[U]$.

\subsection{Explicit optimal potential}
\label{sub:iid_optimal}
In order to prepare the proof of the main result of this section,
\cref{proposition:background_infimum},
we first give a preparatory lemma.
We introduce the functional $\widehat s^2_f\colon \mathcal U \to \real \cup \{\infty\}$ given by
\begin{equation}
    \label{eq:aux_asymvar_iid}
    \widehat s^2_f[U] = \frac{Z[U]}{Z^2} \int_{\domain^d} \lvert f-I \rvert^2 \e^{U-V}.
\end{equation}
Given the convention that $0 \cdot + \infty = + \infty \cdot 0 = 0$,
the quantity $\widehat s^2_f[U]$ is well defined as an element of~$\real \cup \{\infty\}$.
Note that $\widehat s_f^2[U]$ coincides with the asymptotic variance~$s^2_f[U]$ if and only if~$\supp (\mu_U) = \domain^d$.
% For the "only if" part: suppose that supp(μ_U) is not D, and let us show that the quantities are different. Two cases:
% - If it does not hold that f - I = 0 on D \setminus supp(μ_U), then σ² is finite and \widehat σ² is infinite -> different.
% - If it does hold, then the numerators are the same but \mathscr Z < Z -> different.
\begin{lemma}
    \label{lemma:preparatory_asymvar_iid}
    It holds that
    \begin{equation}
        \label{eq:def_infimum_iid}
        \min_{U \in \mathcal U} \widehat s^2_f[U] = \frac{1}{Z^2} \left( \int_{\domain^d} \abs{f-I} \, \d \mu \right)^2 =: s^*_f.
    \end{equation}
    In addition, the minimum is achieved for
    \begin{equation}
        \label{eq:minimizer_lemma_iid}
        U = U_*^{\rm iid} := - \log \abs{f-I} \in \mathcal U,
    \end{equation}
    with the convention that~$\log(0) = - \infty$.
\end{lemma}
\begin{proof}
    We first prove that $\widehat s^2_f[U] \geq s^*_f$ for all $U \in \mathcal U$.
    If $\widehat s^2_f[U] = \infty$,
    then this inequality is clear, so we assume from now on that $\widehat s^2_f[U] < \infty$.
    Using the Cauchy--Schwarz inequality,
    we have
    \[
        \widehat s^2_f[U]
        = \frac{1}{Z^2} \int_{\domain^d} \e^{-U-V} \int_{\domain^d} \lvert f-I \rvert^2 \e^{U-V}
        \geq \frac{1}{Z^2} \left(\int_{\domain^d} \lvert f-I \rvert \e^{-V}\right)^2
        = s^*_f.
    \]
    The statement that the infimum is achieved for~$U_*^{\rm iid}$ in~\eqref{eq:minimizer_lemma_iid} follows from a substitution in~\eqref{eq:aux_asymvar_iid}.
    Note that $U_*^{\rm idd}$ indeed belongs to~$\mathcal U$.
\end{proof}
\begin{remark}
    The function~$U_*^{\rm iid}$ in~\eqref{eq:minimizer_lemma_iid} is not smooth,
    because $f - I$ admits at least one root in $\domain^d$.
\end{remark}

Although~$\widehat s^2_f[U]$ does not coincide with $s^2_f[U]$ in general,
\cref{lemma:preparatory_asymvar_iid} still contains useful information.
Indeed, by regularizing~$U_*^{\rm iid}$ in~\eqref{eq:minimizer_lemma_iid},
we prove in~\cref{proposition:background_infimum} that~$s^*_f$ is a sharp bound from below on the actual asymptotic variance~$s^2_f[U]$
over a class of ``nice'' biasing potentials.
% In view of~\cref{lemma:preparatory_asymvar_iid},
% one may wonder the infimum of the actual asymptotic variance~$s^2_f[U]$ over the set~$\mathcal U$ is also equal to~$s^*_f$ given in~\eqref{eq:def_infimum_iid}.
% As we shall see in~\cref{proposition:background_infimum},
% this is not the case: the infimum of~$s^2_f[U]$ over the set~$\mathcal U$ is always zero under~\cref{assumption:assumptions_V}.
% However, if we restrict~$U$ to a set~$\mathcal U_0 \subset \mathcal U$ of~``nice'' potential functions,
% then~$s^2_f[U]$ is bounded from below by the positive constant~$s^*_f$.
Specifically, let us introduce
\[
    \mathcal U_0 = \Bigl\{ U \in \mathcal U :  \supp(\abs{f-I}\mu) \subset \supp(\mu_U) \Bigr\}.
\]
Conditions of the type $\supp(\abs{f-I}\mu) \subset \supp(\mu_U)$ are used in the importance sampling literature;
see, for example, equation (1.1) in~\cite[Section V.1]{MR2331321}.
Notice that, if this condition is satisfied,
then~\eqref{eq:unbiased} is also satisfied,
but the converse is not true;
in other words, $\mathcal U_0$ is a proper subset of~$\mathcal U$.
We are now ready to state and prove the main result of this section.
We discuss in~\cref{remark:optimal_idd_not_minimizer} after the proof the existence of an optimal potential achieving the infimum over~$\mathcal U_0$.
Here and in the rest of this paper,
the notation~$C^{\infty}_{\rm c}(\domain^d)$ denotes the set of smooth functions with compact support over~$\domain^d$.

\begin{proposition}
    \label{proposition:background_infimum}
    Recall that~\cref{assumption:assumptions_V} is assumed to hold throughout this paper.
    Then,
    \begin{itemize}
        \item
            % The infimum of the asymptotic variance over the full set~$\mathcal U$ is zero:
            It holds that
            \(
                \inf_{U \in \mathcal U} s^2_f[U] = 0.
            \)

        \item
            If~$0 \in \mathcal U_0$,
            then~
            % $s^*_f$ is the infimum of the asymptotic variance over the set~$\mathcal U_0 \subset \mathcal U$:
            \(
                \inf_{U \in \mathcal U_0} s^2_f[U] = s^*_f.
            \)

        \item
            If~$0 \in \mathcal U_0$,
            then
            \(
                \inf_{U \in C^{\infty}_{\rm c}(\domain^d)} s^2_f[U] = s^*_f.
            \)
            % ~$s^*_f$ is the \emph{infimum} of the asymptotic variance $s^2_f[U]$
            % over the set of smooth biasing potentials $U$ with compact support:
    \end{itemize}
\end{proposition}
\begin{remark}
    In view of~\cref{assumption:assumptions_V}, the condition~$0 \in \mathcal U_0$ is equivalent to the condition $f \in L^2(\mu)$.
    We use the former condition in the statement of~\cref{proposition:background_infimum} in order to underline the similarity with~\cref{proposition:sharpness} in the MCMC setting.
\end{remark}
\begin{proof}
    We divide the proof into three parts,
    corresponding to the three items in the statement.

    \noindent\textbf{First item}.
    The idea here is to construct an importance distribution~$\mu_U$ concentrated in a small ball containing a point where~$f = I$.
    Considering balls centered at such a point is usually not sufficient,
    because the condition~\eqref{eq:unbiased} may not be satisfied.
    Take~$\varepsilon > 0$ and let $g_{\varepsilon}\colon \domain^d \to \real$ denote the function
    \[
        g_{\varepsilon}(x) = \frac{1}{|B_{\varepsilon}(0)|}\int_{B_{\varepsilon}(x)} \bigl(f - I\bigr) \e^{-V},
    \]
    where $B_{\varepsilon}(x) \subset \domain^d$ is the open ball of radius~$\varepsilon$ centered at~$x$
    and $|B_{\varepsilon}(0)|$ is the volume of this ball.
    By~\cref{assumption:assumptions_V},
    there exists $(x_1, x_2) \in \domain^d \times \domain^d$ such that $f(x_1) > I$ and $f(x_2) < I$.
    Since $f$ is smooth,
    there is~$\epsilon > 0$ such that $g_{\varepsilon}(x_1) > 0$ and $g_{\varepsilon}(x_2) < 0$ for all $\varepsilon \in (0, \epsilon]$.
    Therefore, by the intermediate value theorem,
    there exists for all $\varepsilon \in (0, \epsilon]$ a point~$x_{\varepsilon} \in \domain^d$ on the segment joining~$x_1$ and~$x_2$ such that~$g_{\varepsilon}(x_{\varepsilon}) = 0$.
    We define
    \[
        U_{\varepsilon}(x) =
        \begin{cases}
            0 & \text{if } \abs{x - x_{\varepsilon}} \leq \varepsilon, \\
            + \infty & \text{otherwise.}
        \end{cases}
    \]
    By construction $U_{\varepsilon} \in \mathcal U$,
    and we calculate from~\eqref{eq:asym_var_iid} that
    \begin{equation*}
        s^2_f[U_{\varepsilon}] =
        \frac
        {\int_{B_{\varepsilon}(x_{\varepsilon})} \lvert f-I \rvert^2 \e^{-V}}
        {\int_{B_{\varepsilon}(x_{\varepsilon})} \e^{-V}}.
    \end{equation*}
    By the mean value theorem,
    there is~$z_{\varepsilon} \in B_{\varepsilon}(x_{\varepsilon})$ such that $s^2_f[U_{\varepsilon}] = \lvert f(z_{\varepsilon}) -I \rvert^2$.
    Since $g_{\varepsilon}(x_{\varepsilon}) = 0$,
    the function~$f-I$ necessarily has a zero in~$B_{\varepsilon}(x_{\varepsilon})$.
    Additionally, since~$(x_{\varepsilon})_{\varepsilon \in (0, \epsilon]}$ is bounded,
    the function~$f-I$ restricted to $\bigcup_{\varepsilon \in (0, \epsilon]} B_{\varepsilon}(x_{\varepsilon})$ is Lipschitz continuous,
    with some constant~$L$.
    From this we obtain that~$s^2_f[U_{\varepsilon}] \leq (2L \varepsilon)^2$,
    and taking the limit~$\varepsilon \to 0$ in this equation,
    we deduce the first item.
    % https://math.stackexchange.com/questions/850873/limit-of-an-integral-as-the-measure-of-the-region-of-integration-approaches-zer?rq=1

    \vspace{.2cm}
    \noindent \textbf{Second item}.
    By the Cauchy--Schwarz inequality,
    the following lower bound holds for all~$U \in \mathcal U$:
    \begin{equation}
        \label{eq:lower_bound_asym_var_iid}
        s^2_f[U]
        \geq \frac{\Z{U}^2}{\mathscr Z[U]^2} \left( \int_{\domain^d} \abs{f-I} \e^U \, \d \mu_{U} \right)^2
        = \frac{1}{\mathscr Z[U]^2} \left( \int_{\supp(\mu_U)} \abs{f-I} \e^{-V} \right)^2.
    \end{equation}
    For~$U \in \mathcal U_0$,
    the right-most expression in~\eqref{eq:lower_bound_asym_var_iid} equals
    \[
        \frac{1}{\mathscr Z[U]^2} \left( \int_{\domain^d} \abs{f-I} \e^{-V} \right)^2 \geq s^*_f,
    \]
    where we used that $\mathscr Z[U] \leq Z$.
    Therefore, it holds that $s^2[U] \geq s^*_f$ for all $U \in \mathcal U_0$.
    That $s^*_f$ is in fact the \emph{infimum} of~$s^2_f$ over~$\mathcal U_0$ will follow from the third item,
    because $C^{\infty}_{\rm c} \subset \mathcal U_0$.
    \vspace{.2cm}

    \noindent \textbf{Third item}.
    For~$\varepsilon > 0$,
    let $\varrho_{\varepsilon}\colon \real \to \real$ denote the mollifier
    \begin{equation}
        \label{eq:mollification}
        \varrho_{\varepsilon}(z) = \varepsilon^{-1} \varrho\bigl( \varepsilon^{-1} z \bigr),
        \qquad
                                 \varrho(z) =
                                 \begin{cases}
                                     k\exp \left( - \frac{1}{1 - \abs{z}^2} \right) & \text{if } \abs{z} \leq 1, \\
                                     0 & \text{if } \abs{z} > 1,
                                 \end{cases}
    \end{equation}
    with $k>0$ such that~$\int_{\real} \varrho(z) \, \d z = 1$.
    Let us introduce the smooth regularization of the absolute value function given by~$\abs{z}_\varepsilon = \varrho_{\varepsilon} \star {\rm abs}(z)$,
    where ${\rm abs}(z) = |z|$ is the absolute value function.
    Notice that~$\abs{z}_\varepsilon = \abs{z}$ for $\lvert z \rvert \geq \varepsilon$ and that
    \[
        \forall z \in \real, \qquad
        0 \leq |z|_{\varepsilon} - \lvert z \rvert \leq |0|_{\varepsilon} \leq \varepsilon.
    \]
    The first inequality follows from the convexity of the absolute value function,
    and the second inequality comes from an application of the reverse triangle inequality:
    \[
        |z|_{\varepsilon} - |z|
        = \int_{\real} \left( \lvert z - y \rvert - \abs{z} \right) \varrho_{\varepsilon} (y)  \, \d y
        \leq \int_{\real} \abs{y} \, \varrho_{\varepsilon}(y) \, \d y = |0|_{\varepsilon}.
    \]
    Moreover,
    let~$U_{\varepsilon}\colon \domain \to \real$ be the smooth biasing potential given by
    \[
        U_{\varepsilon}(x) =
        \begin{cases}
            \chi_{\varepsilon}(x) \Bigl( -  \log \bigl\lvert f-I \bigr\rvert_{\varepsilon} \Bigr) & \textrm{if } \domain = \real,\\
            - \log \bigl\lvert f-I \bigr\rvert_{\varepsilon} & \textrm{if } \domain = \torus,
        \end{cases}
    \]
    where $\chi_{\varepsilon} = \varrho \star \mathds 1_{[-\varepsilon^{-1}, \varepsilon^{-1}]}$.
    Since $\abs{z}_{\varepsilon} \geq \abs{0}_{\varepsilon} > 0$ for all $z \in \real$,
    the function $U_{\varepsilon}$,
    which is a regularization of~\eqref{eq:minimizer_lemma_iid},
    is well-defined everywhere and uniformly bounded from above.
    The choice of~$U_*^{\rm iid}$ as the function we regularize is natural in view of~\cref{lemma:preparatory_asymvar_iid}
    and the fact that the inequality in~\eqref{eq:lower_bound_asym_var_iid} is an equality for~$U = U_*^{\rm iid}$.

    We now show that $s^2_f[U_{\varepsilon}]$ converges to the lower bound in~\eqref{eq:lower_bound_asym_var_iid} in the limit as $\varepsilon \to 0$.
    We focus in this proof on the case where $\domain = \real$,
    but the reasoning applies verbatim to the case where $\domain = \torus$.
    Using~\eqref{eq:asym_var_iid}, we have that
    \[
        s^2_f[U_{\varepsilon}]
        = \frac{\Z{U_{\varepsilon}}}{Z^2} \int_{\real^d} \left(\frac{\lvert f-I \rvert}{\lvert f-I \rvert_{\varepsilon}^{\chi_{\varepsilon}}}\right)^2  \lvert f-I \rvert_{\varepsilon}^{\chi_{\varepsilon}} \e^{-V}.
    \]
    Since $0 \leq \chi_{\varepsilon}(x) \leq 1$ for all $x \in \domain^d$,
    it holds for all $\varepsilon \in (0, 1)$ that
    \[
        \left\{
            \begin{aligned}
                \lvert f-I \rvert_{\varepsilon}^{\chi_{\varepsilon}}
        &\leq \left( \lvert f-I \rvert + 1 \right)^{\chi_{\varepsilon}} \leq \abs{f-I} + 1 \\
        \frac{\lvert f-I \rvert^2}{\lvert f-I \rvert_{\varepsilon}^{\chi_{\varepsilon}}}
        &\leq \lvert f-I \rvert^{2- \chi_{\varepsilon}} \leq \bigl(\lvert f-I \rvert + 1\bigr)^{2 - \chi_{\varepsilon}}
        \leq \bigl(\lvert f-I \rvert + 1\bigr)^2.
            \end{aligned}
        \right.
    \]
    Since $f \in L^2\left(\e^{-V}\right)$ by assumption,
    the right-hand sides of these inequalities are integrable.
    Therefore, using dominated convergence,
    we obtain that
    \[
        \Z{U_{\varepsilon}} = \int_{\real^d} \lvert f-I \rvert_{\varepsilon}^{\chi_{\varepsilon}} \e^{-V}
            \xrightarrow[\varepsilon \to 0]{} \int_{\real^d} \lvert f-I \rvert \e^{-V}
    \]
    and
    \[
        \int_{\real^d} \left(\frac{\lvert f-I \rvert}{\lvert f-I \rvert_{\varepsilon}^{\chi_{\varepsilon}}}\right)^2  \lvert f-I \rvert_{\varepsilon}^{\chi_{\varepsilon}} \e^{-V}
        = \int_{\real^d} \frac{\lvert f-I \rvert^2}{\lvert f-I \rvert_{\varepsilon}^{\chi_{\varepsilon}}} \e^{-V}
        \xrightarrow[\varepsilon \to 0]{} \int_{\real^d} \abs{f-I} \e^{-V}.
    \]
    and so $s^2_f[U_{\varepsilon}] \to s^*_f$ in the limit as $\varepsilon \to 0$.
    Since $s^2_f[U] \geq s^*_f$ for all $U \in \smoothcompact(\real^d)$ by~\eqref{eq:lower_bound_asym_var_iid},
    we deduce the statement.
\end{proof}

% \begin{remark}
%     \label{remark:optimal_idd}
%     The proof of the third item in~\cref{proposition:background_infimum} is based on the construction of a minimizing sequence obtained by regularizing
%     the function~$U_*^{\rm iid}$ in~\eqref{eq:minimizer_lemma_iid}.
%     % the function
%     % \begin{equation}
%     %     \label{eq:function_we_regularize}
%     %     U_*^{\rm iid} := - \log \abs{f-I}.
%     % \end{equation}
%     % With the convention that~$\log(0) = -\infty$,
%     % the function $U_*^{\rm iid}$ is well-defined and belongs to~$\mathcal U_0$.
%     This function is not smooth,
%     because $f - I$ admits at least one root in $\domain^d$.
%     The choice of~$U_*^{\rm iid}$ as the function we regularize is natural in view of~\cref{lemma:preparatory_asymvar_iid},
%     and given that the inequality in~\eqref{eq:lower_bound_asym_var_iid} is an equality for~$U = U_*^{\rm iid}$.
%     % ~\eqref{eq:function_we_regularize}
% \end{remark}
\begin{remark}
    \label{remark:optimal_idd_not_minimizer}
    % In light of~\cref{remark:optimal_idd},
    First note that~$U_*^{\rm iid} \in \mathcal U_0$.
    In light of~\cref{lemma:preparatory_asymvar_iid},
    one may wonder whether $s^2_f[U_*^{\rm iid}] = s^*_f$.
    That is to say, is $U_*^{\rm iid}$ \emph{a minimizer} of~$s^2_f$ in~$\mathcal U_0$?
    Substitution into~\eqref{eq:asym_var_iid} reveals that this is not always the case:
    \begin{equation}
        \label{eq:not_minimum}
        s^2_f[U_*^{\rm iid}] = \frac{Z^2}{\mathscr Z[U_*^{\rm iid}]^2} \left(\int_{\domain^d} \abs{f-I} \d \mu\right)^2  \geq s^*_f.
    \end{equation}
    The inequality in this equation is an equality if and only if $\mathscr Z[U_*^{\rm iid}] = Z$ or,
    equivalently, $f^{-1}(I)$ has zero Lebesgue measure.
    In particular, if $f^{-1}(I)$ has positive Lebesgue measure,
    then the biasing potential~$U_*^{\rm iid}$ does not achieve the lower bound~$s^*_f$,
    even though a minimizing sequence that asymptotically achieves the lower bound~$s^*_f$ can be constructed by regularizing this potential.
\end{remark}

\Cref{example:iid_asym_var} in the appendix illustrates~\cref{proposition:background_infimum}.
We explicitly construct in this example a potential function $U \in \mathcal U \setminus \mathcal U_0$ such that $s^2_f[U] = 0$,
as well as a minimizing sequence $(U_{\varepsilon})_{\varepsilon > 0}$ in~$\mathcal U_0$ such that~$s^2_f[U_{\varepsilon}] \to s^*_f$ in the limit~$\varepsilon \to 0$.
The example also illustrates that~$U_*^{\rm iid}$ is not necessarily a minimizer in~$\mathcal U_0$,
as mentioned in~\cref{remark:optimal_idd_not_minimizer}.

This section reveals the difficulties encountered when no regularity of the biasing potential~$U$ is assumed.
Similar difficulties will be encountered in the analysis of the MCMC estimator~\eqref{eq:estimator}.

\section{Minimizing the asymptotic variance for a single observable}%
\label{sec:minimizing_the_asymptotic_variance_for_one_observable}

In this section, as in the previous one,
we consider a target-oriented approach:
we seek the optimal biasing potential~$U$ for a given observable~$f$.
After presenting the mathematical framework in~\cref{sub:math_setting},
we first obtain in~\cref{sub:asymvar} an expression for the asymptotic variance associated to the estimator~\eqref{eq:estimator}
in terms of the solution to a Poisson equation where~$f$ appears on the right-hand side,
and subsequently address the problem of finding the optimal biasing potential~$U$,
first in the one-dimensional setting in~\cref{sub:explicit_optimal_potential_in_dimension_one},
and then in the multi-dimensional setting in~\cref{sub:optimal_solution_in_the_multi_dimensional_setting}.

\subsection{Mathematical setting}
\label{sub:math_setting}
We will use the following functional spaces:
\[
    L^2_0(\mu_U) =
    \left\{ \varphi \in L^2(\mu_U) \, \middle| \, \int_{\domain^d} \varphi \, \d \mu_U = 0 \right\},
    \qquad
    H^1(\mu_U) = \left\{ \varphi \in L^2(\mu_U) \, \middle| \, \nabla \varphi \in  \bigl(L^2(\mu_U)\bigr)^d \right\},
\]
where~$\mu_U$ is the probability measure defined in~\eqref{eq:mu_u}.
In the i.i.d.\ setting, the condition $(f-I) \e^U \in L^2(\mu_U)$ is necessary and sufficient to guarantee that a central limit theorem holds.
There does not exist such a simple condition in the MCMC setting,
and so, in most of this section,
we work under conditions which are only sufficient.
Specifically, we denote by~$\mathfrak U_0 = \mathfrak U_0(V, f)$ the set of biasing potentials that satisfy the following assumptions.
\begin{assumption}
    \label{assumption:as1}
    $~$
    \begin{itemize}
        \item
            The function $U$ is smooth on $\domain^d$.

        \item
            The function $\e^{-V-U}$ is Lebesgue integrable over $\domain^d$.

        \item
            The probability measure $\mu_{U}$ satisfies a Poincaré inequality:
            there exists $R[U] > 0$ such that
            \begin{equation}
                \label{eq:poincaré}
                \forall \varphi \in H^1(\mu_{U}) \cap L^2_0(\mu_U), \qquad
                \norm{\varphi}_{L^2(\mu_U)}^2 \leq \frac{1}{R[U]}\norm{\nabla \varphi}_{L^2(\mu_U)}^2.
            \end{equation}

        \item
            It holds that $(f-I) \, \e^{U} \in L^2(\mu_{U})$.

        \item
            For any~$x_0\in\domain^d$, there exists a unique strong solution~$X_t$ to~\eqref{eq:mix} with~$X_0 = x_0$.
    \end{itemize}
\end{assumption}
In this section, the functions~$V$ and~$f$ are considered fixed data of the problem,
so the dependence of~$\mathfrak U_0$ on these data is omitted in the notation.
However, in~\cref{sec:minimizing_the_expected_asymptotic_variance} we will write~$\mathfrak U_0(V, f)$ to emphasize this dependence where necessary.
Just like the set~$\mathcal U_0$ in \cref{sub:iid},
the set~$\mathfrak U_0$ in this section contains ``nice'' biasing potentials.
Indeed, the asymptotic behavior of the estimator~\eqref{eq:estimator} can be rigorously characterized for~$U \in \mathfrak U_0$;
see~\cref{lemma:asymptotic_variance}.

\begin{remark}
    If~$U \in \mathfrak U_0$, then the estimator~\eqref{eq:estimator} converges to $I$ almost surely as $T \to \infty$ because,
    by ergodicity,
    \[
        \mu^T_U(f) = \frac{\displaystyle \frac{1}{T}\int_0^T (f \e^U)(X_t) \, \d t}{\displaystyle \frac{1}{T}\int_0^T(\e^U)(X_t) \, \d t}
        \xrightarrow[T \to \infty]{{\rm a.s.}}
        \frac{\displaystyle \int_{\domain^d} (f \e^U) \, \d \mu_U}{\displaystyle \int_{\domain^d} (\e^U) \, \d \mu_U}
        = I.
    \]
    % In general $C^{\infty}_{\rm c}(\real^d) \subset \mathfrak U_0$,
    % and this inclusion is a set equality when~$\domain = \torus$.
\end{remark}

\begin{remark}
In the case where $\domain = \torus$,
a Poincaré inequality of the form~\eqref{eq:poincaré} always holds provided that $V+U$ is smooth.
When $\domain = \real$, however,
the potential $V+U$ must satisfy appropriate growth conditions to ensure that the inequality holds.
For sufficient conditions, see~e.g.~\cite{MR1845806,MR2386063}.
We use the Poincaré inequality to establish the central limit theorem,
but there are other ways to obtain similar conclusions without directly using the Poincaré inequality,
for example by using results of Kipnis--Varadhan or Foster--Lyapunov (see~\cite{MR3483241} and references within).
\end{remark}

\begin{remark}
    When~$\domain = \torus$,
    the first item in~\cref{assumption:as1} implies all the other items,
    and in this setting~$\mathfrak U_0 = C^{\infty}(\torus^d) = C^{\infty}_{\rm c}(\torus^d)$.
\end{remark}

We denote by $\mathcal L_U$ the infinitesimal generator on~$L^2(\mu_U)$ of the Markov semigroup associated to~\eqref{eq:mix},
which is given on~$\smoothcompact(\domain^d)$ by
\begin{equation}
    \label{eq:generator}
    \mathcal L_U = - \nabla (V + U) \cdot \nabla + \laplacian = \e^{V+U} \nabla \cdot (\e^{-V-U} \nabla \dummy).
\end{equation}

\subsection{Asymptotic variance}
\label{sub:asymvar}
The following lemma gives an expression of the asymptotic variance of the estimator $\mu^T_U(f)$ given in~\eqref{eq:estimator}
in terms of the solution to a Poisson equation.

\begin{lemma}
    [Asymptotic variance]
    \label{lemma:asymptotic_variance}
    Suppose that~$U \in \mathfrak U_0$.
    Then there exists a unique distributional solution~$\phi_U \in H^1(\mu_{U}) \cap L^2_0(\mu_{U})$ to
    \begin{equation}
        \label{eq:poisson}
        -\mathcal L_U \phi_{U} = (f- I) \e^U.
    \end{equation}
    The solution $\phi_U$ is smooth and,
    for Lebesgue almost all initial condition,
    it holds that
    \[
        \sqrt{T} \bigl( \mu^T_U(f) - I\bigr)
        \xrightarrow[T \to \infty]{\rm Law} \normal\bigl(0, \sigma^2_f[U]\bigr),
    \]
    where
    \begin{equation}
        \label{eq:asym_var}
        \sigma^2_f[U] := \frac{2\Z{U}^2}{Z^2}\int_{\domain^d} \abs*{\nabla \phi_U}^2 \d\mu_{U}
        = \frac{2\Z{U}^2}{Z^2}\int_{\domain^d} \phi_U (f-I) \, \e^U \, \d\mu_{U}.
    \end{equation}
\end{lemma}
\begin{proof}
    % The beginning of the proof is standard,
    % but we include it for completeness.
    By density of $\smoothcompact(\domain^d)$ in~$H^1(\mu_U)$,
    a function~$\phi_U \in H^1(\mu_{U}) \cap L^2_0(\mu_{U})$ is a distributional solution to~\eqref{eq:poisson} if and only if
    % \[
    %     \forall \varphi \in \smoothcompact(\domain^d), \qquad
    %     \int_{\domain^d} \nabla \phi_U \cdot \nabla \varphi \, \d \mu
    %     = \int_{\domain^d} (f - I) \e^U \, \varphi \, \d \mu.
    % \]
    % this holds true if and only if
    % \begin{equation}
    %     \label{eq:weak_formulation}
    %     \forall \varphi \in H^1(\mu_U), \qquad
    %     \int_{\domain^d} \nabla \phi_U \cdot \nabla \varphi \, \d \mu
    %     = \int_{\domain^d} (f - I) \e^U \, \varphi \, \d \mu.
    % \end{equation}
    % Since the right-hand side of~\eqref{eq:poisson} is mean-zero with respect to~$\mu_U$ by~\cref{assumption:as1},
    % it is equivalent to restrict $\varphi \in H^1(\mu_U) \cap L^2_0(\mu_U)$ in~\eqref{eq:weak_formulation}.
    \begin{equation}
        \label{eq:weak_formulation}
        \forall \varphi \in H^1(\mu_U) \cap L^2_0(\mu_{U}), \qquad
        \int_{\domain^d} \nabla \phi_U \cdot \nabla \varphi \, \d \mu_U
        = \int_{\domain^d} (f - I) \e^U \, \varphi \, \d \mu_U.
    \end{equation}
    The validity of a Poincaré inequality for~$\mu_U$ implies that
    the function space $H^1(\mu_U) \cap L^2_0(\mu_U)$ endowed with the inner product
    \[
        (\varphi_1, \varphi_2) \mapsto
        \int_{\domain^d} \nabla \varphi_1 \cdot \nabla \varphi_2 \, \d \mu_U
    \]
    is a Hilbert space,
    and that the right-hand side of~\eqref{eq:weak_formulation} is a bounded linear functional on this space.
    Therefore, the Lax--Milgram theorem (or the Riesz representation theorem)
    yields the existence of a unique solution $\phi_U$ in $H^1(\mu_{U}) \cap L^2_0(\mu_{U})$.
    Elliptic regularity theory~\cite{MR1814364,evans2010partial}
    % or H\"ormander's theorem~\cite{MR222474}
    then implies that~$\phi_U\in C^{\infty}(\domain^d)$.
    From the definition~\eqref{eq:estimator} of $\mu^T_U(f)$,
    we have
    \[
        \sqrt{T} \bigl( \mu^T_U(f) - I\bigr)
        = \frac{\displaystyle \frac{1}{\sqrt{T}} \int_0^T \left((f-I) \e^U\right)(X_t) \, \d t}
        {\displaystyle \frac{1}{T}\int_0^T\left(\e^U\right)({X}_t) \, \d t}.
    \]
    The numerator converges in law to $\normal\bigl(0, 2\int_{\domain^d} \abs*{\nabla \phi_{U}}^2 d\mu_{U}\bigr)$,
    for instance by~\cite[Theorem~3.1]{MR3069369}
    (the setting there is~$\real^d$,
    but for~$\torus^d$, the argument using~\eqref{eq:basicconv} and the martingale central limit theorem,
    that is essentially~\cite[Theorem~VIII.3.11]{MR1943877},
    works in the same way),
    while the denominator converges almost surely to $Z/\Z{U}$.
    The claimed convergence in law then follows from Slutsky's lemma.
    The last equality in~\eqref{eq:asym_var} follows from the definition~\eqref{eq:weak_formulation} of a weak solution.
\end{proof}

\begin{remark}
    [Stability estimate]
    From the weak formulation~\eqref{eq:weak_formulation} and the Poincaré inequality~\eqref{eq:poincaré},
    we deduce the stability estimate
    \begin{equation}
        \label{eq:stability_estimate}
        \norm{\nabla \phi_U}_{L^2(\mu_U)}
        \leq \frac{1}{\sqrt{R[U]}} \norm*{(f - I) \e^U}_{L^2(\mu_U)}.
    \end{equation}
    This standard estimate  will be useful in the proof of~\cref{lemma:stability_gradient_sol_poisson} in the appendix
    (for a Poisson equation with a different right-hand side).
\end{remark}

\begin{remark}
It is instructive to write the counterpart of~\cref{lemma:asymptotic_variance}
for the estimator in the right-most column and second row of~\cref{table:importance_sampling_estimators},
which is based on evaluations of the solution to~\eqref{eq:mix} at discrete times:
\begin{equation}
    \label{eq:semidiscrete_estimator}
    \widetilde \mu_U^N :=
    \frac
    {\displaystyle \sum_{n=0}^{N-1} (f \e^U)(X_{n \tau})}
    {\displaystyle \sum_{n=0}^{N-1} (\e^U)(X_{n \tau})}.
\end{equation}
For this estimator,
we prove in~\cref{sec:discrete_asym_var} that,
under appropriate conditions including~\cref{assumption:as1},
\begin{equation}
    \label{eq:discrete_convergence}
    \sqrt{N} \bigl( \widetilde \mu^N_U(f) - I\bigr)
    \xrightarrow[N \to \infty]{\rm Law} \normal\left(0, \widetilde \sigma^2_f[U]\right),
\end{equation}
with now
\begin{equation}
    \label{eq:sigma_subsampled}
    \widetilde \sigma^2_f[U] = \frac{\Z{U}^2}{Z^2} \left( 2 \int_{\domain^d} \widetilde \phi_U (f-I) \, \e^U \, d\mu_{U} - \int_{\domain^d} \left\lvert (f-I) \, \e^U \right\rvert^2 \, \d \mu_U \right),
\end{equation}
where $\widetilde \phi_U$ is the unique solution in~$L^2_0(\mu_U)$ to
\begin{equation}
    \label{eq:discrete_poisson}
    - \widetilde {\mathcal L}_U \widetilde \phi_U = (f- I) \e^U,
     \qquad
     \widetilde {\mathcal L}_U := \e^{\tau \mathcal L_U} - \mathcal I.
\end{equation}
Here $\e^{t\mathcal L_U}$ denotes the Markov semigroup corresponding to the stochastic dynamics~\eqref{eq:mix}:
\[
    \left(\e^{t \mathcal L_U} \varphi\right) (x) = \expect \bigl(\varphi(X_t) \big| X_0 = x\bigr).
\]
The asymptotic variance~$\widetilde \sigma^2_f[U]$ converges to that for the i.i.d.\ setting
given in~\eqref{eq:asym_var_iid} in the limit as~$\tau \to \infty$,
and it diverges in the limit as~$\tau \to 0$.
The latter is not surprising as the correlation between successive samples increases in this limit.
However,
since formally $\tau \widetilde \phi_U \to \phi_U$ in $L^2_0(\mu_U)$ in the limit as $\tau \to 0$,
it holds that~%
\(
    \tau \widetilde \sigma^2_f[U] \xrightarrow[\tau \to 0]{} \sigma^2_f[U].
\)
\end{remark}

\subsection{Explicit optimal \texorpdfstring{$U$}{potential} in dimension one}
\label{sub:explicit_optimal_potential_in_dimension_one}
In the one-dimensional setting,
it is possible to write an explicit expression for the asymptotic variance~$\sigma^2_f[U]$,
from which an explicit lower bound on $\sigma^2_f[U]$ can be obtained.
Our strategy in this section is the following:
\begin{itemize}
    \item
        We first obtain an explicit expression for~$\sigma^2_f[U]$ for~$U \in \mathfrak U_0$ (\cref{lemma:asymvar_in_1d}),
        which we then rewrite in a different form~$\widehat \sigma^2_f[U]$ given in~\eqref{eq:explicit_sigma_1d}.

    \item
        We then observe that~$\widehat \sigma^2_f[U]$ is defined more generally for~$U \in \mathfrak U \supset \mathfrak U_0$,
        where~$\mathfrak U$ is an appropriate superset of~$\mathfrak U_0$,
        noting that~$\widehat \sigma^2_f[U]$ is not necessarily the asymptotic variance of~$\mu_U^T(f)$ for~$U \in \mathfrak U \setminus \mathfrak U_0$.

    \item
        Next, we show that~$\widehat \sigma^2_f$ admits an explicit minimizer~$U_*$ over~$\mathfrak U$,
        with associated minimum~$\sigma^*_f$.
        This is proved in~\cref{proposition:bound_asym_var_1d}.

    \item
        Finally,
        using the expression of~$U_*$,
        we prove that~$\sigma^*_f$ is the infimum of the actual asymptotic variance~$\sigma^2_f[U]$ over~$\mathfrak U_0$,
        and that this infimum can be approached within the class of smooth biasing potentials with compact support.
        This is the content of~\cref{proposition:sharpness}.
\end{itemize}

The main result of this section, \cref{proposition:sharpness},
and preceding auxiliary result, \cref{proposition:bound_asym_var_1d},
should be viewed as the counterparts in the MCMC setting of~\cref{proposition:background_infimum} and~\cref{lemma:preparatory_asymvar_iid} in the i.i.d.\ setting.
% In~\cref{proposition:sharpness},
% we shall obtain an explicit expression for the infimum of the asymptotic variance~$\sigma^2_f$ associated with the~MCMC estimator~\eqref{eq:estimator} over a class of~``nice'' potential functions,
% and then prove that this infimum can be approached arbitrarily close within the class of smooth biasing potentials with compact support.

\begin{lemma}
    [Explicit expression for the asymptotic variance in dimension 1]
    \label{lemma:asymvar_in_1d}
    For $U \in \mathfrak U_0$ and in dimension $d=1$,
    % the solution $\phi_U \in H^1(\mu_{U})\cap L_0^2(\mu_{U})$ to \eqref{eq:poisson} and
    the asymptotic variance \eqref{eq:asym_var} writes
    \begin{align}
        % \label{eq:phiprime}
        % \phi_{U}' &= -\bigl( F - A_{\domain}[U] \bigr)  \e^{V+U}, \\
        \sigma^2_f[U]
        \label{eq:asymvar_1d}
                  &= \frac{2\Z{U}^2}{Z^2} \int_{\domain} \left\lvert \bigl(F-A_{\domain}[U]\bigr) \, \e^{V+U} \right\rvert^2 \, \d \mu_{U},
    \end{align}
    where $F:\domain \rightarrow \real$ is given by
    \begin{equation*}
    F(x) = \int_0^x \bigl( f(\xi)-I \bigr) \e^{-V(\xi)}\d \xi,
    \end{equation*}
    and
    \begin{equation}
        \label{eq:A_domain}
        A_{\domain}[U] =
        \begin{cases}
            \displaystyle
            - \int_{-\infty}^{0} \bigl( f-I \bigr) \e^{-V} & \text{ if $\domain = \real$, } \\[.4cm]
            \displaystyle
            \frac{\displaystyle \int_{\torus} F \, \e^{V+U}}
            {\displaystyle \int_{\torus} \e^{V+U}} & \text{ if $\domain = \torus$. }
        \end{cases}
    \end{equation}
\end{lemma}

Note that $A_{\real}[U]$ is independent of~$U$;
we shall henceforth drop the dependence in the notation.
It seems from~\eqref{eq:A_domain} that $A_{\real}$ and $A_{\torus}[U]$ have very different expressions.
In fact, the constant~$A_{\real}$ may be obtained as a limit of~$A_{\torus}[U]$ for an increasingly large torus;
see~\cref{remark:constant_A}.
\begin{proof}
    In dimension one, the Poisson equation~\eqref{eq:poisson} reads
    \begin{equation}
        \label{eq:1d_equation}
        -\e^{V+U} \left(\e^{-(V+U)} \phi_{U}'\right)'  = (f-I) \e^U.
    \end{equation}
    By integration of~\eqref{eq:1d_equation},
    it holds that
    \begin{equation}
        \label{eq:derivative_of_phi}
        \phi_{U}'(x) = -\left(\int_0^x \bigl( f(\xi)-I \bigr) \e^{-V(\xi)}\d \xi - A\right) \e^{(V+U)(x)}
        = -\bigl( F(x) - A \bigr)  \e^{(V+U)(x)},
    \end{equation}
    and so
    \begin{equation}
        \label{eq:phi_general_form}
        \phi_{U}(x) = B - \int_{0}^{x} \bigl( F(\xi) - A \bigr)  \e^{(V+U)(\xi)} \, \d \xi,
    \end{equation}
    for some constants $A \in \real$ and $B \in \real$.
    The requirement that $\phi_{U} \in H^1(\mu_{U})$ enables to determine the constant $A$:
    \begin{itemize}
        \item
            When $\domain = \torus$, the embedding~$H^1(\mu_{U}) \subset C(\torus)$ gives that
            $\phi_{U}(-\pi) = \phi_{U}(\pi)$, which leads to the equation for $A_{\torus}[U]$ in~\eqref{eq:A_domain}.

        \item
            When $\domain = \real$,
            the requirement that $\phi_{U} \in H^1(\mu_{U})$ implies that
            \begin{equation}
                \label{eq:constant_A_real}
                A
                = \lim_{x \to \infty} F(x),
            \end{equation}
            where
            \[
                \lim_{x \to \infty} F(x)
                = \int_{0}^{\infty} \bigl( f(x)-I \bigr) \e^{-V(x)} \, \d x
                = - \int_{-\infty}^{0} \bigl( f(x)-I \bigr) \e^{-V(x)} \, \d x,
            \]
            because otherwise~$\phi_{U}'$ in~\eqref{eq:derivative_of_phi} is not in $L^2(\mu_{U})$.
            Indeed, assume for contradiction that
            \[
                \lim_{x \to \infty} F(x)
                =: L \neq A.
            \]
            (The limit exists because $(f-I)\e^{-V} \in L^1(\real)$ by~\cref{assumption:assumptions_V}.)
            Then there is $K \in \real$ such that
            \[
                \inf_{x \geq K} \abs{F(x) - A} \geq \frac{1}{2} \abs{L - A},
            \]
            and so by~\eqref{eq:derivative_of_phi}, it holds that
            \[
                \int_{\real} \lvert \phi_U' \rvert^2 \, \e^{-(V+U)} \, \d x
                \geq \frac{1}{4} \abs{L - A}^2 \int_{K}^{\infty} \e^{V+U} \, \d x.
            \]
            The right-hand side of this equation is infinite because,
            by the Cauchy--Schwarz inequality,
            \[
                + \infty = \int_{K}^{\infty} \sqrt{\e^{V+U}} \sqrt{\e^{-V-U}} \, \d x
                \leq \int_{K}^{\infty} \e^{V+U} \, \d x \int_{K}^{\infty} \e^{-V-U} \, \d x.
            \]
    \end{itemize}
    Equation~\eqref{eq:asymvar_1d} is then obtained by substitution of~\eqref{eq:derivative_of_phi} in~\eqref{eq:asym_var}.
\end{proof}

\begin{remark}
    Once $A$ in~\eqref{eq:phi_general_form} has been determined,
    the value of $B$ can be obtained from the condition that $\phi_{U} \in L^2_0(\mu_{U})$ is mean-zero.
    This constant is not required for our purposes,
    because it cancels out in the formula~\eqref{eq:asymvar_1d} for the asymptotic variance,
    and so its explicit expression is omitted.
\end{remark}

\iffalse
\begin{remark}
    \Cref{lemma:asymvar_in_1d} proves that,
    although choosing~$U$ in such a manner that the free energy associated with a reaction coordinate is constant
    -- an approach known in the literature as \emph{free energy biasing}~--~alleviates metastability~\cite{MR2681239},
    this strategy is in general not optimal in terms of asymptotic variance for a specific observable.
\end{remark}
\fi

For all~$U \in \mathfrak U_0$, the right-hand side of~\eqref{eq:asymvar_1d} coincides,
both when $\domain = \real$ and when $\domain = \torus$,
with
\begin{equation}
    \label{eq:explicit_sigma_1d}
    \widehat \sigma^2_f[U] :=
    \frac{2\Z{U}}{Z^2} \inf_{A \in \real} \int_{\domain} \abs{F - A}^2 \e^{V+U} .
\end{equation}
\begin{lemma}
    \label{remark:infimum_in_hat_sigma}
    For all~$U \in \mathfrak U_0$,
    it holds that~$\sigma^2_f[U] = \widehat \sigma^2_f[U]$.
\end{lemma}
\begin{proof}
    When~$\domain = \real$,
    the infimum in~\eqref{eq:explicit_sigma_1d} is achieved for~$A =  A_{\real}$,
    because the integral is infinite for any other value of~$A$.
    Likewise, when~$\domain = \torus$
    % and~$F \e^{U+V} \in L^1(\torus)$,
    the infimum in~\eqref{eq:explicit_sigma_1d} is achieved for~$A =  A_{\torus}[U]$,
    because the mean under the probability measure proportional to~$\e^{V+U}$ is the approximation by a constant in the~$L^2(\e^{V+U})$ norm;
    see, for instance, \cite[Exercise 1.4.23]{MR2760872}.
\end{proof}

With the convention that $0 \cdot + \infty = + \infty \cdot 0 = 0$,
the right-hand side of~\eqref{eq:explicit_sigma_1d} makes sense as an element of~$\real\cup\{\infty\}$ for all~$U \in \mathfrak U$,
where $\mathfrak U = \mathfrak U(V, f)$ is the set of measurable functions~$U\colon \domain \to \real \cup \{ \infty \}$ such that
the following assumptions are satisfied:
\begin{itemize}
    \item
        It holds that $\e^{-V - U} \in L^1(\domain^d)$ and $Z[U] > 0$,
        so that $\mu_U$ is well defined.

    \item
        The condition~\eqref{eq:unbiased} is satisfied.
\end{itemize}
We emphasize that~$\mathfrak U$ is a proper superset of~$\mathfrak U_0$;
it contains elements which violate~\cref{assumption:as1}.
For biasing potentials not in~$\mathfrak U_0$,
the quantity~\eqref{eq:explicit_sigma_1d} is not in general an asymptotic variance.
In particular,
it is possible to construct examples where~$\sigma^2_f[U]$ is zero even though $\widehat \sigma^2_f[U] > 0$,
as we illustrate in~\cref{example:not_asymvar}.
\begin{example}
    \label{example:not_asymvar}
    Consider the setting where $V\colon \torus \to \real$ is zero and~$f\colon [-\pi,\pi] \to \real$ is given by
    \begin{equation}
        \label{eq:not_asymvar}
        f(x) =
        \begin{cases}
            \sgn(x) \qquad &\text{if $\abs{x} \geq \frac{\pi}{2}$}, \\
            0 \qquad &\text{otherwise}.
        \end{cases},
        \qquad
        \qquad \sgn(x) :=
        \begin{cases}
            1 & \text{ if $x > 0$,} \\
            0 & \text{ if $x = 0$,} \\
            -1 & \text{ if $x < 0$.} \\
        \end{cases}
    \end{equation}
    Here we identify~$[-\pi, \pi]$ with its image under the quotient map~$\real \to \torus$.
    If~$U$ is a potential such that (i)~there exists a unique strong solution to~\eqref{eq:mix} with initial condition~$X_0 = 0$
    and (ii) this solution satisfies~$X_t \in (-\pi/2, \pi/2)$ with probability 1 for all times,
    then~\eqref{eq:estimator} is a well defined estimator with zero asymptotic variance.
    However $\widehat \sigma^2_f[U] > 0$,
    which can be viewed from~\eqref{eq:explicit_sigma_1d} and is confirmed in~\cref{proposition:bound_asym_var_1d} hereafter.
\end{example}

Although $\widehat \sigma^2_f$ does not in general correspond to an asymptotic variance,
obtaining a bound from below on~$\widehat \sigma^2_f$ over~$\mathfrak U$ will be useful in order to motivate the proof of~\cref{proposition:sharpness},
just like~\cref{lemma:preparatory_asymvar_iid} proved useful for establishing~\cref{proposition:background_infimum} in the i.i.d.\ setting.

\begin{lemma}
    \label{proposition:bound_asym_var_1d}
    It holds that
    \begin{equation}
        \label{eq:1dsig}
        \min_{U \in \mathfrak U} \widehat \sigma^2_f[U] =
        \frac{2}{Z^2} \bigg(\int_{\domain} \bigl\lvert F(x) - A^*_{\domain} \bigr\rvert \d x \bigg)^2 =: \sigma^*_f,
    \end{equation}
    with $A^*_{\real} := A_{\real}$ and
    \begin{equation}
        \label{eq:equation_A}
        A^*_{\torus} := \sup \left\{ A \in \real : \int_{\torus} \sgn (F-A) \geq 0 \right\}.
        % \qquad \sgn(x) :=
        % \begin{cases}
        %     1 & \text{ if $x > 0$,} \\
        %     0 & \text{ if $x = 0$,} \\
        %     -1 & \text{ if $x < 0$.} \\
        % \end{cases}
    \end{equation}
    In addition, the infimum is achieved for
    \begin{equation}
        \label{eq:optu}
        U = U_*(x) := - V(x) -\log\abs*{F(x) - A^*_{\domain}} \in \mathfrak U.
    \end{equation}
    In this case
    \(
        \e^{-V-U_*} \propto \abs*{F(x) - A^*_{\domain}}.
    \)
\end{lemma}
\begin{proof}
    It is sufficient to show that, for all $A \in \real$,
    \begin{equation}\label{eq:foralla}
        \frac{2\Z{U}}{Z^2} \int_{\domain} \abs{F - A}^2 \e^{V+U}
        \geq \frac{2}{Z^2} \bigg(\int_{\domain} \bigl\lvert F(x) - A^*_{\domain} \bigr\rvert \, \d x \bigg)^2.
    \end{equation}
    If the left-hand side of~\eqref{eq:foralla} is infinite,
    then the inequality is trivially satisfied.
    % on which~$\abs{F-A}^2 \e^{V+U} = \infty$ has positive Lebesgue measure then the left-hand side of~\eqref{eq:foralla} is infinite,
    % so~\eqref{eq:foralla} holds.
    On the other hand,
    if the left-hand side is finite,
    in which case the set on which~$\abs{F-A}^2 \e^{V+U} = \infty$ is of measure zero,
    then we have by the Cauchy--Schwarz inequality that
    \begin{align*}
        \frac{2\Z{U}}{Z^2} \int_{\domain} \abs{F- A}^2 \e^{V+U}
        &= \frac{2}{Z^2} \int_{\domain} \e^{-V-U} \int_{\domain} \abs{F-A}^2 \e^{V+U}
        \geq \frac{2}{Z^2} \left( \int_{\domain} \abs{F - A}  \right)^2.
    \end{align*}
    In the case $\domain = \real$,
    the right-hand side is finite only if $A = A_{\real}$,
    which leads to~\eqref{eq:1dsig}.
    In the case~$\domain = \torus$,
    the inequality~\eqref{eq:1dsig} is obtained by noting that
    \[
        \int_{\domain} \bigl\lvert F - A \bigr\rvert
        = \expect \bigl\lvert F(X) - A \bigr\rvert,
    \]
    where $X \sim \mathcal U(\torus)$ is a random variable uniformly distributed on the torus.
    It is well known,
    see for example~\cite[Exercise 1.4.23]{MR2760872},
    that the expectation on the right-hand side is minimized for any $A$ that is a median of~$F(X)$.
    Here~$F$ is continuous, so the median of $F(X)$ is unique and given by~$A^*_{\torus}$,
    which implies that~\eqref{eq:1dsig} holds.

    The fact that the lower bound is achieved for~$U$ in~\eqref{eq:optu}
    follows from the inequality
    \begin{align*}
        \widehat \sigma^2_f[U_*]
        &= \frac{2\Z{U_*}}{Z^2} \inf_{A \in \real} \int_{\domain} \left\lvert F-A  \right\rvert^2  \e^{V+U_*}  \\
        % &= \frac{2\Z{U_*}}{Z^2} \inf_{A \in \real} \int_{\domain} \left\lvert F - A \right\rvert^2  \e^{V+U_*} \, \mathds 1_{\{\e^{-V-U_*} > 0\}}
        % \leq \frac{2\Z{U_*}}{Z^2} \inf_{A \in \real} \int_{\domain} \left\lvert F - A \right\rvert^2  \e^{V+U_*} \\
          &\leq \frac{2\Z{U_*}}{Z^2} \int_{\domain} \left\lvert F - A_\domain^* \right\rvert^2 \e^{V+U_*}
          = \frac{2}{Z^2} \left( \int_{\domain} \left\lvert F - A_\domain^* \right\rvert \right)^2,
    \end{align*}
    which concludes the proof.
\end{proof}

\begin{remark}
    \label{remark:domain_divided}
    The singularities in the biasing potential~$U_*$ coincide with zeros of the function $F(x) - A^*_{\domain}$.
    Consider for simplicity the case where $\domain = \real$.
    If~$x_*$ denotes a zero of the function $F(x) - A_{\real}^*$,
    then it holds by definition of~$F(x)$ that
    \[
        0 = F(x_*) - A^*_{\real}
        = \int_{-\infty}^{x_*} \bigl(f(x) - I\bigr) \e^{-V(x)} \, \d x.
    \]
    Rearranging this equation,
    we obtain
    \[
        \frac
        {\displaystyle \int_{-\infty}^{x_*} f(x) \, \e^{-V(x)} \, \d x}
        {\displaystyle \int_{-\infty}^{x_*}  \e^{-V(x)} \, \d x} = I.
    \]
    In other words,
    the average of $f$ with respect to the measure~$\mu$ restricted to~$[-\infty,x_*]$
    coincides with its average with respect to~$\mu$ over the real line.
    When singular, the biasing potential~\eqref{eq:optu} effectively divides the domain into regions that suffice for the estimation of~$I$.
    Several numerical experiments illustrating this behavior are presented in~\cref{sec:examples_and_numerical_experiments}.
\end{remark}

\begin{remark}
    Equation~\eqref{eq:equation_A} implies that~$A_{\torus}^*$ is the median associated with $F(X)$,
    where $X$ is a random variable with uniform distribution over~$\torus$.
    Just as $A_{\real}$ is obtained as a limit of~$A_{\torus}$ for an increasingly large torus,
    so too $A_{\real}^* = A_{\real}$ is recovered as a limit of~$A^*_{\torus}$;
    see~\cref{remark:constant_A}.
\end{remark}

The potential~$U_*$ defined by~\eqref{eq:optu} does not necessarily satisfy~\cref{assumption:as1},
and the measure $\mu_{U_*}$ may not have full support.
However, regularizing~$U_*$ enables to show that~$\sigma^*_f$ is the \emph{infimum} of the asymptotic variance~$\sigma^2_f[U]$,
not only over~$\mathfrak U_0$,
but also over the smaller subset $C^{\infty}_{\rm c} \subset \mathfrak U_0$ of smooth and compactly supported biasing potentials.
This is the content of the following result.
\Cref{table:comparison} after the proof summarizes the main results obtained in this section
and presents a comparison with the i.i.d.\ setting.
\begin{proposition}
    \label{proposition:sharpness}
    Suppose that $0 \in \mathfrak U_0$.
    Then,
    \begin{itemize}
        \item
            % The quantity $\sigma^*_f$ is the infimum of the asymptotic variance over the set~$\mathfrak U_0 \subset \mathfrak U$:
            It holds that
            \(
                \inf_{U \in \mathfrak U_0} \sigma^2_f[U] = \sigma^*_f.
            \)

        \item
            % The quantity $\sigma^*_f$ is the \emph{infimum} of the asymptotic variance $\sigma^2_f[U]$
            % over the set of smooth biasing potentials $U$ with compact support:
            It holds that
            \(
                \inf_{U \in C^{\infty}_{\rm c}(\domain)} \sigma^2_f[U] = \sigma^*_f.
            \)
    \end{itemize}
\end{proposition}

\begin{proof}
    Since~$\sigma^2_f[U] = \widehat \sigma^2_f[U]$ for~$U \in \mathfrak U_0$ and $C^{\infty}_{\rm c}(\domain^d) \subset \mathfrak U_0$,
    \cref{proposition:bound_asym_var_1d} and the second statement of~\cref{proposition:sharpness} imply the first statement.
    In order to prove the second statement,
    we use the same notation in this proof as in~\cref{proposition:background_infimum}.
    Let~$U_{\varepsilon}\colon \domain \to \real$ be the smooth biasing potential given by
    \[
        U_{\varepsilon}(x) = \begin{cases}
            - \chi_{\varepsilon}(x) \Bigl(V(x) + \log \bigl\lvert F(x) - A^*_{\real} \bigr\rvert_{\varepsilon}\Bigr) & \textrm{if } \domain = \real,\\
            - V(x) - \log \bigl\lvert F(x) - A^*_{\torus} \bigr\rvert_{\varepsilon} & \textrm{if } \domain = \torus,
        \end{cases}
    \]
    where for $\varepsilon \in (0, 1)$, $\chi_{\varepsilon} = \varrho \star \mathds 1_{[-\ell_{\varepsilon}, \ell_{\varepsilon}]}$ with $\ell_{\varepsilon} = \varepsilon^{-1/2} - 1$.
    (This choice of~$\ell_{\varepsilon}$ enables to write the bound in~\eqref{eq:motivation_choice_ell} below.)
    The probability distribution $\mu_{U_{\varepsilon}}$,
    with density proportional to~$\e^{-V-U_{\varepsilon}}$,
    satisfies a Poincaré inequality.
    This follows from the fact that $\e^{-V-U_{\varepsilon}}$ is uniformly bounded from below when~$\domain = \torus$,
    and from the classical Holley--Stroock biasing argument when~$\domain = \real$;
    see~\cite{MR893137}, as for example reviewed in~\cite[Theorem~2.11]{MR3509213}.
    Therefore~$U_{\varepsilon} \in \mathfrak U_0$ and so,
    by \cref{lemma:asymvar_in_1d},
    the associated asymptotic variance is given by
    \begin{align}
        \sigma^2_f[U_{\varepsilon}]
        &= \frac{2\Z{U_{\varepsilon}}}{Z^2} \int_{\domain} \bigl(F-A_{\varepsilon}\bigr)^2 \, \e^{V+U_{\varepsilon}},
    \label{eq:epsig0}
    \end{align}
    where $A_{\varepsilon} := A_{\domain}[U_{\varepsilon}]$ is given by~\eqref{eq:A_domain}
    with~$U = U_{\varepsilon}$.
    We now prove, separately for~$\domain = \torus$ and~$\domain = \real$,
    that
    \[
        \lim_{\varepsilon \to 0} \Z{U_{\varepsilon}} \to \int_{\domain} \abs{F - A^*_{\domain}}
        \qquad \text{ and } \qquad
        \limsup_{\varepsilon \to 0} \int_{\domain} \bigl(F-A_{\varepsilon}\bigr)^2 \, \e^{V+U_{\varepsilon}}
        \leq \int_{\domain} \bigl\lvert F(x) - A^*_{\domain} \bigr\rvert \, \d x.
    \]
    Given these results,
    taking the limit superior as $\varepsilon \to 0$ in~\eqref{eq:epsig0} gives
    \[
        \limsup_{\varepsilon \to 0}\sigma^2_f[U_{\varepsilon}]
        \leq \frac{2}{Z^2} \left(\int_{\domain} \bigl\lvert F(x) - A^*_{\domain} \bigr\rvert \, \d x\right)^2.
    \]
    Since the right-hand side is a lower bound on~$\sigma^2_f[U_{\varepsilon}] = \widehat \sigma^2_f[U_{\varepsilon}]$ by~\eqref{eq:1dsig},
    the result will be proved.

    \paragraph{Case $\domain = \torus$.}
    In this setting,
    it holds by dominated convergence,
    with an argument similar to the one used to prove the third item of~\cref{proposition:background_infimum},
    that
    \[
        \Z{U_{\varepsilon}}
        = \int_{\torus} \bigl\lvert F - A^*_{\torus} \bigr\rvert_{\varepsilon}
        \xrightarrow[\varepsilon \to 0]{} \int_{\torus} \bigl\lvert F - A^*_{\torus} \bigr\rvert.
    \]
    In addition, since~$A_{\varepsilon}$ is the average of $F$ over~$\torus$ with respect to the probability measure with density proportional to~$\e^{V + U_{\varepsilon}}$,
    it holds for all $\varepsilon \in (0, 1)$ that
    \begin{align*}
        \int_{\torus} \bigl\lvert F-A_{\varepsilon}\bigr \rvert^2 \, \e^{V+U_{\varepsilon}}
    &= \inf_{C \in \real} \int_{\torus} \bigl\lvert F - C\bigr \rvert^2 \, \e^{V+U_{\varepsilon}} \\
    &\leq \int_{\torus} \bigl\lvert F - A_{\torus}^* \bigr \rvert^2 \, \e^{V+U_{\varepsilon}}
    = \int_{\torus} \frac{\bigl\lvert F - A_{\torus}^* \bigr \rvert^2}{\bigl\lvert F - A_{\torus}^* \bigr \rvert_{\varepsilon}} \leq \int_{\torus} \bigl\lvert F - A_{\torus}^* \bigr \rvert,
    \end{align*}
    which enables to conclude.

    \paragraph{Case $\domain = \real$.}
    % In the following the~$\domain = \real$ case is dealt with;
    % the~$\domain = \torus$ case involves similar but simpler calculations.

    The numerator of the fraction on the right-hand side of~\eqref{eq:epsig0} can be written as
    \begin{align}
        \Z{U_{\varepsilon }}
    &= \int_{\real} \exp\Big( -V(x) + \chi_{\varepsilon}(x) \left(V(x) + \log\bigl\lvert F(x) - A^*_{\real} \bigr\rvert_{\varepsilon}\right) \Big) \, \d x \nonumber\\
    &= \int_{\real} \exp \bigl( (\chi_{\varepsilon} - 1 )V \bigr)\abs{F - A^*_{\real}}_{\varepsilon}^{\chi_{\varepsilon}}. \label{eq:integrand}
    \end{align}
    By convexity of the exponential function,
    it holds that
    \begin{align*}
        \nonumber
        \exp \bigl( (1 - \chi_{\varepsilon}) (-V) + \chi_{\varepsilon} \log \abs{F - A^*_{\real}}_{\varepsilon} \bigr)
    &\leq (1 - \chi_{\varepsilon}) \e^{-V} + \chi_{\varepsilon} \abs{F - A^*_{\real}}_{\varepsilon} \\
    \nonumber
    &\leq \e^{-V} + \chi_{\varepsilon} \bigl(\abs{F - A^*_{\real}} + \varepsilon\bigr) \\
    \label{eq:convexity_exponential}
    &\leq \e^{-V} + \abs{F - A^*_{\real}} + \varepsilon \chi_{\varepsilon}.
    \end{align*}
    Since $\varepsilon = (\ell_{\varepsilon} + 1)^{-2}$,
    all three terms on the right-hand side are dominated by an integrable function over~$\real$ independent of~$\varepsilon$,
    because
    \begin{equation}
        \label{eq:motivation_choice_ell}
        \forall x \in  \real, \qquad
        \varepsilon \chi_{\varepsilon}(x)
        \leq \frac{\mathds 1_{[-\ell_{\varepsilon}- 1, \ell_{\varepsilon}+ 1]}(x)}{(\ell_{\varepsilon}+ 1)^2}
        \leq \min \left\{1, \frac{1}{x^2} \right\}.
    \end{equation}
    Therefore,
    by dominated convergence,
    we deduce from~\eqref{eq:integrand} that $\Z{U_{\varepsilon}} \to \int_{\real} \abs{F - A^*_{\real}}$ in the limit as~$\varepsilon = (\ell_{\varepsilon} + 1)^{-2} \to 0$.
    For the integral on the right-hand side of~\eqref{eq:epsig0},
    noting that $A_{\varepsilon} = A^*_{\real}$ and recalling that this constant is independent of~$U$,
    we have that
    \begin{align}
        \notag
        \int_{\real} ( F - A_{\varepsilon} )^2 \e^{V+U_{\varepsilon}}
    &= \int_{\real} \abs{F - A^*_{\real}}^{2} \exp\Bigl((1 - \chi_{\varepsilon}) V + \chi_\varepsilon \bigl(- \log \abs{F - A^*_{\real}}_{\varepsilon}\bigr) \Bigr)  \\
    \notag
    &\leq \int_{\real} \abs{F - A^*_{\real}}^{2} \left( (1 - \chi_{\varepsilon}) \e^V + \frac{\chi_\varepsilon}{\abs{F - A^*_{\real}}_{\varepsilon}} \right) \\
    \label{eq:rhs_infimum}
    &\leq \int_{\real \setminus B_{\ell_{\varepsilon}-1}} \abs{F - A^*_{\real}}^{2} \e^V
    + \int_{\real} \abs{F  -  A^*_{\real}},
    \end{align}
    where we used the convexity of the exponential function and the notation $B_{\ell_{\varepsilon}-1} = [-\ell_{\varepsilon}+1, \ell_{\varepsilon}-1]$.
    Since~$\sigma_f^2[0]$ is finite, so is~$\int_{\real} \abs{F - A^*_{\real}}^2 \e^V$,
    implying that the first term on the right-hand side of~\eqref{eq:rhs_infimum} converges to 0 in the limit $\varepsilon \to 0$.
\end{proof}

\begin{table}[ht]
    \centering
    \begin{tabular}{|c|c|c|c|}
         \hline
         \phantom{$\Big($}
         & Minimal assumptions & ``Nice'' biasing potentials & $C^{\infty}_{\rm c}$
         \\ \hline
         i.i.d.\ setting
         (any $d$)
         & \phantom{$\frac{\Bigg(}{\Bigg(}$}\hspace{-4mm}
            \(
                \displaystyle
                \begin{aligned}
                    &\inf_{U \in \mathcal U} s^2_f[U] = 0
                    % \\[.2cm]
                    % &\min_{U \in \mathcal U} \widehat s^2_f[U] = s^*_f
                \end{aligned}
            \)
         &
            \(
                \displaystyle
                \inf_{U \in \mathcal U_0} s^2_f[U] = s^*_f
            \)
        &
            \(
                \displaystyle
                \inf_{U \in C^{\infty}_{\rm c}(\domain^d)} s^2_f[U] = s^*_f
            \)
         \\ \hline
         MCMC setting
         ($d=1$)
         &  \phantom{$\frac{\Bigg(}{\Bigg(}$}\hspace{-4mm}
            \(
                \displaystyle
                \begin{aligned}
                    &\inf_{U \in \mathfrak U} \sigma^2_f[U] \leq \sigma^*_f
                    % \\[.2cm]
                    % &\inf_{U \in \mathfrak U} \widehat \sigma^2_f[U] = \sigma^*_f
                \end{aligned}
            \)
         &
            \(
                \displaystyle
                \inf_{U \in \mathfrak U_0} \sigma^2_f[U] = \sigma^*_f
            \)
        &
            \(
                \displaystyle
                \inf_{U \in C^{\infty}_{\rm c}(\domain)} \sigma^2_f[U] = \sigma^*_f
            \)
         \\ \hline
    \end{tabular}
    \caption{%
        Summary of the main results obtained in~\cref{sub:explicit_optimal_potential_in_dimension_one},
        and comparison with the corresponding results in the i.i.d.\ setting obtained in~\cref{sub:iid}.
        See~\eqref{eq:def_infimum_iid} and~\eqref{eq:1dsig} respectively for the definitions of~$s_f^*$ and~$\sigma_f^{*}$.
        Concerning the infimum of the asymptotic variance under minimal assumptions in the MCMC setting and $d = 1$,
        we showed in~\cref{example:not_asymvar} that
        there are cases where $\sigma_f^* > 0$ but the minimum over~$\mathfrak U$ of the asymptotic variance for the~MCMC estimator is 0.
        It may be possible to prove,
        by extending the reasoning in the proof of~\cref{proposition:background_infimum},
        that the infimum of the asymptotic variance of the MCMC estimator~\eqref{eq:estimator} is~0 for any potential~$V$ and observable~$f$ satisfying~\cref{assumption:assumptions_V},
        but we do not address this question here.
        % Let us emphasize that,
        % in the MCMC setting,
        % our results do not provide information on the infimum of~$\sigma^2_f$ outside of~$\mathfrak U_0$,
        % but only an upper bound.
    }
    \label{table:comparison}
\end{table}

\begin{remark}
    The results in~\cref{proposition:sharpness} parallel the second and third items in~\cref{proposition:background_infimum}.
    We do not aim at rigorously establishing an analogue of the first item in~\cref{proposition:background_infimum},
    which would require analyzing the well-posedness of~\eqref{eq:mix} and the properties of the estimator~\eqref{eq:estimator} when the biasing potential is irregular and unbounded.
\end{remark}

\begin{remark}
    Since $U_* \in \mathfrak U$ defined in~\eqref{eq:optu} is a minimizer of~$\widehat \sigma^2_f$,
    and since the lower bound~$\sigma^*_f$ on $\sigma^2_f[U]$ may be approached by regularizing this potential,
    we often refer to~$U_*$ as \emph{the optimal biasing potential}.
    This is, of course, a slight abuse of terminology given that~$U_*$ is not in general a minimizer of the actual asymptotic variance~$\sigma^2_f[U]$,
    neither on~$\mathfrak U$ (because an asymptotic variance smaller than~$\sigma^*_f$ can sometimes be achieved) nor on~$\mathfrak U_0$ (because it does not hold that $U_* \in \mathfrak U_0$ in general).
\end{remark}

\begin{remark}
    To conclude this section, we note that the parallel between the i.i.d.\ and MCMC settings is not perfect:
    while $\sigma^2_f[U]$ and $\widehat \sigma^2_f[U]$ coincide for~$U \in \mathfrak U_0$ in the MCMC setting,
    $s^2_f[U]$ and $\widehat s^2_f[U]$ do not generally coincide for~$U \in \mathcal U_0$ in the i.i.d. setting.
\end{remark}

\subsection{Optimal \texorpdfstring{$U$}{potential} in the multi-dimensional setting via steepest descent}%
\label{sub:optimal_solution_in_the_multi_dimensional_setting}

In the multi-dimensional setting,
obtaining an explicit expression for the optimal biasing~$U$ is not possible.
However, analogously to~\cite{optfric}, the functional derivative of the asymptotic variance with respect to the biasing potential can be expressed
in terms of the solution to a Poisson equation.
This enables a numerical strategy based on a steepest descent for finding a good biasing potential.

\subsubsection*{Computation of the functional derivative}%
\label{ssub:Computation of the fucntional derivative}

In the following,
the directional derivative of a functional~$E: C^{\infty}(\domain^d) \rightarrow \real$ at~$U\in C^{\infty}(\domain^d)$
in the direction~$\delta U \in \smoothcompact(\domain^d)$ is denoted by
\begin{equation}\label{eq:funcd}
    \d E[U] \cdot \delta U = \lim_{\varepsilon \rightarrow 0} \frac{1}{\varepsilon} \bigl(E[U+\varepsilon \delta U] - E[U]\bigr),
\end{equation}
whenever the limit exists.
\begin{theorem}
    [Functional derivative of the asymptotic variance]
    \label{proposition:functional_derivative_asym_var}
    Suppose that~$U \in \mathfrak U_0$ and let $\phi_U$ be as in~\cref{lemma:asymptotic_variance}.
    % The functional derivative with respect to $U$ of $\sigma^2_f[U]$ is given by
    Then for all $\delta U \in \smoothcompact(\domain^d)$,
    it holds that
    \begin{align}
        \notag
        \frac{1}{2} \d \sigma^2_f[U] \cdot \delta U
        &:= \frac{1}{2}\lim_{\varepsilon \to 0} \frac{1}{\varepsilon} \bigl(\sigma^2_f[U + \varepsilon \delta U] - \sigma^2_f[U]\bigr) \\
        \label{eq:funcder}
        &= \frac{\Z{U}^2}{Z^2} \int_{\domain^d} \delta U \bigg( \abs*{\nabla{\phi_{U}}}^2 - \int_{\domain^d} \abs*{\nabla {\phi_{U}}}^2 \, \d \mu_{U} \bigg) \, \d \mu_{U}.
    \end{align}
\end{theorem}
\begin{proof}
    We first rewrite~\eqref{eq:asym_var} as
    \[
        \sigma^2_f[U]
        = \frac{2\Z{U}^2}{Z^2}\int_{\domain^d} \phi_{U} (f-I) \e^U \, \d \mu_{U}
        = \frac{2\Z{U}}{Z^2}\int_{\domain^d} \phi_{U} (f-I)\e^{-V},
    \]
    so that the only factors depending on $U$ are $\Z{U}$ and $\phi_{U}$.
    By definition of the functional derivative,
    using~\eqref{eq:weak_formulation} for the first integral term on the right-hand side,
    we have
    \begin{align}
        \label{eq:def_func_der}
        \frac{1}{2} \d\sigma^2_f[U] \cdot \delta U
        &= \frac{\d\Z{U} \cdot \delta U}{Z^2}\int_{\domain^d} \abs*{\nabla {\phi_{U}}}^2 \e^{-V-U}
        + \lim_{\varepsilon\rightarrow 0}\frac{\Z{U}}{\varepsilon Z^2} \int_{\domain^d} (\phi_{U + \varepsilon \delta U} - \phi_{U}) (f-I) \e^{-V},
    \end{align}
    where $\phi_{U + \varepsilon \delta U}$ is the solution to the perturbed Poisson equation
    \begin{equation}
        \label{eq:perturbed_poisson}
        - \mathcal L_{U + \varepsilon \delta U} \phi_{U + \varepsilon \delta U}
        % = - \e^{V + U + \varepsilon \delta U} \nabla \cdot (\e^{-V - U - \varepsilon \delta U} \nabla \phi_{U + \varepsilon \delta U})
        = (f- I) \e^{U + \varepsilon \delta U}.
    \end{equation}
    The function $(f-I)\e^{U+\varepsilon \delta U}$ has zero mean with respect to $\mu_{U + \varepsilon \delta U}$.
    By~\cref{assumption:as1} and the fact that $\delta U \in \smoothcompact(\domain^d)$,
    it holds that $(f-I)\e^{U+\varepsilon \delta U} \in L^2(\e^{-V-U-\varepsilon \delta U})$
    and that $\mu_{U+\varepsilon \delta U}$ satisfies a Poincaré inequality,
    by the Holley--Stroock theorem.
    Consequently, there exists a unique solution in $L^2_0(\e^{-V-U-\varepsilon \delta U})$ to~\eqref{eq:perturbed_poisson} by~\cref{lemma:asymptotic_variance}.
    A simple calculation using the fact that~$\delta U\in \smoothcompact(\domain^d)$ gives
    \begin{equation}
        \label{eq:func_der_ztilde}
        \d\Z{U} \cdot \delta U = - \int_{\domain^d} \delta U \e^{-V-U}.
    \end{equation}
    For the second term on the right-hand side of~\eqref{eq:def_func_der},
    we have by~\cref{lemma:stability_gradient_sol_poisson} that
    \begin{align*}
        &\frac{1}{\varepsilon} \int_{\domain^d} (\phi_{U + \varepsilon \delta U} - \phi_{U}) (f-I) \e^{-V}
        = - \frac{1}{\varepsilon} \int_{\domain^d} (\phi_{U + \varepsilon \delta U} - \phi_{U}) (\mathcal L_U \phi_U) \, \e^{-V-U} \\
        &\qquad = \frac{1}{\varepsilon} \int_{\domain^d} \nabla(\phi_{U + \varepsilon \delta U} - \phi_{U}) \cdot \nabla \phi_U \, \e^{-V-U}
        \xrightarrow[\varepsilon \to 0]{}
        \int_{\domain^d} \nabla \psi_{U,\delta U} \cdot \nabla  \phi_{U} \, \e^{-V-U} \\
        &\qquad = \int_{\domain^d} (\mathcal L_U \psi_{U,\delta U}) \phi_{U} \, \e^{-V-U}
        = \int_{\domain^d} \Bigl(- \e^{U+V} \nabla \cdot \left(\e^{-U-V}\delta U \grad \phi_U\right)\Bigr)   \phi_{U} \, \e^{-V-U} \\
        &\qquad = \int_{\domain^d} \delta U \lvert \nabla \phi_U \rvert^2  \, \e^{-V-U},
    \end{align*}
    where~$\psi_{U,\delta U}$ is the solution to the Poisson equation~\eqref{eq:pert2}.
    The equalities before and after the limit follow from the definitions of~$\phi_U$ and~$\psi_{U,\delta U}$ as weak solutions to~\eqref{eq:poisson} and~\eqref{eq:pert2},
    respectively.
    The last inequality is obtained by integration by parts,
    which is justified because~$\delta U$ is compactly supported.
    % The first bracketed expression in the integrand converges,
    % when divided by~$\varepsilon$, to $-\delta U$ in $L^2(\mu_U)$.
    % The second bracketed expression converges in the same limit to $\lvert \nabla \phi_U \rvert^2$ by~\cref{lemma:stability_gradient_sol_poisson}.
    Combining this equation with~\eqref{eq:def_func_der} and~\eqref{eq:func_der_ztilde},
    we deduce~\eqref{eq:funcder}.
\end{proof}

Before presenting the numerical method for approaching the optimal biasing potential~$U$,
we mention two corollaries of~\cref{proposition:functional_derivative_asym_var}.
\begin{corollary}
    [Critical points]
    \label{corollary_critical_points}
    Suppose that~$U \in \mathfrak U_0$.
    Then $U$ is a critical point of the asymptotic variance viewed as a functional of $U$
    if and only if the corresponding solution to the Poisson equation~\eqref{eq:poisson} satisfies
    \begin{equation}
        \label{eq:eik}
        \abs*{\nabla{\phi_{U}}}^2 = \int_{\domain^d} \abs*{\nabla \phi_{U}}^2 \d \mu_{U}.
    \end{equation}
    In other words, the norm of $\nabla{\phi_{U}}$ is constant over $\domain^d$.
\end{corollary}

\begin{corollary}
    [No smooth minimizer]
    \label{corollary:no_smooth_minimizer}
    Let $\domain = \torus$.
    Then there is no biasing potential~$U \in C^\infty(\torus^d)$ that is a critical point of the asymptotic variance~$\sigma^2_f[U]$.
\end{corollary}
\begin{proof}
    Assume for contradiction that there is a smooth biasing potential $U$ for which~\eqref{eq:eik} holds.
    By elliptic regularity, the corresponding solution~$\phi_{U}$ to the Poisson equation~\eqref{eq:poisson} is also smooth and so,
    by the extreme value theorem,
    it attains its minimum at some point in~$\torus^d$, where $\nabla \phi_U$ vanishes.
    By~\eqref{eq:eik}, this implies that $\nabla \phi_{U} = 0$ for all $x \in \torus^d$ and, therefore, $- \mathcal L_U \phi_{U} = 0 = (f-I) \, \e^{U}$.
    This is a contradiction because,  by~\cref{assumption:assumptions_V},
    the observable~$f$ is not everywhere equal to~$I$.
\end{proof}

\Cref{corollary:no_smooth_minimizer} highlights a limitation of the target-oriented approach taken in this section,
as singular potentials are impractical at the numerical level and unlikely to be of any use for different observables.
This motivates the approach taken in~\cref{sec:minimizing_the_expected_asymptotic_variance},
which aims at finding a biasing potential~$U$ that leads to a reduction in variance for not just one but a family of observables.

\subsubsection*{Steepest descent method}%
To conclude this section,
we present an iterative approach for approximating a minimizer of~$\sigma^2_f[U]$.
We focus on the case where $\domain = \torus$ for simplicity,
but the approach is easy to generalize to other settings (see \cref{sub:numopt}).
Since an expression for the functional derivative of~$\sigma^2_f[U]$ is available by~\cref{proposition:functional_derivative_asym_var},
we employ a method based on steepest descent.
Each step of the method may be decomposed into three stages:
\begin{itemize}
    \item
        First, an approximate solution to the Poisson equation~\eqref{eq:poisson} is computed.
        A number of numerical methods can be employed to this end.
        Given that the optimal potential always exhibits singularities when $\domain = \torus$,
        we opt for a finite difference approach rather than,
        for example, a spectral method~\cite{MR3683698,shen2011spectral}.
        The details of the finite difference method are presented in~\cref{sec:numerical_discretization_of_the_Poisson_equation},
        together with a convergence proof in the setting where~$U$ is regular.

    \item
        Then, from the solution to~\eqref{eq:poisson},
        an approximation of the gradient~$G$ of the asymptotic variance is calculated based on \cref{proposition:functional_derivative_asym_var}.

    \item
        Finally, the potential~$U$ is updated according to $U \gets U - \eta \, G$,
        where $\eta$ is found by backtracking line search following Armijo's method~\cite{MR191071}.
\end{itemize}
These steps are repeated until the $L^2(\torus)$ norm of the gradient is sufficiently small.
A couple of comments are in order.
First, the expression of~$G$ depends on the considered Hilbert functional space.
By~\eqref{eq:funcder}, the functions
\begin{align}
    \nonumber
    &G_{L^2(\e^{-V-U})} =
    \frac{\Z{U}}{Z^2} \bigg( \abs*{\nabla{\phi_{U}}}^2 - \int \abs*{\nabla {\phi_{U}}}^2 \, \d \mu_{U} \bigg),  \\
    \label{eq:explicit_functional_derivative_weighted}
    &G_{L^2(\e^{-V})} =
    \frac{\Z{U}}{Z^2} \bigg( \abs*{\nabla{\phi_{U}}}^2 - \int \abs*{\nabla {\phi_{U}}}^2 \, \d \mu_{U} \bigg) \e^{-U}, \\
    \label{eq:explicit_functional_derivative}
    &G_{L^2(\torus)} =
    \frac{\Z{U}}{Z^2} \bigg( \abs*{\nabla{\phi_{U}}}^2 - \int \abs*{\nabla {\phi_{U}}}^2 \, \d \mu_{U} \bigg) \e^{-U-V},
\end{align}
are all ascent directions for $\sigma^2_f[U]$,
corresponding to the gradients in $L^2(\e^{-V-U})$, in $L^2(\e^{-V})$ and in $L^2(\torus)$,
respectively.
Of these three options, the latter two are better suited for use in an optimization method.
Indeed, employing the $L^2(\e^{-V-U})$ derivative would lead to a change of metric with each update of~$U$,
which precludes the use of methods that rely on information from multiple steps,
such as the Barzilai--Borwein method~\cite{MR967848}.
In the numerical experiments presented in~\cref{sec:examples_and_numerical_experiments},
we use~\eqref{eq:explicit_functional_derivative},
but good results can also be obtained by using~\eqref{eq:explicit_functional_derivative_weighted}.

Second,
the gradient needs to be discretized in practice.
In order to avoid convergence issues,
it is desirable that the discretized gradient is itself the gradient of a function.
Therefore, we use the discretization given in~\eqref{eq:funcder_discrete},
which is guaranteed to be the gradient of an appropriate discretization of the asymptotic variance;
see~\cref{proposition:functional_derivative_asym_var_discrete}.

\section{Minimizing the asymptotic variance for a class of observables}%
\label{sec:minimizing_the_expected_asymptotic_variance}
Assume that the set of observables of which we want to compute the expectation is well described by a Gaussian random field
\begin{equation}
    \label{eq:random_observable}
    f = \sum_{j=1}^{J} \sqrt{\lambda_j} u_j f_j,
    \qquad u_j \sim \mathcal N(0, 1),
    \qquad \lambda_j \in (0, \infty),
\end{equation}
where $(f_j)_{1\leq j \leq J}$ are given functions from $\domain^d$ to $\real$ and the random variables $(u_j)_{1 \leq j \leq J}$ are independent.
This equation defines a probability distribution $\mathcal F$ on the space of observables
as the pushforward of the finite-dimensional Gaussian measure $\mathcal G = \mathcal N(0, \matid_J)$ on $\real^J$;
the probability measure $\mathcal F$ assigns a probability 1 to $\Span(f_1, \dotsc, f_J)$.
One may wonder whether it is possible to minimize the average asymptotic variance for observables drawn from this distribution.
For clarity, we denote by $\expect_{\mathcal F}$ expectations with respect to observables.
Within this section,
we also use the notation
\[
    \mathfrak U_0 = \bigcap_{j=1}^{J} \mathfrak U_0(V, f_j),
    \qquad \mathfrak U = \bigcap_{j=1}^J \mathfrak U(V,f_j),
\]
where~$\mathfrak U_0(V,f_j)$ and~$\mathfrak U(V,f_j)$ are the sets defined before~\cref{assumption:as1} and after~\eqref{eq:explicit_sigma_1d}, respectively.
These sets are assumed to be non-empty in this section.
Denoting by $\phi$ the solution in $H^1(\mu_{U}) \cap L^2_0(\mu_{U})$ to
\[
    - \mathcal L_U \phi = (f - I) \e^{U},
\]
and assuming that~$U \in \mathfrak U_0$,
we have by~\cref{lemma:asymptotic_variance} that
\begin{align*}
    \expect_{\mathcal F} \bigl[\sigma^2_f[U] \bigr]
    &=\expect_{\mathcal F} \left[ \frac{2\Z{U}}{Z^2}\int_{\domain^d} \phi (f-I) \e^{-V}  \right] \\
    &=  \frac{2\Z{U}}{Z^2} \, \expect_{\mathcal G} \left[  \sum_{j=1}^{J} \sum_{k=1}^{J} u_j u_k \sqrt{\lambda_j \lambda_k}\int_{\domain^d} \phi_k (f_j-I_j) \e^{-V}  \right],
\end{align*}
where, for $1 \leq j \leq J$,
the function $\phi_j$ is the unique solution in $H^1(\mu_{U}) \cap L^2_0(\mu_{U})$ to the Poisson equation~$-\mathcal L_U \phi_j = (f_j - I_j) \e^U$,
with $I_j := \mu(f_j)$.
Since $\expect_{\mathcal G} [u_j u_k] = \delta_{jk}$,
we obtain by rearranging the previous expression that
\begin{align}
    \label{eq:expectation_asymptotic_variance}
    \sigma^2[U] :=
    \expect_{\mathcal F} \bigl[\sigma^2_f[U] \bigr]
    = \frac{2\Z{U}}{Z^2} \sum_{j=1}^{J} \lambda_j \int_{\domain^d} \phi_j (f_j-I_j) \e^{-V}
    = \sum_{j=1}^{J} \lambda_j \sigma^2_{f_j}[U].
\end{align}
Therefore, minimizing the expectation $\expect_{\mathcal F} [ \sigma^2_f ]$ amounts to minimizing the sum on the right-hand side of~\eqref{eq:expectation_asymptotic_variance}.

\subsection{\texorpdfstring%
    {Optimal biasing potential in the one-dimensional setting with $\domain = \real$}
    {Optimal biasing potential in the one-dimensional setting on the real line}
}
In the one-dimensional setting with $\domain = \real$,
an explicit expression for the infimum of the asymptotic variance can be obtained,
similar to that~\cref{proposition:bound_asym_var_1d} in the case of a single observable.
In order to state a precise result,
we introduce the notation
\[
    \widehat \sigma^2[U]
    = \sum_{j=1}^{J} \lambda_j \widehat \sigma^2_{f_j}[U],
\]
where $\widehat \sigma^2_{g}$ for an observable~$g$ was defined in~\eqref{eq:explicit_sigma_1d}.
Let us recall that, by the reasoning in~\cref{sub:explicit_optimal_potential_in_dimension_one},
the quantity $\widehat \sigma^2[U]$ coincides with~$\sigma^2[U]$ in~\eqref{eq:expectation_asymptotic_variance} when~$U \in \mathfrak U_0$,
but~$\widehat \sigma[U]$ is well-defined more generally for any~$U \in \mathfrak U$.
In~\cref{propopsition:optimal_potential_class} below,
we give a bound from below on~$\widehat \sigma^2[U]$ in the particular case where~$d = 1$ and $\domain = \real$.
Before presenting this result,
we introduce the notation
\[
    F_j(x) = \int_{0}^{x} \bigl(f_j(\xi)-I_j\bigr) \e^{-V(\xi)} \d \xi, \qquad
    A_{\real,j} = - \int_{-\infty}^{0} \bigl(f_j(\xi)-I_j\bigr) \e^{-V(\xi)} \d \xi.
\]
\begin{remark}
    In the rest of this section,
    $F_j$ and~$A_{\real,j}$ always appear together as $F_j - A_{\real,j}$,
    which can be rewritten as an integral over~$(-\infty, x]$ with the same integrand as in the definition of~$F_j$,
    that is to say
    \[
        F_j(x) - A_{\real,j} = \int_{- \infty}^x \bigl(f_j(\xi) - I_j\bigr) \e^{-V(\xi)} \, \d \xi.
    \]
    For the sake of conciseness,
    we could introduce a new notation to refer to $F_j - A_{\real,j}$,
    but we refrain from doing so in order to keep the notation consistent with that used in~\cref{sec:minimizing_the_asymptotic_variance_for_one_observable}.
\end{remark}
\begin{lemma}
    \label{propopsition:optimal_potential_class}
    Assume that $d = 1$ and $\domain = \real$.
    Then, for all $U \in \mathfrak U$, it holds that
    \begin{equation}
        \label{eq:lower_bound_class}
        \min_{U \in \mathfrak U} \widehat \sigma^2(U)
        = \frac{2}{Z^2} \left( \int_{\real}  \sqrt{ \sum_{j=1}^{J}  \lambda_j \lvert F_j - A_{\real,j} \rvert^2 }  \right)^2.
    \end{equation}
    The minimum is achieved for
    \begin{equation}
        \label{eq:optimal_solution_class}
        U = U_* := - V -\log \left( \sqrt{\sum_{j=1}^{J} \lambda_j \lvert F_j - A_{\real,j} \rvert^2} \right) \in \mathfrak U.
    \end{equation}
    In this case
    \(
        \e^{-V-U_*} \text{ is proportional to } \sqrt{\sum_{j=1}^{J} \lambda_j \lvert F_j - A_{\real,j} \rvert^2}.
    \)
\end{lemma}

\begin{proof}
    We first show that $\widehat \sigma^2[U]$ is bounded from below by the right-hand side of~\eqref{eq:lower_bound_class}.
    This is trivial if~$\widehat \sigma^2[U]$ is infinite,
    so we assume from now on that~$\widehat \sigma^2[U] < \infty$.
    Using the definition of $\widehat \sigma^2[U]$,
    we have
    \begin{align}
        \label{eq:expression_class_sigma_hat}
        \widehat \sigma^2[U] =
        \sum_{j=1}^{J} \lambda_j \widehat \sigma^2_{f_j}[U]
        =  \frac{2\Z{U}}{Z^2} \int_{\real} \left( \sum_{j=1}^{J}  \lambda_j \lvert F_j - A_{\real,j} \rvert^2 \right) \e^{V+U}.
    \end{align}
    Here, we used that the infimum in the definition~\eqref{eq:explicit_sigma_1d} of~$\widehat \sigma^2_{f_j}$ is achieved for~$A = A_{\real,j}$,
    as explained in~\cref{remark:infimum_in_hat_sigma}.
    Let us introduce the notation
    \[
        G = \sqrt{\sum_{j=1}^{J} \int_{\real} \lambda_j  \lvert F_j - A_{\real,j} \rvert^2}.
    \]
    Since~$\widehat \sigma^2[U] < \infty$ by assumption,
    the set over which $G^2 \e^{V+U}$ takes an infinite value is of zero Lebesgue measure.
    Therefore, by the Cauchy--Schwarz inequality,
    \[
        \widehat \sigma^2[U]
        = \frac{2\Z{U}}{Z^2} \int_{\real} \lvert G \rvert^2 \e^{U+V}
        \geq \frac{2}{Z^2} \left(\int_{\real} \lvert G \rvert \sqrt{\e^{U+V}} \sqrt{\e^{-V-U}}\right)^2
        = \frac{2}{Z^2} \left(\int_{\real} G \right)^2.
    \]
    The claim that~$U_*$ in~\eqref{eq:optimal_solution_class} achieves the lower bound can be verified by substitution in~\eqref{eq:expression_class_sigma_hat}.
\end{proof}

\begin{remark}
    Notice that, unless the functions $(F_j - A_{\real,j})_{1\leq j \leq J}$ share a common root,
    the optimal biasing potential~$U_*$ in~\eqref{eq:optimal_solution_class} is a smooth function.
\end{remark}

In the same way that we obtained~\cref{proposition:sharpness} from~\cref{proposition:bound_asym_var_1d},
we deduce from~\cref{propopsition:optimal_potential_class} the following result.
The proof is a simple adaptation of that of~\cref{proposition:sharpness}, so we omit it.
\begin{proposition}
    Assume that $d = 1$ and $\domain = \real$,
    and suppose that~$0 \in \mathfrak U_0$.
    Then
    \[
        \inf_{U \in \mathfrak U_0} \sigma^2[U]
        = \inf_{U \in C^{\infty}_{\rm c}(\real)} \sigma^2[U]
        = \frac{2}{Z^2} \left( \int_{\real}  \sqrt{ \sum_{j=1}^{J}  \lambda_j \lvert F_j - A_{\real,j} \rvert^2 }  \right)^2.
    \]
\end{proposition}
% \begin{proof}
% \end{proof}

\subsection{Numerical optimization}
\label{sub:numopt}
A result similar to~\cref{propopsition:optimal_potential_class} is not easily available when~$d=1$ and~$\domain = \torus$,
because the constant~$A_{\torus}$ given in~\eqref{eq:A_domain} depends on $U$.
In this case or in the multi-dimensional setting,
one can resort to a steepest descent approach in order to find the optimal biasing potential.
It is easy to prove,
based on~\cref{proposition:functional_derivative_asym_var},
that the functional derivative of $\sigma^2[U]$ is given by
\begin{equation}
    \label{eq:functional_derivative_class_of_obs}
    \frac{1}{2} \d \sigma^2[U] \cdot \delta U
    = \frac{\Z{U}^2}{Z^2} \int_{\domain^d} \left( \delta U - \int_{\domain^d} \delta U \, \d \mu_U \right) \left( \sum_{j=1}^{J} \lambda_j \abs*{\nabla{\phi_j}}^2  \right) \, \d \mu_{U}.
\end{equation}
The approach presented in~\cref{sub:optimal_solution_in_the_multi_dimensional_setting} can then be applied \emph{mutatis mutandis}.
Numerical experiments illustrating the potential found as a result of this procedure are presented in~\cref{sec:examples_and_numerical_experiments}.

\subsection{Free energy biasing}
In the molecular dynamics literature,
variance reduction over a compact state space is often achieved by \emph{free energy biasing},
a heuristic approach which,
in the absence of coarse graining via a reaction coordinate,
amounts to setting $U = -V$;
see~\cite{MR3509213} and the references therein.
To conclude this section, we address the following related question:
is there a probability distribution~$\mathcal F$ on observables such that $U = -V$ is a minimizer of the average asymptotic variance $\sigma^2[U]$,
i.e.\ for which energy biasing is optimal?
While we will not be able to provide a definite answer to this question,
we shall prove in \cref{proposition:when_is_free_energy_biasing_optimal} that,
for an appropriate probability measure on observables, the biasing $U = -V$ corresponds to a critical point of~$\sigma^2[U]$.
This does not imply that $U = -V$ is necessarily optimal,
because $\sigma^2[U]$ is not convex in general;
see~\cref{remark:non_convexity_torus} in \cref{sec:second_variation_of_the_asymptotic_variance}.

\begin{proposition}
    \label{proposition:when_is_free_energy_biasing_optimal}
    Suppose that~$\domain = \torus$ with~$d=1$ and assume,
    for~$J \in 2\nat_{>0}$,
    that $(f_j, \lambda_j)_{1 \leq j \leq J}$
    are the~$J$ first eigenfunctions and eigenvalues of the operator $\mathcal K = \e^{V}(-\laplacian + \tau^2 \id)^{-\alpha} \e^{-V}$,
    where $\alpha \in (0, \infty)$ and~$\tau \in \real$ are parameters of the random field
    and~$\mathcal K$ is viewed as a compact self-adjoint operator on the following space of functions defined on $\torus$:
    \[
        \left\{ f \in L^2\left(\e^{-2V}\right) \colon \int_{\torus} f \e^{-V} = 0 \right\}.
    \]
    Then the average asymptotic variance $\sigma^2[U]$ admits a critical point for $U = -V$.
\end{proposition}
\begin{proof}
    The eigenfunctions of the operator~$\mathcal K$ are given by
    \begin{equation}
        \label{eq:definition_trigonometric_functions}
        f_j(x) = \e^{V}
        \left\{ \begin{aligned}
                &\sin\left(\frac{j + 1}{2}x\right), \quad && \text{if $j$ is odd}, \\
                &\cos\left(\frac{j}{2}x\right), \quad && \text{if $j$ is even}.
        \end{aligned} \right.
    \end{equation}
    Note that $I_j = 0$ for all $1 \leq j \leq J$.
    When $U = -V$,
    the generator of the Markov semigroup associated with~\eqref{eq:mix} is just the Laplacian operator.
    Therefore, for a given $j \in \{1, \dotsc, J\}$,
    the solution to the Poisson equation~$- \mathcal L_U \phi_j = (f_j - I_j) \e^{U}$ is given by
    \begin{equation*}
        \phi_j(x) = \left\lceil \frac{j}{2} \right\rceil^{-2} f_j(x).
    \end{equation*}
    Since $\lambda_j = \lambda_{j+1}$ for all odd values of $j$
    and since $\sin^2 + \cos^2 = 1$,
    we deduce that the sum on the right-hand side of~\eqref{eq:functional_derivative_class_of_obs} is constant,
    implying that the functional derivative of $\sigma^2[U]$ is zero when evaluated at the biasing potential $U = -V$.
\end{proof}
\begin{remark}
    The choice of the operator~$\mathcal K$ is motivated by the form~\eqref{eq:func_der_ztilde} of the desired observables,
    which is itself motivated by the fact that these observables lead to Poisson equations with explicit trigonometric solutions when~$U = -V$.
\end{remark}
\begin{remark}
    In this section,
    we assumed for simplicity that the random observable admitted a finite expansion of the form~\eqref{eq:random_observable}.
    The results we obtained could in principle be extended to the case of a more general Gaussian field,
    with an infinite Karhunen--Loève series.
    For background on Gaussian variables in infinite dimension,
    see for example~\cite[Section 1.5]{MR3288096}.
    The reference~\cite{MR3839555} is also useful for understanding the regularity of infinite series of the form~\eqref{eq:random_observable}.
\end{remark}

\section{Examples and numerical experiments}
\label{sec:examples_and_numerical_experiments}
We begin in~\cref{sub:one_dim_examples} by presenting examples and numerical experiments in dimension 1.
Then, in~\cref{sub:two_dim_examples},
we present numerical experiments for the case where the state space is $\torus^2$ and the optimal biasing potential is approximated by steepest descent.
Finally, in~\cref{sub:class_examples},
we illustrate the approach proposed in~\cref{sec:minimizing_the_expected_asymptotic_variance}.

\subsection{One-dimensional examples}
\label{sub:one_dim_examples}

In this subsection,
we present a few examples illustrating the optimal biasing potential~$U$ for various observables and underlying potentials~$V$.
In all the figures,
the optimal potential depicted is calculated numerically using the steepest descent approach presented in~\cref{sub:explicit_optimal_potential_in_dimension_one}.
It is apparent in~\cref{example:one_dim_cosine,example:1d_subset_example} that this approach indeed yields the optimal biasing potential~\eqref{eq:optu},
an explicit expression of which is given in these examples.

The first few examples aim at illustrating settings where the optimal biasing potential exhibits different levels of singularity.
The optimal potential is smooth in~\cref{example:smooth},
it exhibits two singularity points in~\cref{example:one_dim_cosine},
and it blows up over a whole interval in~\cref{example:1d_subset_example}.
In these examples, the reference probability measure is unimodal,
and the gain in asymptotic variance obtained from using the optimal potential is small.
Finally, in example \cref{example:1d_metastable},
a multi-modal reference probability measure is considered,
and it is observed that importance sampling enables a considerable decrease in asymptotic variance in this case.
Without loss of generality,
we normalize in the figures the potentials~$U$ so that the minimum value of $V+U$ over the domain considered is 0.

\begin{example}
    \label{example:smooth}
    Assume that $\domain = \real$ and $f = V'$.
    Then $I = 0$ and $F(x) - A_{\real} = \e^{-V(x)}$;
    in this case,
    the optimal biasing potential~\eqref{eq:optu} is $U_* = 0$.
\end{example}

\begin{example}
    \label{example:one_dim_cosine}
    Assume that $\domain = \torus$, $V = 0$ and $f(x) = \cos(x)$.
    Then $I = 0$ and $F(x) = \sin(x)$.
    The constant $A^*_{\torus}$ in~\eqref{eq:equation_A} is 0,
    and so the optimal biasing potential~\eqref{eq:optu} is given by
    \[
        U_*(x) = - \log \abs{\sin(x)}.
    \]
    This potential is illustrated in~\cref{fig:optimal_perturbation_potential}
    together with the corresponding measure~$\mu_{U_*}$.
    \iffalse
    It is possible to show that~$\mu_{U_*}$ does not satisfy a Poincaré inequality.
    Indeed, let $g_{\varepsilon}\colon \torus \to \real$ be given by
    \begin{align*}
        g_{\varepsilon}(x) = \sgn(x)\log\left(\frac{\sin(\lvert x \rvert) +\varepsilon}{\varepsilon}\right),
    \end{align*}
    where we identify the torus with $[-\pi, \pi]$.
    The function~$g_{\varepsilon}$ is odd,
    so its mean with respect to~$\mu_{U_*}$ is zero,
    and we observe that
    \begin{align*}
        \int_{\torus} \abs{g_{\varepsilon}'(x)}^2 \, \abs{\sin(x)} \d x
        &= 2\int_{0}^{\pi} \frac{\abs{\sin(x)} \bigl(\cos(x)\bigr)^2}{\bigl(\sin(x) + \varepsilon\bigr)^2} \, \d x \\
        &\leq 2\int_{0}^{\pi} \frac{\abs{\cos(x)}}{\sin(x) + \varepsilon} \, \d x
        = 4\int_{0}^{\pi/2} \frac{\cos(x)}{\sin(x) + \varepsilon} \, \d x
        = 4 \log \left(\frac{1 + \varepsilon}{\varepsilon}\right),
    \end{align*}
    whereas
    \[
        \int_{\torus} \abs{g_{\varepsilon}(x)}^2 \, \abs{\sin(x)} \,\d x
        \geq  2\int_{\pi/4}^{3\pi/4} \abs{g_{\varepsilon}(x)}^2 \, \sin(x) \, \d x
        \geq \pi\sin\left(\frac{\pi}{4}\right) g_{\varepsilon}\left(\frac{\pi}{4}\right)^2 .
    \]
    Since the latter squared norm diverges faster than the former in the limit as $\varepsilon \to 0$,
    it is impossible that a Poincaré inequality holds.
    This can also be verified via Muckenhoupt's criterion;
    see, for example, \cite[Theorem 4.5.1]{MR3155209}.
    \fi
    Notice that the singularities divide the domain into the two regions $[-\pi, 0]$ and $[0, \pi]$,
    but the average of~$f$ with respect to $\e^{-V}$ conditioned to either region is equal to $I = 0$,
    in accordance with the discussion in~\cref{remark:domain_divided}.

    When $U = 0$,
    the solution to the Poisson equation~\eqref{eq:poisson} is given by~$\phi(x) = \cos(x)$ and $Z = 1$,
    so
    the asymptotic variance~\eqref{eq:asymvar_1d} is given by
    \[
        \sigma^2_f[U = 0] = 2 \int_{-\pi}^{\pi} \abs{\sin(x)}^2 \, \frac{\d x}{2 \pi} = 1.
    \]
    The infimum~\eqref{eq:1dsig} of the asymptotic variance,
    on the other hand,
    is given by
    \[
        \frac{2}{Z^2} \bigg(\int_{\torus} \bigl\lvert F(x) - A^*_{\torus} \bigr\rvert \d x \bigg)^2 =
        2 \bigg(\int_{\torus} \bigl\lvert \sin(x) \bigr\rvert \frac{\d x}{2 \pi} \bigg)^2
        = \frac{8}{\pi^2} = 0.810\dotsc
    \]
    The optimal biasing potential therefore leads a reduction in variance of about 19\%.
\end{example}
\begin{figure}[ht]
    \centering
    \includegraphics[width=0.49\linewidth]{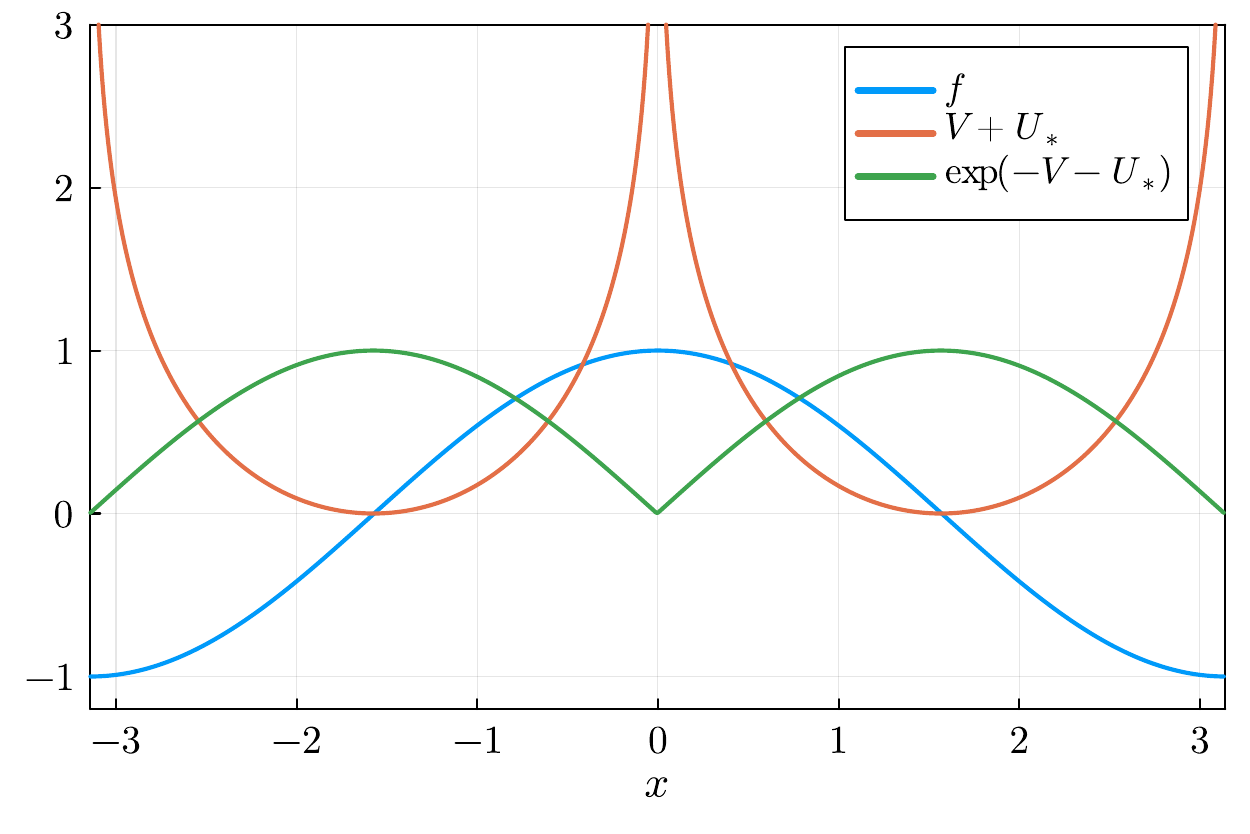}
    \includegraphics[width=0.49\linewidth]{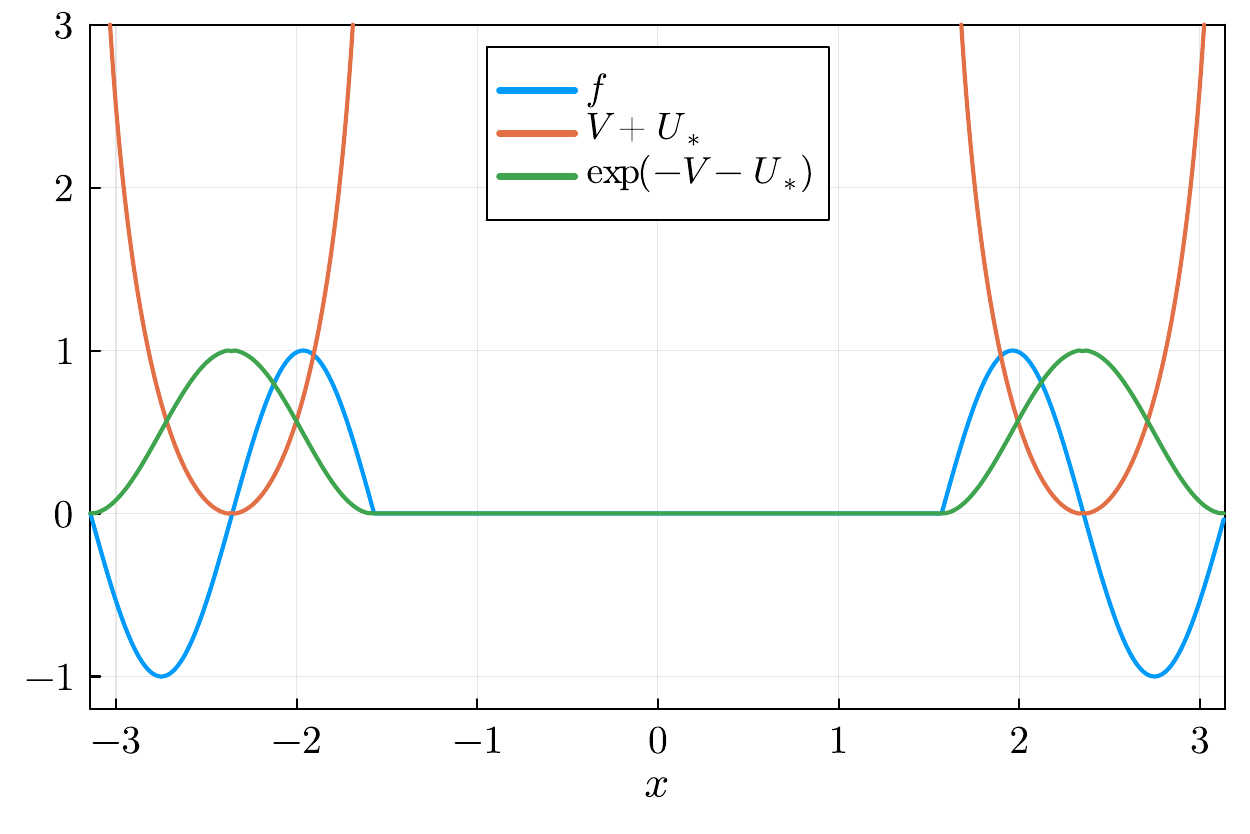}
    \caption{Optimal potential for~\cref{example:one_dim_cosine} (left) and~\cref{example:1d_subset_example} (right).}
    \label{fig:optimal_perturbation_potential}
\end{figure}

As we show above in~\cref{corollary:no_smooth_minimizer},
singularities in the optimal biasing potential are inevitable when the state space is the torus in any dimension.
Moreover, it is possible to construct examples where the optimal measure~$\mu_{U_*}$ is supported on a subset of $\torus$,
as shown in the following example.
\begin{example} % Not a big deal but this f is not smooth
    \label{example:1d_subset_example}
    Consider the case where $\domain = \torus$ and $V(x) = 0$,
    with the observable
    \[
        f(x) =
        \begin{cases}
            \sin(4 \abs{x}) \qquad &\text{if $\abs{x} \geq \frac{\pi}{2}$}, \\
            0 \qquad &\text{otherwise}.
        \end{cases}
    \]
    Then $I = 0$ and
    \[
        F(x) =
        \int_{0}^x f(\xi) \, \e^{-V(\xi)} \, \d \xi =
        \begin{cases}
            \frac{1}{4} \bigl(- 1 + \cos(4 x)\bigr) \qquad &\text{if $x \leq -\frac{\pi}{2}$}, \\
            \frac{1}{4} \bigl(1 - \cos(4 x)\bigr)   \qquad &\text{if $x \geq \frac{\pi}{2}$}, \\
            0 \qquad &\text{otherwise}.
        \end{cases}
    \]
    The constant $A^*_{\torus}$ in~\eqref{eq:equation_A} is again zero,
    and so the optimal biasing potential is~$U_*(x) = - \log |F(x)|$.
    The optimal biasing potential and the corresponding measure~$\mu_{U_*}$ for this example
    are depicted in \cref{fig:optimal_perturbation_potential}.
    As the figure illustrates,
    the Lebesgue density of~$\mu_{U_*}$ with respect to the Lebesgue measure is zero over the interval~$[-\pi/2, \pi/2]$.
\end{example}

%The examples we considered so far did not involve multimodal probability measures,
%and so variance reduction is of limited interest in these cases.
% M: I comment this out because I think if we do not find a natural,
% satisfactory setting where the uniform measure is optimal, then we should not
% emphasize the point that there is not a huge improvement over that case.
% Otherwise it raises the question of what the point is when one can just take
% U=-V to remove metastability. Still, we answer in the negative the question
% of whether the uniform measure is optimal in general.
To conclude this section,
we present an example where using the biasing potential $U_*$ leads to a significant decrease of the variance.
\begin{example}
    \label{example:1d_metastable}
    Consider again the case where $\domain = \torus$,
    with this time $V(x) = 5\cos(2 x)$ and the observable~$f(x) = \sin(x)$.
    The reference dynamics, i.e.\ the dynamics~\eqref{eq:mix} with $U = 0$,
    is metastable in view of the high potential barrier;
    the asymptotic variance when~$U=0$ for the observable considered,
    estimated numerically from~\eqref{eq:asymvar_1d},
    is equal to 3459 after rounding to the closest integer.

    The optimal total potential~$V + U_*$ and probability distribution~$\mu_{U_*}$ are depicted in~\cref{fig:optimal_perturbation_potential_1d_metastable}.
    The asymptotic variance associated with~$U_*$ is about~$3.64$,
    roughly $1000$ times smaller than the asymptotic variance for the reference dynamics.
\end{example}
\begin{figure}[ht]
    \centering
    \includegraphics[width=0.49\linewidth]{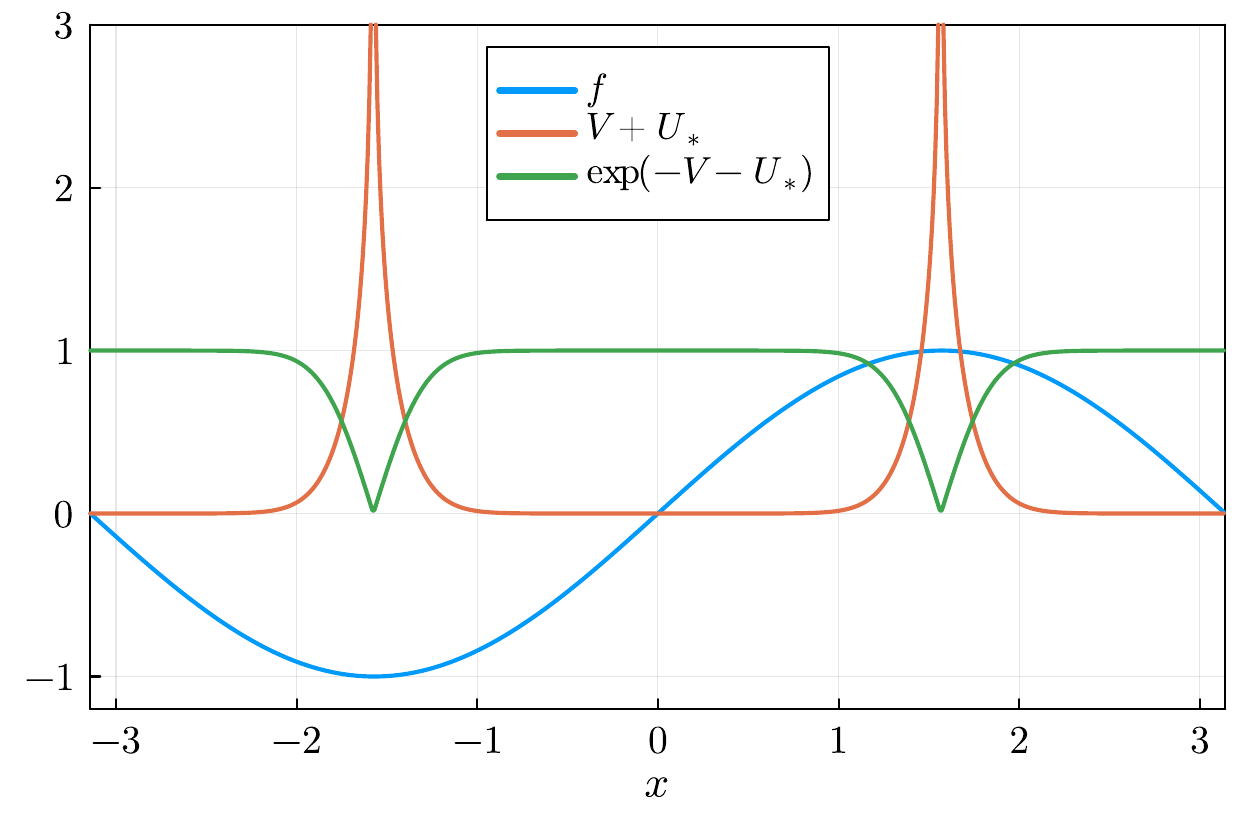}
    \includegraphics[width=0.49\linewidth]{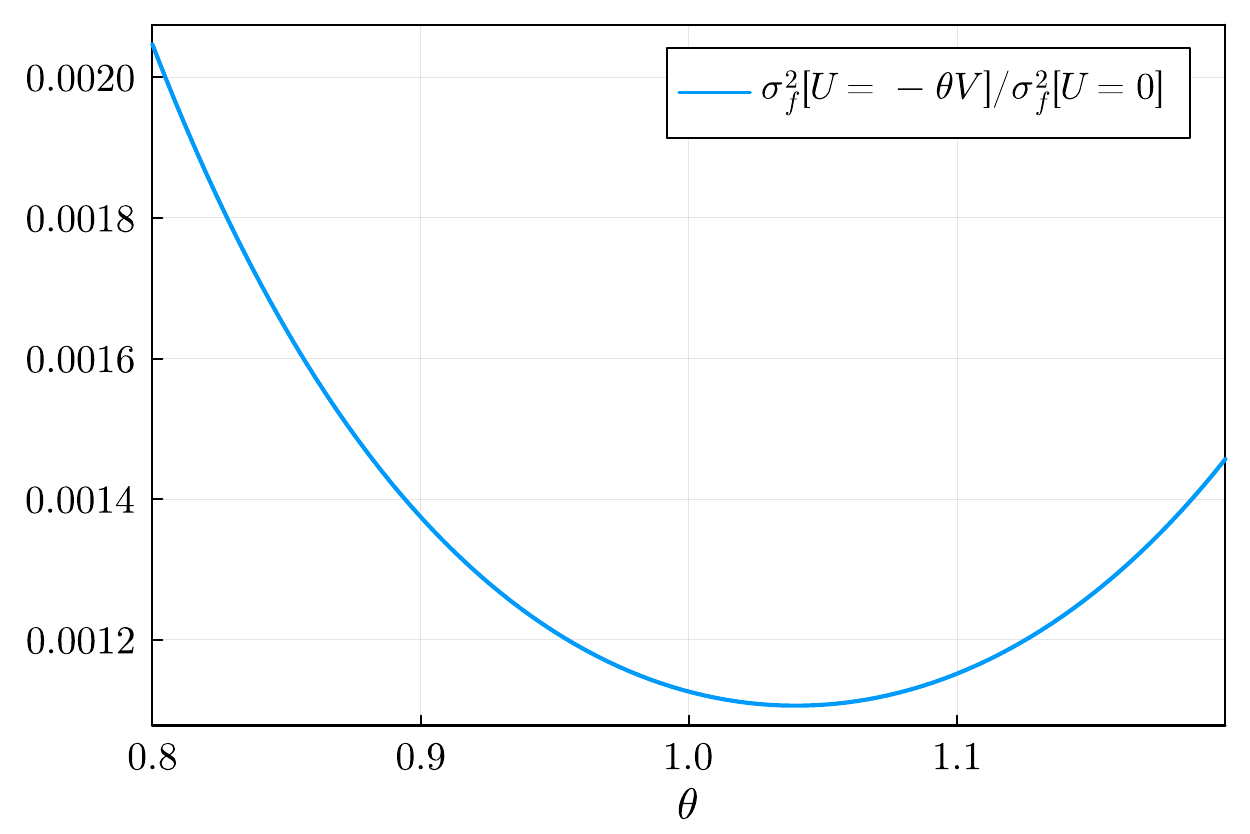}
    \caption{%
        Optimal potential for~\cref{example:1d_metastable} (left),
        and asymptotic variance corresponding to the biasing potential $U = -\theta V$ over the range $[0.8, 1.2]$ of values for $\theta$.
    }
    \label{fig:optimal_perturbation_potential_1d_metastable}
\end{figure}

The numerical values of the asymptotic variance for the examples considered in this section and different choices of~$U$ are summarized in~\cref{tab:asym_vars_1d}.
We also present in this table,
for each of the examples, the value of the asymptotic variance corresponding to case where~$U = - \theta V$,
for the value of~$\theta$ that yields the largest variance reduction.
This approach is found to yield a variance reduction close to the optimal one in the setting of~\cref{example:1d_metastable}.
\begin{table}[ht]
    \centering
    \begin{tabular}{|c|c|c|c|c|}
         \hline
         Test case & $U = 0$ & $U = -V$ & $U = - \theta_* V$ & Optimal $U$
         \\ \hline
         \cref{example:one_dim_cosine} & $1$ & $1$ & $1$ & $0.811$
         \\ \hline
         \cref{example:1d_subset_example} & $1$ & $1$ & $1$ & $0.334$
         \\ \hline
         \cref{example:1d_metastable} & $1$ & $0.00113$ & $0.00111$ ($\theta_* = 1.038$) & $0.00105$
         \\ \hline
    \end{tabular}
    \caption{%
        Ratio of the asymptotic variance of the importance sampling estimator~\eqref{eq:estimator} relative to its value in the case where $U = 0$,
        for different choices of the biasing potential~$U$ and in the one-dimensional setting.
        The parameter $\theta_*$ is the minimizer of $\sigma^2_f[-\theta_* V]$,
        calculated by using the \texttt{Optim.jl}~\cite{Optim.jl-2018} implementation of the LBFGS solver.
        The numbers are rounded to 3 significant digits.
    }
    \label{tab:asym_vars_1d}%
\end{table}

\subsection{Two-dimensional examples}
\label{sub:two_dim_examples}

By~\cref{corollary:no_smooth_minimizer},
the optimal potential when the domain is the two-dimensional torus~$\torus^2$ also exhibits singularities.
In all the examples presented hereafter, these are line singularities.
In order to approximate the optimal potential in the numerical experiments presented in this section,
we use the steepest descent approach presented in~\cref{sub:optimal_solution_in_the_multi_dimensional_setting} with~$150 \times 150$ discretization points.
We begin in \cref{example:2d_first_example,example:2d_difficult} by considering cases where
the reference dynamics does not suffer from metastability.
The gain in asymptotic variance provided by importance sampling is small in these settings.
Then, in~\cref{example:2d_metastable}, we consider a setting where the reference measure is multi-modal,
for which importance sampling leads to a considerable decrease in asymptotic variance.
% By a slight abuse of terminology,
% we call ``optimal potential'' the minimizer obtained by the gradient descent

\begin{example}
    \label{example:2d_first_example}
    We consider the case where the potential is~$V(x) = 0$ and the observable is~$f(x) = \sin(x_1) + \sin(x_2)$.
    This observable has average zero not only with respect to~$\mu$,
    but also with respect to the restrictions of this measure to the subsets $[-\pi/2, \pi/2] \times [-\pi/2, \pi/2]$,
    $[-\pi/2, \pi/2] \times [\pi/2, 3\pi/2]$, $[\pi/2, 3\pi/2] \times [-\pi/2, \pi/2]$,
    and~$[\pi/2, 3\pi/2] \times [\pi/2, 3\pi/2]$
    which together form a partition of $\torus^2$.
    Here we identify subsets of $\real^2$ with their image under the quotient map $\real^2 \to \torus^2$.
    % provided that we identify points in $\real^2$ with their representative in $\torus^2$.
    Interestingly,
    the total potential~$V+U$ corresponding to the optimal biasing potential exhibits singularities precisely at the boundaries between these regions,
    effectively dividing the state space into four separate regions;
    see~\cref{fig:2d_first_example}.
    It appears clearly from the right panel in the same figure that,
    in agreement with~\cref{corollary_critical_points},
    the solution to the corresponding Poisson equation~\eqref{eq:poisson} is affine by parts,
    with discontinuities of the first derivative at singularities of $V+U$.
    The reduction in asymptotic variance corresponding to the optimal potential in this case is only about~$19\%$.
\end{example}
\begin{figure}[ht]
    \centering
    \includegraphics[width=0.49\linewidth]{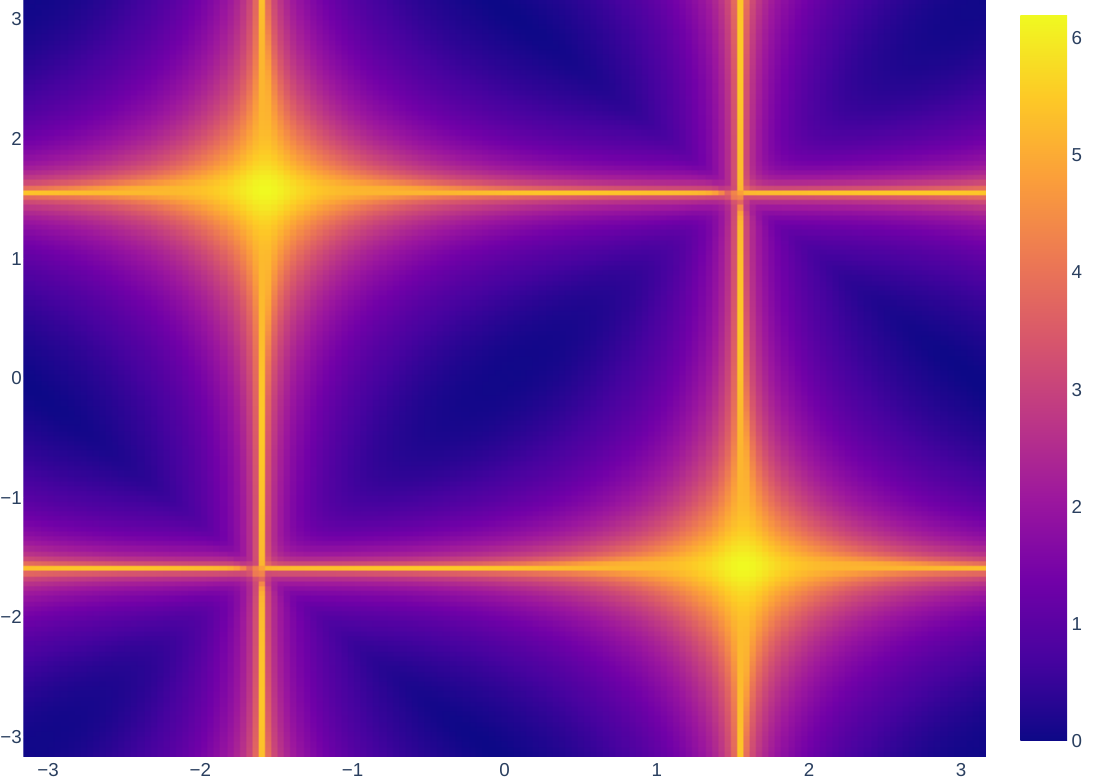}
    \includegraphics[width=0.49\linewidth]{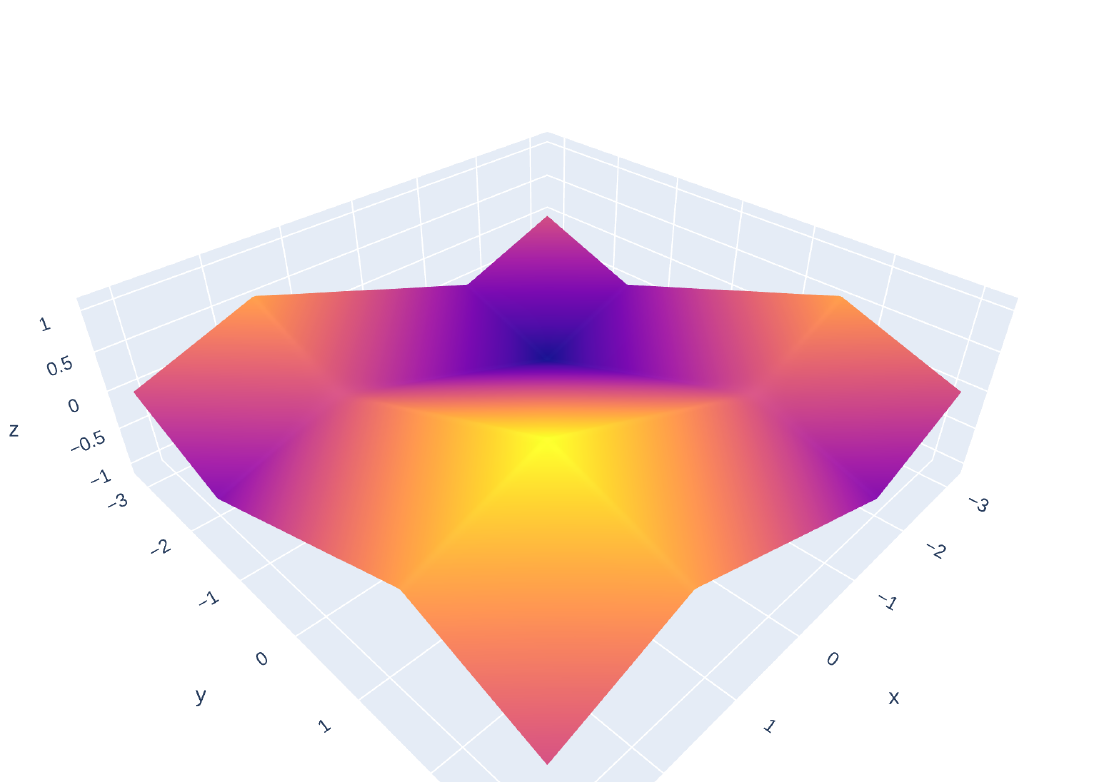}
    \caption{%
        Optimal potential~$V+U$ for~\cref{example:2d_first_example} (left),
        together with the solution to the associated Poisson equation~\eqref{eq:poisson} (right).
    }
    \label{fig:2d_first_example}
\end{figure}

We now present an example with a non-uniform reference distribution~$\mu$.
\begin{example}
    \label{example:2d_difficult}
    In this example, we consider that the potential and observable are given by
    \[
        V(x) = \exp\left(\cos(x_1) \sin(x_2) + \frac{1}{5} \cos(3x_1)\right),
        \qquad
        f(x) = \sin\Bigl(x_1 + \cos(x_2)\Bigr)^3.
    \]
    The potential~$V$ and observable~$f$ are illustrated respectively in the top left and right panels of~\cref{fig:2d_difficult}.
    The corresponding optimal total potential $V+U$,
    together with the associated solution to the Poisson equation~\eqref{eq:poisson},
    are depicted in the bottom left and right panels respectively.
    Once again, it appears that the optimal potential divides the domain into two separate regions where the averages of~$f$ are the same.
    The reduction in asymptotic variance obtained by employing the perturbed dynamics~\eqref{eq:mix} is about~$20\%$.
\end{example}

\begin{figure}[ht]
    \centering
    \includegraphics[width=0.49\linewidth]{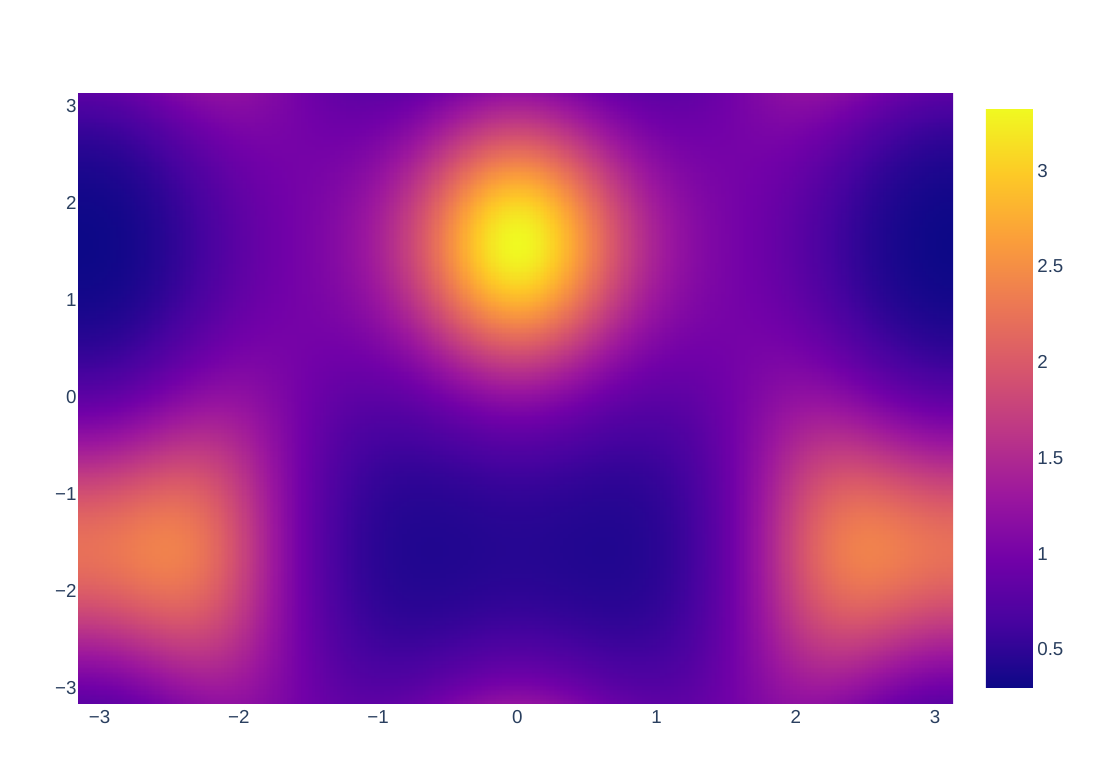}
    \includegraphics[width=0.49\linewidth]{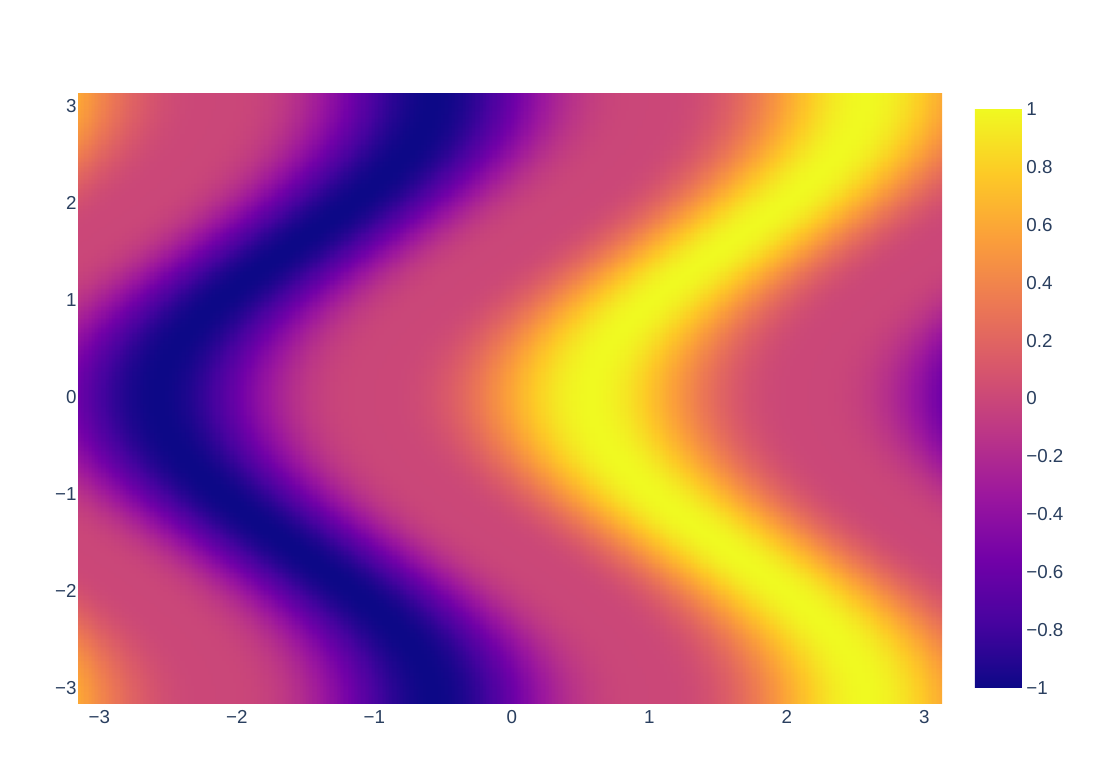}
    \includegraphics[width=0.49\linewidth]{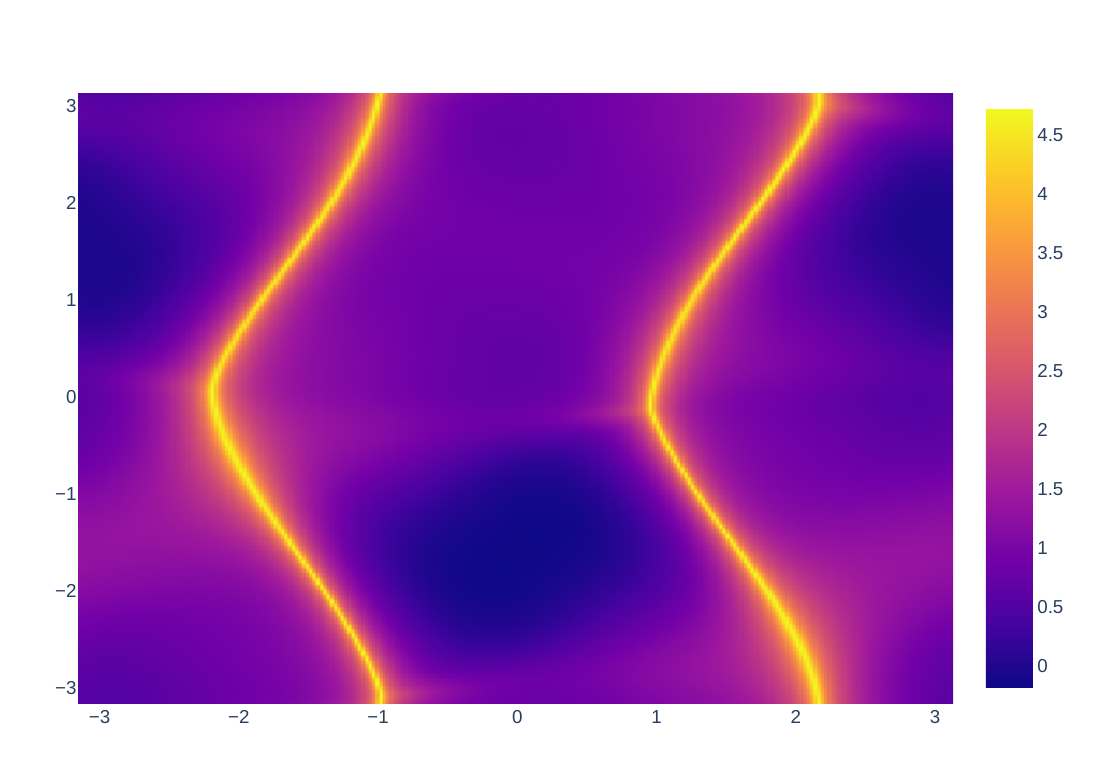}
    \includegraphics[width=0.49\linewidth]{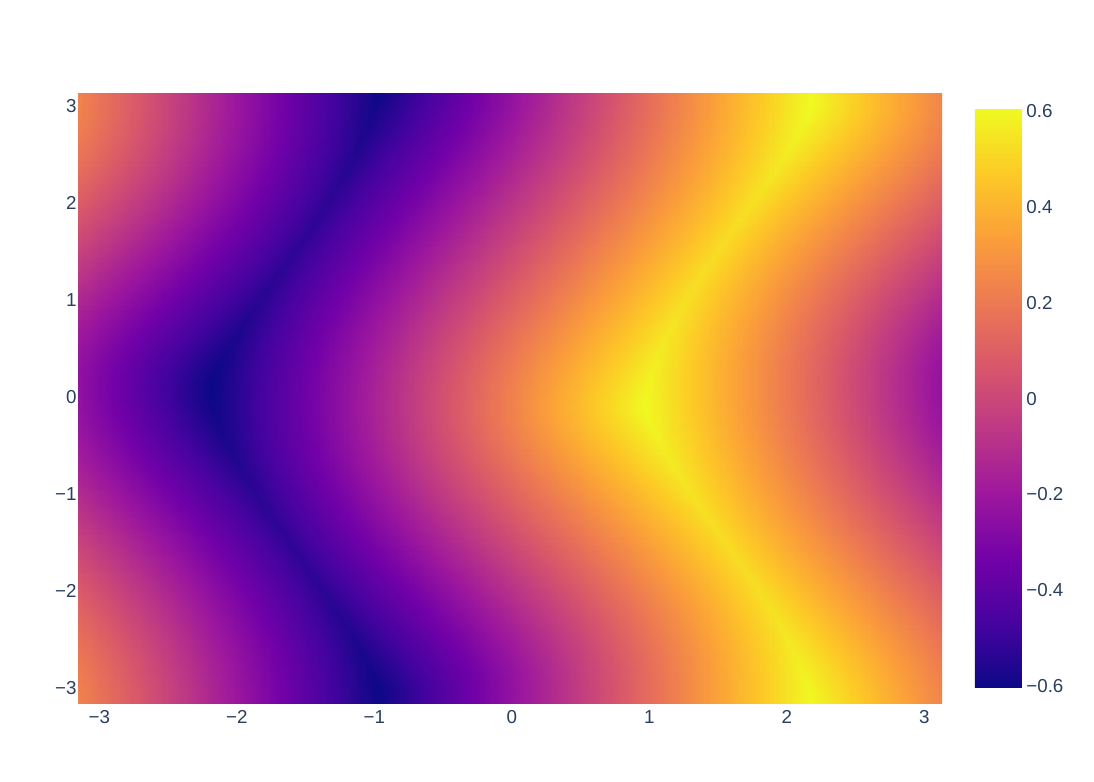}
    \caption{%
        Unperturbed potential~$V$ (top left), observable~$f$ (top right), optimal potential~$V+U_*$ (bottom left), and corresponding solution to the Poisson equation~$\phi_U$ (bottom right)
        for~\cref{example:2d_difficult}.
    }
    \label{fig:2d_difficult}
\end{figure}

To conclude this section,
we present an example where the target probability distribution is multimodal,
in which case a considerable reduction of the asymptotic variance can be achieved.
\begin{example}
    \label{example:2d_metastable}
    We consider the case where $V(x) = 2 \cos(2 x_1) - \cos(x_2)$
    and~$f(x) = \sin(x_1)$.
    The potential~$V(x)$ has two global minima located at~$(\pi/2, 0)$ and~$(-\pi/2, 0)$,
    and the observable~$f(x)$ takes different values when evaluated at these points.
    The optimal total potential~$V+U$ is illustrated in~\cref{fig:2d_metastable}.
    We observe two line singularities
    which effectively divide the domain into two separate regions where the average of~$f$ is equal to $I = 0$.
    The reduction in asymptotic variance obtained by employing the perturbed dynamics~\eqref{eq:mix} is about~$86\%$.
\end{example}

\begin{figure}[ht]
    \centering
    \includegraphics[width=0.49\linewidth]{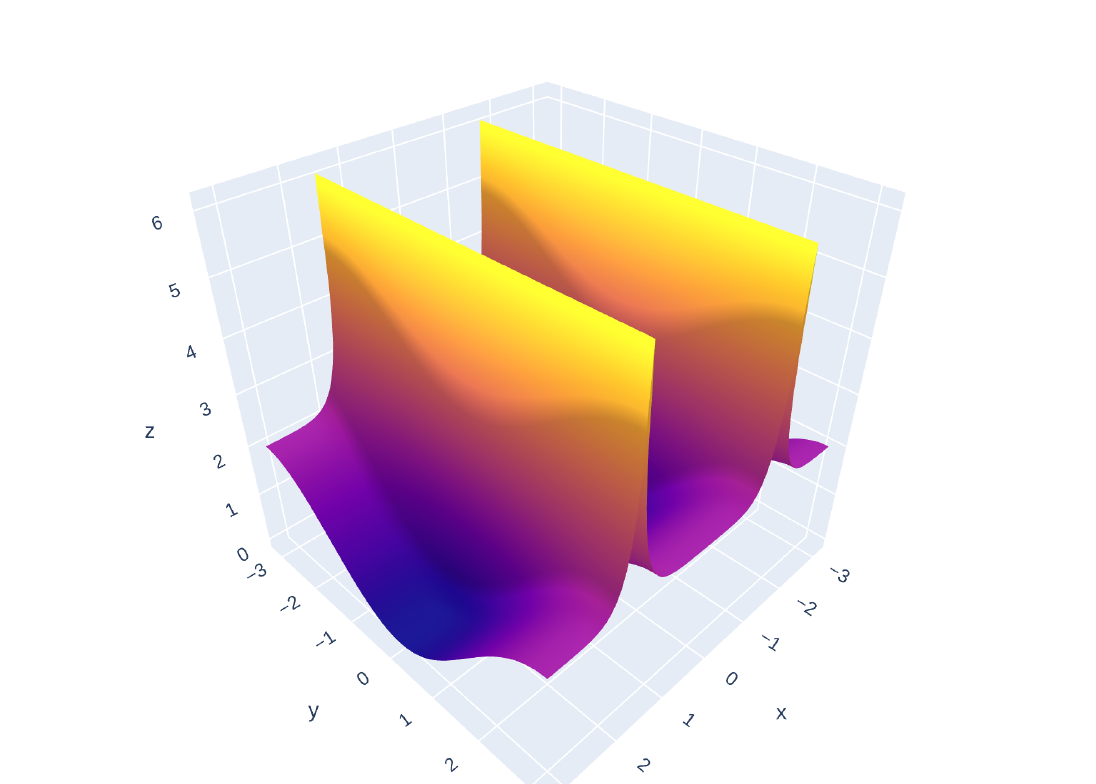}
    \includegraphics[width=0.49\linewidth]{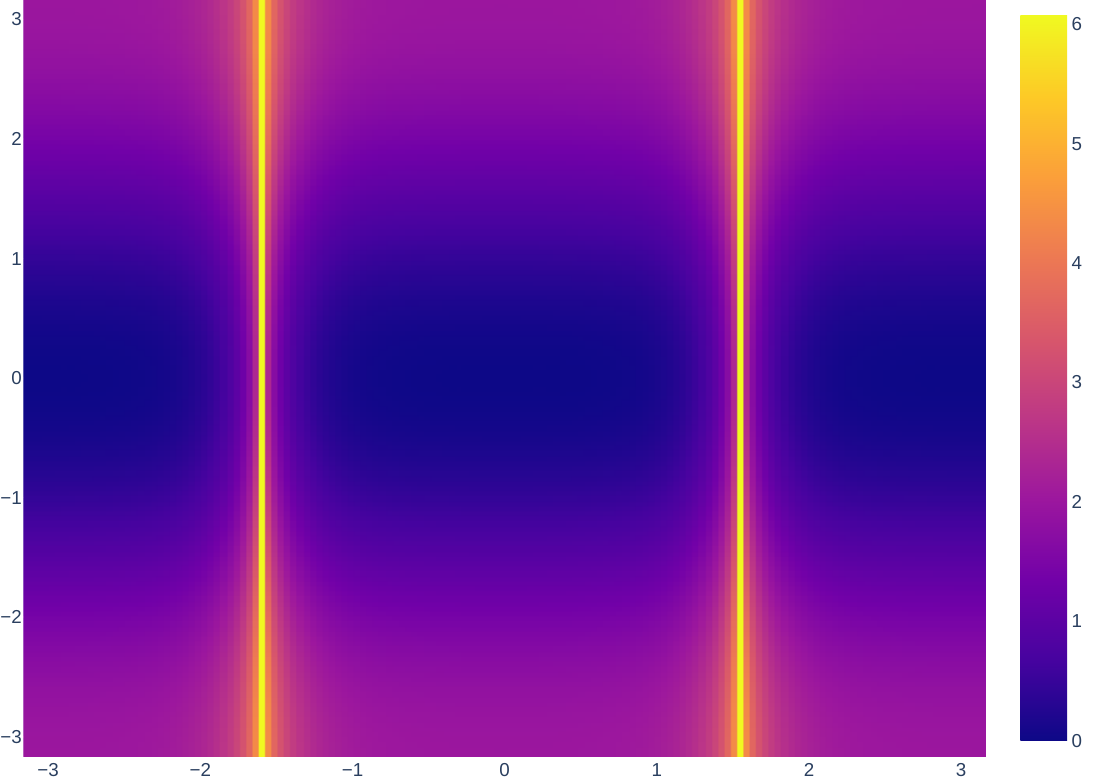}
    \caption{Optimal potential~$V+U_*$ for~\cref{example:2d_metastable}.}
    \label{fig:2d_metastable}
\end{figure}

The numerical values of the asymptotic variances for the examples considered in this section and different choices of~$U$ are collated in~\cref{tab:asym_vars_2d}.
The asymptotic variances corresponding to the simple biasing with~$U = - \theta V$ with optimal~$\theta$ are also presented.
This approach is found to perform quite well in the multimodal setting of~\cref{example:2d_metastable}.

\begin{table}[ht]
    \centering
    \begin{tabular}{|c|c|c|c|c|}
         \hline
         Test case & $U = 0$ & $U = -V$ & $U = - \theta_* V$ & Optimal $U$
         \\ \hline
         \cref{example:2d_first_example} & $1$ & $1$ & $1$ & $0.811$
         \\ \hline
         \cref{example:2d_difficult} & $1$ & $0.997$ & $0.987$ ($\theta_* = 0.614$) & $0.804$
         \\ \hline
         \cref{example:2d_metastable} & $1$ & $0.177$ &$0.177$ ($\theta_* = 0.994$) & $0.132$
         \\ \hline
         % \cref{example:1d_metastable} & $1$ & $0.00113$ & $0.00105$
         % \\ \hline
    \end{tabular}
    \caption{%
        Ratio of the asymptotic variance of the importance sampling estimator~\eqref{eq:estimator} relative to its value in the case where $U = 0$,
        for different choices of the biasing potential~$U$ and in the two-dimensional setting.
        The parameter $\theta_*$ is the minimizer of $\sigma^2_f[-\theta_* V]$,
        calculated by using the \texttt{Optim.jl}~\cite{Optim.jl-2018} implementation of the LBFGS solver.
        The numbers are rounded to 3 significant digits.
    }
    \label{tab:asym_vars_2d}%
\end{table}

\begin{remark}
    In all the examples presented in this subsection,
    the optimal biasing potential effectively partitions the domain into several regions that suffice for the estimation of~$I$.
    It is natural to wonder whether this observation holds true in general:
    is it \emph{always} the case that,
    when such partitioning of the domain occurs,
    averages of the observable with respect to the corresponding conditioned measures coincide with the target average~$I$?
    We gave in~\cref{remark:domain_divided} a positive answer to this question in the one-dimensional setting.
    Although we are not able to provide an equally rigorous answer in the multi-dimensional setting,
    we motivate hereafter our belief that the answer is also positive in this case.
    To this end, suppose that~$U_*$ partitions the domain into a number of regions,
    corresponding to the connected components of $\{U_* < \infty\}$.
    Suppose also that there exists an ensemble~$(U_{\varepsilon})_{\varepsilon > 0}$ of smooth biasing potentials such that
    $\sigma^2_f[U_{\varepsilon}] \to \sigma^2_f[U_*]$ and $\e^{-V-U_{\varepsilon}} \to \e^{-V-U_*}$ in $L^{\infty}(\domain^d)$ as $\varepsilon \to 0$.
    In particular, it holds under this assumption that $U_{\varepsilon}(x) \to U_*(x)$ for almost all $x \in \{U_* < \infty\}$.
    It is well known, see e.g.~\cite[Section 7.3]{MR3288096},
    that the average escape time from a potential well for the overdamped Langevin dynamics~\eqref{eq:mix} scales exponentially with respect to the height of the potential barrier.
    Therefore, for very small~$\varepsilon$,
    it would take a very long time for the dynamics to visit all the regions of the state space.
    In these conditions,
    the asymptotic variance would be very large,
    unless the averages of the observable with respect to the probability measure~$\mu$ conditioned to each of the regions happen to coincide.
    More precisely, if~$\sigma^2_f[U_{\varepsilon}]$ does not diverge as~$\varepsilon \to 0$,
    then the conditional averages in all the regions must necessarily coincide.
\end{remark}

\subsection{Minimizing the asymptotic variance for a class of observables}
\label{sub:class_examples}
In this section,
we illustrate the approach proposed in~\cref{sec:minimizing_the_expected_asymptotic_variance},
first for a one-dimensional example and then for a two-dimensional example.

\begin{example}
    \label{example:1d_class}
    We consider the same potential as in~\cref{example:1d_metastable},
    i.e.\ $V(x) = 5 \cos(2x)$,
    and a set-up similar to that of~\cref{proposition:when_is_free_energy_biasing_optimal}.
    Specifically, the observables
    and associated weights,
    denoted by~$(\lambda_j)_{1 \leq j \leq J}$ in~\eqref{eq:random_observable},
    are given by the first $J = 21$ eigenpairs of the operator $(-\laplacian + \mathcal I)^{-1}$,
    equipped with periodic boundary conditions on the space of mean-zero functions with respect to the Lebesgue measure.
    The associated Gaussian random field $f$ is stationary,
    in the sense that the covariance ${\rm cov}\bigl(f(x_1), f(x_2)\bigr)$ depends only on the difference~$x_1 - x_2$.
    The optimal potential in this case is illustrated in the left panel of~\cref{fig:1d_class}.
    In contrast with the examples of~\cref{sub:one_dim_examples},
    the optimal potential is smooth and, therefore, more easily usable in an MCMC scheme.
    The average asymptotic variance~$\sigma^2[U]$ given in~\eqref{eq:expectation_asymptotic_variance} is reduced by a factor equal to about~900.

    In the right-panel of~\cref{fig:1d_class},
    we illustrate the optimal potential when the observables are instead the eigenfunctions
    of $\e^{V} (-\laplacian + \mathcal I)^{-1} \e^{-V}$ equipped with periodic boundary conditions,
    which is precisely the setting considered in~\cref{proposition:when_is_free_energy_biasing_optimal}.
    In this case, the optimal potential is indeed $V+U = 0$,
    in agreement with the latter result.
    The average asymptotic variance~$\sigma^2[U]$ is reduced by a factor equal to about~700.
\end{example}

\begin{figure}[ht]
    \centering
    \includegraphics[width=0.49\linewidth]{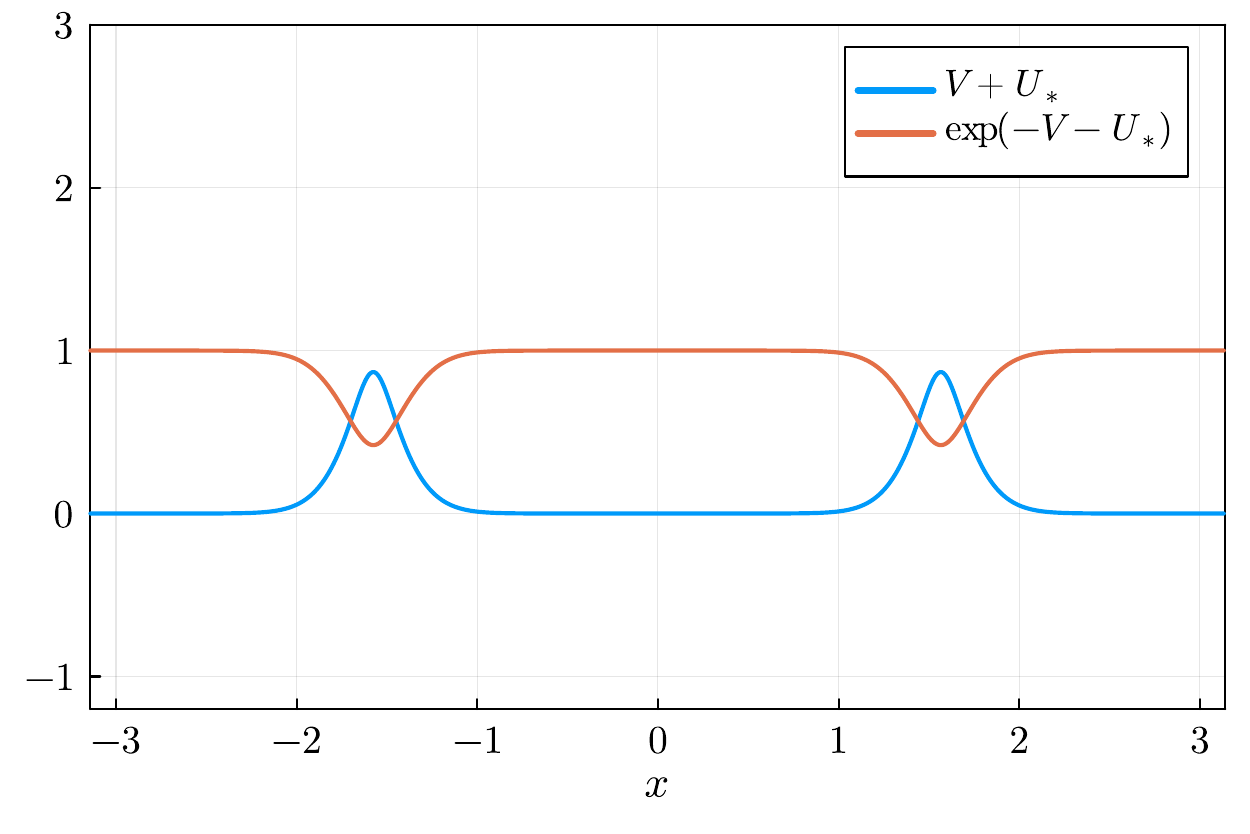}
    \includegraphics[width=0.49\linewidth]{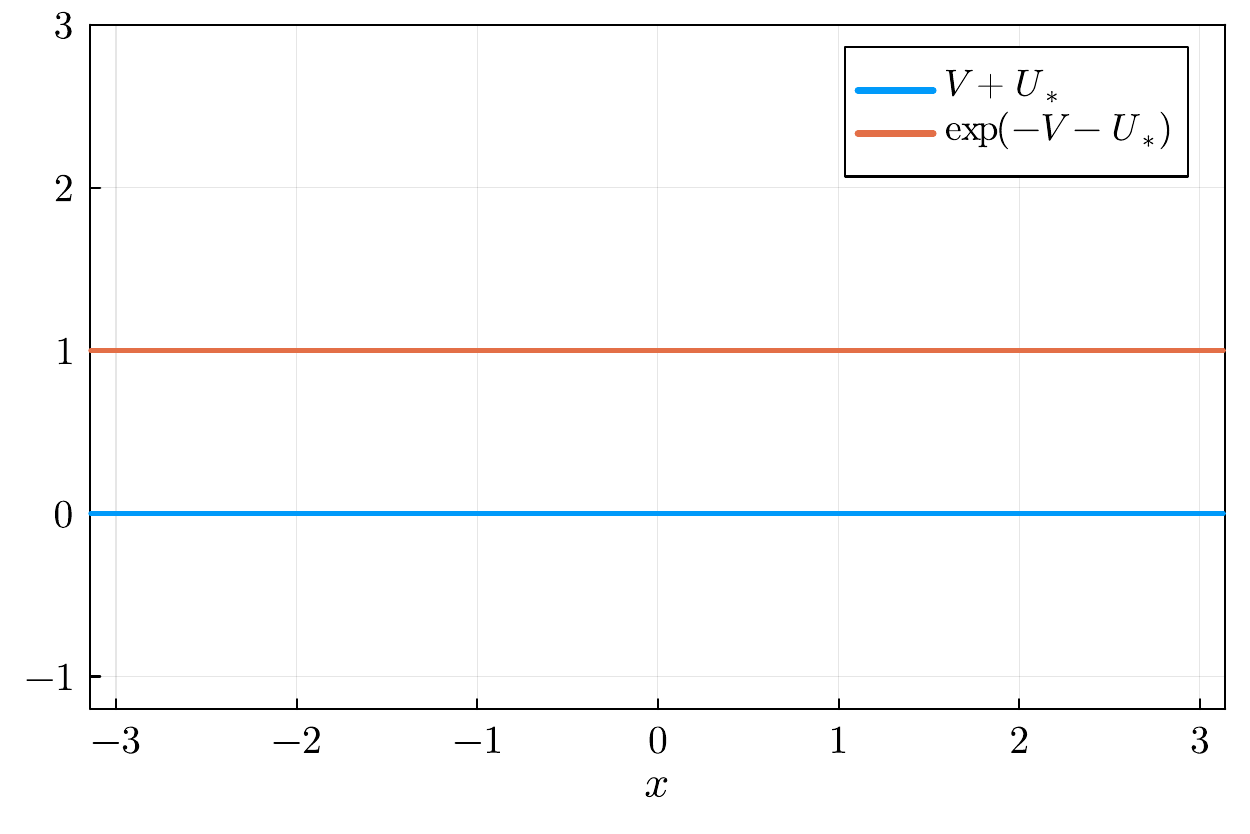}
    \caption{%
        Optimal potentials for~\cref{example:1d_class},
        when different probability measures are placed on the observables.
    }
    \label{fig:1d_class}
\end{figure}

\begin{example}
    \label{example:2d_class}
    We consider the same potential as in~\cref{example:2d_metastable},
    i.e.\ $V(x) = 2 \cos(2 x_1) - \cos(x_2)$.
    For the observables and corresponding weights in~\eqref{eq:random_observable},
    we take the eigenpairs of the operator~$(-\laplacian + \mathcal I)^{-1}$ with periodic boundary conditions on the space of mean-zero functions with respect to the Lebesgue measure.
    The eigenfunctions are of the form
    \[
        \cos(m x_1) \cos(n x_2),
        \qquad
        \cos(m x_1) \sin(n x_2),
        \qquad
        \sin(m x_1) \cos(n x_2),
        \qquad
        \sin(m x_1) \sin(n x_2).
    \]
    We consider all the eigenpairs with $m \leq 4$ and $n \leq 4$.
    The optimal potential~$V + U_*$ in this case is depicted in~\cref{fig:2d_class},
    together with the initial potential~$V$.
    We observe that the potential has been flattened in the direction~$x$.
    The resulting reduction in the average asymptotic variance is about $70\%$.
\end{example}

\begin{figure}[ht]
    \centering
    \includegraphics[width=0.49\linewidth]{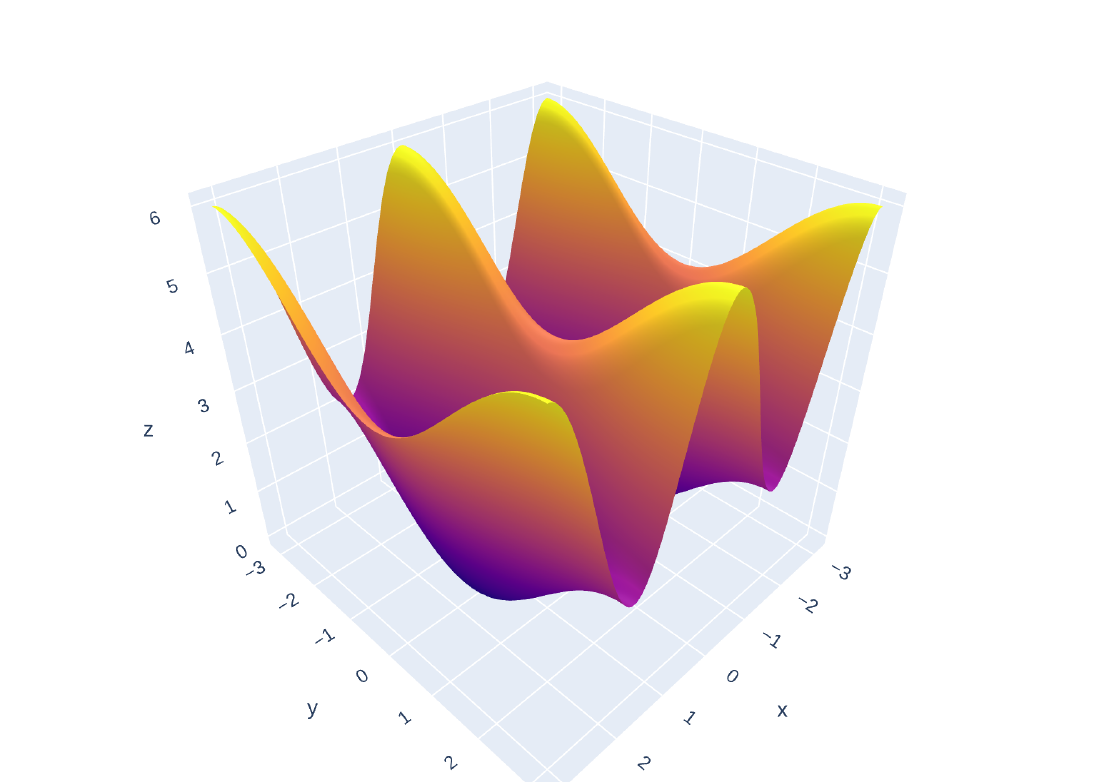}
    \includegraphics[width=0.49\linewidth]{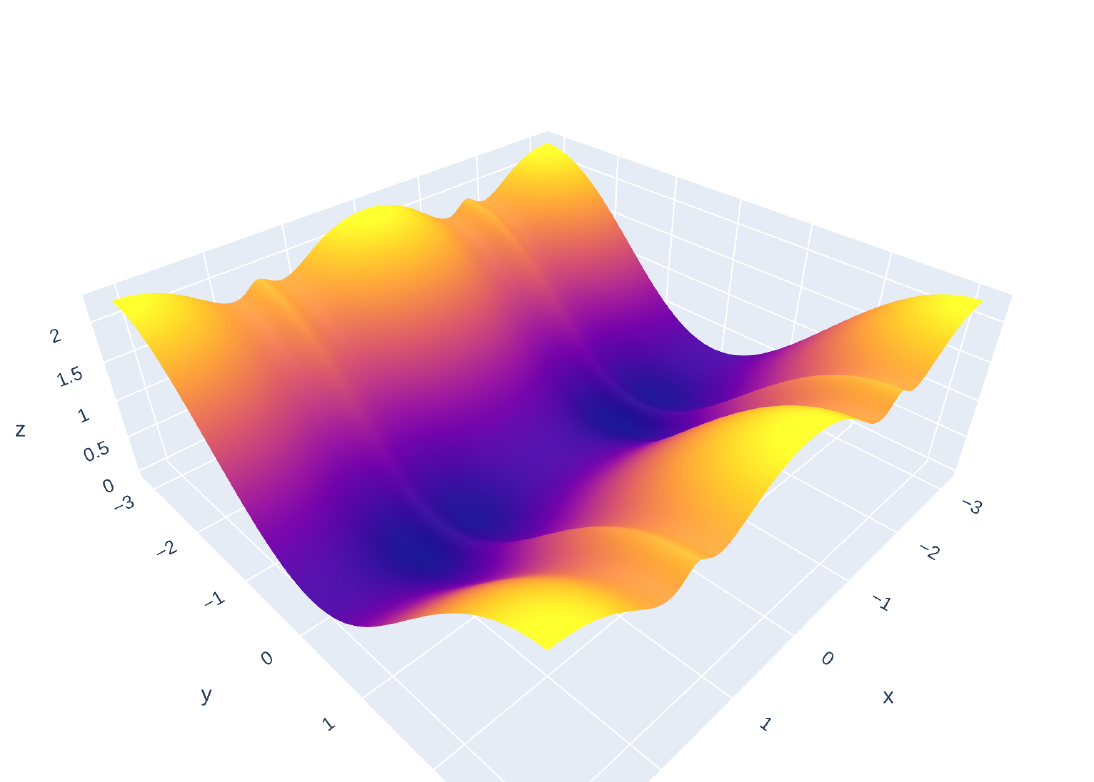}
    \caption{%
        Potential~$V$ (left) and optimal potential~$V+U_*$ (right) corresponding to~\cref{example:2d_class}.
    }
    \label{fig:2d_class}
\end{figure}

\section{Conclusions and perspectives for future works}
\label{sec:conclusion}
In this work,
we considered an importance sampling method based on the overdamped Langevin dynamics in a perturbed potential
and present a novel approach for constructing the biasing potential.
Under appropriate assumptions,
this potential is optimal,
in the sense that it leads to the minimum asymptotic variance
when employed for calculating the average of one or a class of observables with respect to the target probability measure.
The optimal biasing potential is explicit in dimension 1,
and may be approximated by steepest descent in higher dimensions.

We demonstrated the performance of the method by means of numerical experiments in dimensions~1 and~2.
In the multimodal setting, in particular,
using the optimal importance distribution enables a considerable reduction in asymptotic variance.
Finally, our numerical experiments show that,
while minimizing the asymptotic variance for just one observable leads to singularities in the potential,
targeting a number of observables simultaneously leads to smooth potentials which can more easily be employed in numerical schemes.

A drawback of the proposed methodology is that the construction of the optimal biasing potential relies on an iterative method which,
at each step, requires the solution of a Poisson equation.
While feasible in low dimension,
this approach is computationally too costly in a high-dimensional setting.
A possible approach in this case is to
reduce the dimension of the problem by requiring that the biasing potential
is a function of only a few well-chosen degrees of freedom (so-called collective variables),
which ideally capture the metastable behavior of the dynamics.
This corresponds to the setting of free energy computation~\cite{MR2681239},
and suggests to consider the variance as a functional of some free energy,
which particularly makes sense when the observable under investigation itself depends only on the collective variables.
Investigation of this approach will be the subject of future work.
% Another direction for future work would be to extend our analysis to obtain a directional derivative directly from the asymptotic variance associated to the discrete-time estimator.
Another direction for future work would be to investigate whether a similar approach can be employed to minimize the asymptotic variance of estimators based on discrete-time MCMC schemes using overdamped Langevin dynamics.
We expect the resulting optimal biasing potentials in that case to be close to those considered here.

\appendix

\section{Technical auxiliary results}
\label{sec:auxiliary_results}
In this section,
we collect technical auxiliary results used in~\cref{sub:iid,sec:minimizing_the_asymptotic_variance_for_one_observable}.

\begin{example}
    \label{example:iid_asym_var}
    Consider the setting where~$\domain = \torus$ in dimension $d=1$ and~$V = 0$,
    with the observable~$f\colon [-\pi, \pi] \to \real$ given by
    \[
        f(x) =
        \begin{cases}
            \sgn(x) \qquad &\text{if $\abs{x} \geq \frac{\pi}{2}$}, \\
            0 \qquad &\text{otherwise}.
        \end{cases}
    \]
    where the $\sgn$ function is defined in~\eqref{eq:not_asymvar}.
    Here we identify~$[-\pi, \pi]$ with its image under the quotient map~$\real \to \torus$.
    In this case $I = 0$ and we have the following:
    \begin{itemize}
        \item
            If $U$ is such that $\e^{-U} = \mathds 1_{[-\pi/4, \pi/4]}$,
            where $\mathds 1_S$ is the indicator function of the set~$S$,
            then it holds that $f(X^n) = I$ and $\e^U(X^n)$ with probability 1 for $X^n \sim \mu_U$,
            and so
            \[
                s^2_f[U] = 0.
            \]
            This is in agreement with the first item in~\cref{proposition:background_infimum}.
            In this particular case, 0 is not only the infimum but also the minimum of the asymptotic variance over~$\mathcal U$.

        \item
            The variance in~\eqref{eq:def_infimum_iid} is given by
            \[
                s^*_f := \left( \frac{1}{2\pi} \int_{\torus} \abs{f-I} \right)^2 = \frac{1}{4}.
            \]

        \item
            The potential~$U_*^{\rm iid}$ in~\eqref{eq:minimizer_lemma_iid} is given by $U_*^{\rm iid} = - \log\lvert f \rvert$.
            If $X^n \sim \mu_{U_*^{\rm iid}}$, then the random variable~$f(X^n)$ is equal to either $-1$ and $1$,
            each with probability $1/2$,
            and the random variable~$(\e^U)(X^n)$ is equal to~1 almost surely.
            Therefore, the associated asymptotic variance is given by
            \[
                s^2_f[U_*^{\rm iid}] = 1.
            \]
            This equation can also be obtained from~\eqref{eq:asym_var_iid}.
            We observe that~$s^2_f[U_*^{\rm iid}] > s^*_f$,
            which is consistent with the discussion in~\cref{remark:optimal_idd_not_minimizer}.

        \item
            Let $U_{\varepsilon} := - \log \left( \lvert f \rvert + \varepsilon \right)$,
            which may be viewed as a discontinuous but bounded regularization of~$U_*^{\rm iid}$.
            Then, for $X^n \sim \mu_{U_{\varepsilon}}$,
            using the notation w.~p.~to mean ``with probability'',
            we have that
            \[
                \left(f \e^{U_{\varepsilon}}\right)(X^n) =
                \begin{cases}
                    \frac{1}{1 + \varepsilon} \quad & \text{ w.~p.~ $\frac{1 + \varepsilon}{2 + 4 \varepsilon}$}\\
                    0 \quad & \text{ w.~p.~ $\frac{2\varepsilon}{2 + 4 \varepsilon}$}\\
                    - \frac{1}{1 + \varepsilon}\quad & \text{ w.~p.~ $\frac{1 + \varepsilon}{2 + 4 \varepsilon}$}
                \end{cases},
                \qquad
                \left(\e^{U_{\varepsilon}}\right)(X^n) =
                \begin{cases}
                    \frac{1}{1 + \varepsilon} \quad & \text{ w.~p.~ $\frac{1 + \varepsilon}{2 + 4 \varepsilon}$}\\
                    \frac{1}{\varepsilon} \quad & \text{ w.~p.~ $\frac{2\varepsilon}{2 + 4 \varepsilon}$}\\
                    \frac{1}{1 + \varepsilon}\quad & \text{ w.~p.~ $\frac{1 + \varepsilon}{2 + 4 \varepsilon}$}
                \end{cases}.
            \]
            It follows that the variance of $\left(f \e^{U_{\varepsilon}}\right)(X^n)$ is given by
            \[
                \frac{1 + \varepsilon}{1 + 2 \varepsilon} \left(\frac{1}{1 + \varepsilon}\right)^2 = \frac{1}{(1 + 2 \varepsilon)(1 + \varepsilon)},
            \]
            and that $\expect\bigl[\e^{U_\varepsilon}(X^n)\bigr] = \frac{2}{1 + 2 \varepsilon}$.
            Therefore, by Slutsky's lemma,
            or from Equation~\eqref{eq:asym_var_iid},
            we obtain
            that
            \[
                s^2_f[U_{\varepsilon}] = \frac{1 + 2\varepsilon}{4 + 4 \varepsilon}.
            \]
            We observe that $s^2_f[U_{\varepsilon}] \to s^*_f$ in the limit as~$\varepsilon \to 0$.
    \end{itemize}
    This example shows that the biasing potential $U_*^{\rm iid}$ is sometimes suboptimal in~$\mathcal U_0$;
    here we constructed a regularized biasing potential associated with a smaller asymptotic variance than that associated with~$U_*^{\rm iid}$.
    Furthermore, this example illustrates that the quantity $s^*_f$ is not in general a lower bound on the asymptotic variance
    over the set of biasing potentials in~$\mathcal U$.
\end{example}

\begin{lemma}
    [Asymptotic variance for the estimator given in~\eqref{eq:semidiscrete_estimator}]
    \label{sec:discrete_asym_var}
    Suppose that~\cref{assumption:as1} is satisfied, that $\domain = \torus$ and that $X_0 \sim \mu_U$.
    Then there exists a unique solution~$\widetilde \phi_U$ in~$L^2_0(\mu_{U})$ to~\eqref{eq:discrete_poisson}
    and it holds that
    \[
        \sqrt{N} \bigl( \widetilde \mu^N_U(f) - I\bigr)
        \xrightarrow[N \to \infty]{\rm Law} \normal\left(0, \widetilde \sigma^2_f[U]\right),
    \]
    where~$\widetilde \sigma^2_f[U]$ is given by~\eqref{eq:sigma_subsampled}.
\end{lemma}
\begin{remark}
    The assumptions that $\domain = \torus$ and $X_0 \sim \mu$ should be viewed as technical;
    it should in principle be possible to relax them.
\end{remark}
\begin{proof}
    We begin by showing the existence and uniqueness of a solution in~$L^2_0(\mu_U)$ to the Poisson equation~\eqref{eq:discrete_poisson}.
    To this end,
    we recall that, under~\cref{assumption:as1},
    \begin{equation}\label{eq:expbound}
        \norm{\e^{t \mathcal L_U}}_{\mathcal B\left(L^2_0(\mu_U)\right)} \leq \e^{- R[U] t},
    \end{equation}
    where $\mathcal B\left(L^2_0(\mu_U)\right)$ is the Banach space of continuous linear operators on $L^2_0(\mu_U)$ and $R[U]$ is the Poincaré constant in~\eqref{eq:poincaré} associated with $\mu_U$;
    see e.g.~\cite[Proposition~2.3]{MR3509213} and~\cite[Theorem~4.4]{MR3288096}.
    Therefore,
    the Neumann series $\sum_{n=0}^{\infty} \e^{n\tau \mathcal L_U}$ is convergent in $L^2_0(\mu_U)$,
    which implies that $\mathcal I - \e^{\tau \mathcal L_U}$ is invertible with inverse~$(\mathcal I - \e^{\tau \mathcal L_U})^{-1}$ equal to the series.
    Therefore,
    there exists a unique solution~$\widetilde \phi_U \in L^2_0(\mu_U)$ to~\eqref{eq:discrete_poisson}.
    The estimator~\eqref{eq:iid_importance_sampling} may be rewritten as
    \begin{equation}
        \label{eq:discrete_importance_sampling}
        \widetilde \mu_U^N(f)
        = I + \displaystyle \frac
        {\sum_{n=0}^{N-1} g(X_{n \tau})}
        {\sum_{n=0}^{N-1} (\e^U)(X_{n \tau})},
        \qquad g := (f-I) \e^U.
    \end{equation}
    The key idea
    in order to understand the asymptotic behavior of the numerator in~\eqref{eq:discrete_importance_sampling} is to rewrite the sum as
    \begin{align*}
        \sum_{n=0}^{N-1} g(X_{n \tau})
        &= \sum_{n=0}^{N-1} \Bigl( \left( \mathcal I - \e^{\tau \mathcal L_U} \right) \widetilde \phi_U \Bigr)\left(X_{n \tau}\right) \\
        &= \sum_{n=0}^{N-1}\left( \widetilde \phi_U\left(X_{(n+1)\tau}\right) - \e^{\tau \mathcal L_U} \widetilde \phi_U\left(X_{n\tau}\right) \right)
    -  \widetilde \phi_U\left(X_{N\tau}\right) + \widetilde \phi_U(X_0).
    \end{align*}
    This approach dates back to the work of Kipnis and Varadhan~\cite{MR834478}.
    The first term is a sum of uncorrelated, identically distributed random variables with mean zero and variance
    \begin{align*}
        \gamma^2_f[U] :=
        \expect \left(  \left\lvert \widetilde \phi_U\left(X_{\tau}\right) - \e^{\tau \mathcal L_U} \widetilde \phi_U\left(X_0\right) \right\rvert^2  \right)
        &= \int_{\domain^d} \left(\left( \e^{\tau \mathcal L_U}  \left\lvert  \widetilde \phi_U\right\rvert^2 \right)(x) - \left\lvert \left( \e^{\tau \mathcal L_U} \widetilde \phi_U \right)(x)  \right\rvert^2\right) \, \mu_U(\d x) \\
        &= \int_{\domain^d}  \left( \left\lvert  \widetilde \phi_U(x)\right\rvert^2  - \left\lvert \left( \e^{\tau \mathcal L_U} \widetilde \phi_U \right)(x)  \right\rvert^2 \right)\, \mu_U(\d x),
    \end{align*}
    where we used the invariance of $\mu_U$ by the dynamics with generator~$\mathcal L_U$ for the first term in the last integral.
    Since $\widetilde \phi_U$ is a solution to~\eqref{eq:discrete_poisson},
    it holds that $\e^{\tau \mathcal L_U} \widetilde \phi_U = \widetilde \phi_U - g$,
    which by substitution gives that
    \[
        \gamma^2_f[U] = \int_{\domain^d} 2\widetilde \phi_U g - g^2 \, \d \mu_U.
    \]
    Using an approach similar to that in~\cite[Theorem 17.4.4]{MR2509253},
    we can show that the conditions of the martingale central limit theorem~\cite{MR668684}
    (see also~\cite{major} for a detailed pedagogical proof) are satisfied,
    and so it holds that
    \[
        \frac{1}{\sqrt{N}} \sum_{n=0}^{N-1} \left( \widetilde \phi_U\left(X_{(n+1)\tau}\right) - \e^{\tau \mathcal L_U} \widetilde \phi_U\left(X_{n\tau}\right) \right) \xrightarrow[N \to \infty]{\rm Law} \mathcal N\left(0, \gamma^2_f[U]\right).
    \]
    Hence, since it is clear in the setting where $\domain = \torus$ that
    \[
        \frac{1}{\sqrt{N}} \left( \widetilde \phi_U\left(X_{N\tau}\right) -  \widetilde \phi_U(X_0) \right)
        \xrightarrow[N \to \infty] {{\rm Law}} 0,
        % \rightarrow % M: Not sure this made sense, what's the variable of integration?
    \]
    % in mean as~$N\rightarrow\infty$,
    it follows from Slutsky's lemma that
    \[
        \frac{1}{\sqrt{N}} \sum_{n=0}^{N-1} g(X_{n \tau}) \xrightarrow[N \to \infty]{\rm Law} \mathcal N\left(0, \gamma^2_f[U]\right).
    \]
    A similar approach, based on the Poisson equation
    \[
        - \widetilde{\mathcal L}_U \widetilde \psi_U = \left( \e^U - \frac{Z}{\Z{U}} \right),
    \]
    of which the right-hand side is in $L^2_0(\mu_U)$ by the assumption that~$\domain = \torus$,
    can be employed to understand the asymptotic behavior of the denominator in~\eqref{eq:discrete_importance_sampling}.
    Specifically, it holds that
    \begin{align*}
        \frac{1}{N} \sum_{n=0}^{N-1} \left(\e^{U(X_{n \tau})} - \frac{Z}{\Z{U}}\right)
        &= \frac{1}{N} \sum_{n=0}^{N-1}\left( \widetilde \psi_U\left(X_{(n+1)\tau}\right) - \e^{\tau \mathcal L_U} \widetilde \psi_U\left(X_{n\tau}\right) \right)
        - \frac{\widetilde \psi_U\left(X_{N\tau}\right) + \widetilde \psi_U(X_0)}{N}.
    \end{align*}
    An explicit calculation,
    using that the first term on the right-hand side is a sum of uncorrelated, identically distributed random variables,
    gives that the variance of the right-hand side converges to 0 in the limit as~$N \to \infty$,
    implying the convergence
    \[
        \frac{1}{N} \sum_{n=0}^{N-1} \e^{U(X_{n \tau})} \xrightarrow[N \to \infty]{\rm Law} \frac{Z}{\Z{U}}.
    \]
    The proof can then be concluded by using Slutsky's lemma once more.
\end{proof}

\begin{lemma}
    [Solution to the perturbed Poisson equation]
    \label{lemma:stability_gradient_sol_poisson}
    Suppose that~\cref{assumption:as1} is satisfied and that $\delta U \in \smoothcompact(\domain^d)$,
    and let $\phi_{U + \varepsilon \delta U}$ denote the solution to the Poisson equation~\eqref{eq:perturbed_poisson} posed in $L^2_0(\e^{-V-U-\varepsilon \delta U})$.
    Then
    \begin{equation}
        \label{eq:statement_perturbation_poisson}
        \frac{\nabla \phi_{U + \varepsilon \delta U} - \nabla \phi_U}{\varepsilon} \xrightarrow[\varepsilon \to 0]{}  \nabla \psi_{U,\delta U}~\text{in $L^2(\mu_U)$},
    \end{equation}
    where~$\psi_{U,\delta U}$ denotes the unique solution in~$H^1(\mu_{U}) \cap L^2_0(\mu_{U})$ to
    \begin{align}
        \nonumber
        -\mathcal L_{U} \psi_{U,\delta U}
        &=  (f-I) \e^U \delta U - \nabla (\delta U) \cdot \nabla \phi_U \\
        \label{eq:pert2}
        &= - \e^{U+V} \nabla \cdot \left(\e^{-U-V}\delta U \grad \phi_U\right).
    \end{align}
\end{lemma}
\begin{remark}
    By integration by parts,
    which is allowed since $\phi_U \in C^{\infty}(\domain^d)$ and $\delta U \in \smoothcompact(\domain^d)$,
    we can check that the right-hand side of~\eqref{eq:pert2} is indeed mean zero with respect to~$\mu_U$:
    \begin{align*}
        - \int_{\domain^d} \e^{U+V} \nabla \cdot \left(\e^{-U-V}\delta U \grad \phi_U\right) \, \d \mu_U
        = \frac{1}{\Z{U}}\int_{\domain^d} \nabla \cdot \left(\e^{-U-V}\delta U \grad \phi_U\right) \, \d x
        = 0.
    \end{align*}
    Therefore, there indeed exists a unique distributional solution in~$H^1(\mu_{U}) \cap L^2_0(\mu_{U})$ to~\eqref{eq:pert2} by the Lax--Milgram theorem.
\end{remark}
\begin{proof}
    [Proof of \cref{lemma:stability_gradient_sol_poisson}]
    Between the Poisson equations~\eqref{eq:poisson} and~\eqref{eq:perturbed_poisson},
    both the operator and the right-hand side differ.
    % Additionally, the operators have different ranges,
    % which precludes a direct approach by Neumann series.
    We begin by rewriting
    \begin{equation}
        \label{eq:perturbation_generator}
        \mathcal L_{U + \varepsilon \delta U} = \mathcal L_U - \varepsilon \nabla (\delta U) \cdot \nabla.
    \end{equation}
    Let $\psi_{\varepsilon} = \varepsilon^{-1} (\phi_{U+\varepsilon \delta U} - \phi_U)$.
    It holds that
    \begin{equation}
        \label{eq:pert1}
        -\mathcal L_{U + \varepsilon \delta U} \psi_{\varepsilon}
        =  (f-I) \frac{\e^{U + \varepsilon \delta U} - \e^{U}}{\varepsilon} - \nabla (\delta U) \cdot \nabla \phi_U.
    \end{equation}
    The right-hand side is mean zero with respect to~$\mu_{U+\varepsilon \delta U}$ by construction,
    and so by the Lax--Milgram theorem there exists a unique distributional solution in~$H^1(\mu_{U+\varepsilon \delta U}) \cap L^2_0(\mu_{U+\varepsilon \delta U})$ to~\eqref{eq:pert1},
    which coincides with~$\psi_{\varepsilon}$ up to an additive constant.
    Subtracting~\eqref{eq:pert2} from~\eqref{eq:pert1},
    we deduce that
    \[
        -\mathcal L_{U + \varepsilon \delta U} (\psi_{\varepsilon}  - \psi_{U,\delta U})
        = - (\mathcal L_U - \mathcal L_{U + \varepsilon \delta U}) \psi_{U,\delta U}
        + (f-I) \left(\frac{\e^{U + \varepsilon \delta U} - \e^{U}}{\varepsilon} - \e^U \delta U \right)
        =: \zeta_{\varepsilon}.
    \]
    % The right-hand side of this eis mean zero with respect to~$\mu_{U+\varepsilon \delta U}$ by construction.
    The second term on the right-hand side converges to 0 in $L^2(\mu_U)$ in the limit as $\varepsilon \to 0$,
    as does the first term in view of~\eqref{eq:perturbation_generator}.
    By the Holley--Stroock theorem,
    the probability measure~$\mu_{U+\varepsilon \delta U}$ satisfies the Poincaré inequality~\eqref{eq:poincaré} with a constant $R[U+\varepsilon \delta U]$
    that converges to~$R[U]$ in the limit~$\varepsilon \to 0$.
    Consequently, we deduce from the standard stability estimate~\eqref{eq:stability_estimate} that
    \[
        \norm{\nabla (\psi_{\varepsilon} - \psi_{U,\delta U})}_{L^2(\mu_{U+\varepsilon \delta U})}
        \leq \frac{\norm{\zeta_{\varepsilon}}_{L^2(\mu_{U+\varepsilon \delta U})}}{R[U+\varepsilon \delta U]}.
    \]
    Since the right-hand side converges to 0 in the limit~$\varepsilon \to 0$,
    so must the left-hand side,
    which leads to the convergence $\norm{\nabla (\psi_{\varepsilon} - \psi_{U,\delta U})}_{L^2(\mu_{U})} \to 0$
    given the equivalence between the norms of~$L^2(\mu_U)$ and~$L^2(\mu_{U+\varepsilon \delta U})$.
    This concludes the proof.
    % since $\delta U$ is bounded, these function spaces coincide.
    % It remains to show that~$\psi_{\varepsilon} \to \psi_{U,\delta U}$ in $L^2(\mu_U)$ as $\varepsilon \to 0$.
    % Using the triangle and Poincaré inequalities,
    % we obtain that
    % \[
    %     \norm{\psi_{\varepsilon} - \psi_{U,\delta U}}_{L^2(\mu_{U})}
    %     \leq \left\lvert \int_{\domain^d} (\psi_{\varepsilon} - \psi_{U,\delta U}) \, \d \mu_{U} \right\rvert
    %     + \frac{\norm{\nabla (\psi_{\varepsilon} - \psi_{U,\delta U})}_{L^2(\mu_{U})}}{R[U]} .
    % \]
    % We already showed that the second term on the right-hand side converges to 0,
    % and we conclude the proof by noticing that
    % \begin{align*}
    %      \int_{\domain^d} (\psi_{\varepsilon} - \psi_{U,\delta U}) \, \d \mu_{U}
    %     &=  \frac{1}{\varepsilon} \int_{\domain^d} \bigl(\phi_{U+\varepsilon \delta U} - \phi\bigr)\d \mu_U  \\
    %     &= \frac{\Z{U}}{\Z{ U + \varepsilon \delta U }}  \int_{\domain^d} \psi_{U,\delta U} \left(1  - \e^{-\varepsilon \delta U} \right)  \d \mu_U \xrightarrow[\varepsilon \to 0]{} 0,
    % \end{align*}
    % where we used the fact that~$\psi_{U,\delta U}$ and~$\psi_{\varepsilon}$ are mean-zero with respect to~$\e^{-V-U}$ and~$\e^{-V-U-\varepsilon \delta U}$,
    % respectively.
\end{proof}
\begin{remark}
    One may wonder whether the statement~\eqref{eq:statement_perturbation_poisson} can be strengthened to
    \begin{equation}
        \label{eq:not_statement_perturbation_poisson}
        \frac{\phi_{U + \varepsilon \delta U} - \phi_U}{\varepsilon} \xrightarrow[\varepsilon \to 0]{} \psi_{U,\delta U}~\text{in $H^1(\mu_U)$}.
    \end{equation}
    The answer to this question is negative.
    Indeed, assume by contradiction that~\eqref{eq:not_statement_perturbation_poisson} holds.
    Then in particular~$\phi_{U + \varepsilon \delta U} - \phi_U \to 0$ in $L^2(\mu_U)$ in the limit as $\varepsilon \to 0$ and so
    \begin{align}
        \notag
        \int_{\domain^d} \frac{\phi_{U + \varepsilon \delta U} - \phi_U}{\varepsilon} \, \d \mu_{U}
        &= \frac{1}{\varepsilon} \int_{\domain^d} \phi_{U + \varepsilon \delta U} \d \mu_{U} \\
        \notag
        &= \frac{1}{Z[U]} \int_{\domain^d} \phi_{U+\varepsilon \delta U} \left( \frac{\e^{-V-U} - \e^{-V-U- \varepsilon \delta U}}{\varepsilon} \right) \\
        \label{eq:not_l2}
        &\xrightarrow[\varepsilon \to 0]{} \int_{\domain^d} \phi_U \delta U \d \mu_U,
    \end{align}
    where we used that~$\phi_{U}$ and~$\phi_{U + \varepsilon \delta U}$ are mean-zero with respect to~$\e^{-V-U}$ and~$\e^{-V-U-\varepsilon \delta U}$,
    respectively.
    This is a contradiction because~\eqref{eq:not_statement_perturbation_poisson} implies that
    \[
        \int_{\domain^d} \frac{\phi_{U + \varepsilon \delta U} - \phi_U}{\varepsilon}  \d \mu_{U}
        \xrightarrow[\varepsilon \to 0]{} \int_{\domain^d} \psi_{U,\delta U} \d \mu_{U} = 0,
    \]
    and so~\eqref{eq:not_statement_perturbation_poisson} does not hold.

    It is, however, simple to show that ~$\phi_{U + \varepsilon \delta U} \to \phi_U$~in $L^2(\mu_U)$ in the limit as $\varepsilon \to 0$.
    Additionally, it holds by~\cref{lemma:stability_gradient_sol_poisson} and the Poincaré inequality that
    \[
        \frac{\phi_{U + \varepsilon \delta U} - \phi_U}{\varepsilon} - \int_{\domain^d} \frac{\phi_{U + \varepsilon \delta U} - \phi_U}{\varepsilon} \, \d \mu_U \xrightarrow[\varepsilon \to 0]{} \psi_{U,\delta U}~\text{in $H^1(\mu_U)$},
    \]
    but these statements are not useful for our purposes in this paper.
\end{remark}
\begin{remark}
    Since~$\psi_{U,\delta U}$ is a weak solution to~\eqref{eq:pert2},
    it holds for every $\delta W \in \smoothcompact(\domain^d)$ that
    \begin{align}
        \nonumber
        \int_{\domain^d} \nabla \psi_{U,\delta W} \cdot \nabla \psi_{U,\delta U} \, d \mu_U
        &= \int_{\domain^d} \psi_{U,\delta W} \Bigl( - \e^{U+V} \nabla \cdot \left(\e^{-U-V}\delta U \grad \phi_U\right) \Bigr)  \, \d \mu_U \\
        \label{eq:auxiliary_func_der}
        &= \int_{\domain^d} \delta U \grad \phi_U  \cdot \nabla \psi_{U,\delta W} \, \d \mu_U,
    \end{align}
    where integration by parts is justified because $\delta U \in \smoothcompact(\domain^d)$.
    This equality, where the roles of~$\delta U$ and $\delta W$ can be reversed,
    is useful in the proof of~\cref{proposition:second_variation_asymptotic_variance} below.
\end{remark}
\begin{remark}
    \label{remark:psi_in_dim_one}
    In dimension $d = 1$,
    it follows from~\eqref{eq:pert2} that
    \begin{equation}
        \label{eq:psi_prime}
        \psi_{U,\delta U}' = \delta U \phi_{U}' + C_{\domain}[U,\delta U] \e^{V+U},
    \end{equation}
    for some constant $C_{\domain}[U, \delta U]$ such that $\psi_{U,\delta U}' \in L^2(\mu_U)$.
    Clearly $C_{\real}[U, \delta U] = 0$.
    When $\domain = \torus$,
    we obtain the value of~$C_{\torus} [U,\delta U]$ by requiring periodicity,
    that is
    \[
        0 = \int_{\torus} \psi_{U,\delta U}'
        = \int_{\torus} \delta_U \phi_{U}' + C_{\torus}[U, \delta U] \int_{\torus} \e^{V+U},
    \]
    which leads to
    \[
        C_{\torus}[U, \delta U]
        = - \frac
        {\displaystyle \int_{\torus} \delta U \phi_{U}'}
        % {\Big/}
        {\displaystyle \int_{\torus}\e^{V+U}}.
    \]
    We note that this formula may also be obtained by considering the differential of~\eqref{eq:derivative_of_phi} viewed as a functional of~$U$,
    an approach which reveals that $C_{\domain}[U, \delta U] = \d A_{\domain}[U] \cdot \delta U$.
\end{remark}

\section{Connection between \texorpdfstring{$(A_{\real}, A^*_{\real})$ and $(A_{\torus}, A^*_{\torus})$}{the constants A}}
\label{remark:constant_A}

In this section,
we discuss the links between the constants defined in~\eqref{eq:A_domain} and~\eqref{eq:equation_A}.

\paragraph{Connection between~$A_{\real}$ and~$A_{\torus}$.}
The constant $A_{\real}^* = A_{\real}$ is recovered as a limit of~$A^*_{\torus}$ for an increasingly large torus.
More precisely, it holds that
\begin{equation}
    \label{eq:limit}
    A_{\real} =
    \lim_{L \to \infty}
    \frac
    {\displaystyle \int_{-L}^{L} F \, \e^{V+U}}
    % \Big/
    {\displaystyle \int_{-L}^{L} \e^{V+U}}.
\end{equation}
Indeed, for any $\ell > 0$ and $L > \ell$,
it holds that
\begin{align}
    \notag
    \frac
    {\displaystyle \int_{-L}^{L} (F-A_{\real}) \, \e^{V+U}}
    {\displaystyle \int_{-L}^{L} \e^{V+U}}
    =\, &
    \frac
    {\displaystyle \int_{[-L, -\ell) \cup (\ell,L]} (F-A_{\real}) \, \e^{V+U}}
    {\displaystyle \int_{[-L, -\ell) \cup (\ell,L]} \e^{V+U}}
    \frac
    {\displaystyle \int_{[-L, -\ell) \cup (\ell,L]} \e^{V+U}}
    {\displaystyle \int_{[-L, L]} \e^{V+U}} \\
    \label{eq:rewrite_as_convex_sum}
    &+
    \frac
    {\displaystyle \int_{[-\ell, \ell]} (F - A_{\real}) \, \e^{V+U}}
    {\displaystyle \int_{[-\ell, \ell]} \e^{V+U}}
    \frac
    {\displaystyle \int_{[-\ell, \ell]} \e^{V+U}}
    {\displaystyle \int_{[-L, L]} \e^{V+U}}.
\end{align}
The right-hand side is a convex combination of the averages of~$F- A_{\real}$ restricted to the sets~$[-L, -\ell) \cup (\ell,L]$, for the first term,
and~$[-\ell, \ell]$, for the second term.
In the proof of~\cref{lemma:asymvar_in_1d},
we proved that
\[
    \int_{\real} \e^{U+V} = \infty.
\]
Therefore, since $F$ is uniformly bounded,
the second summand on the right-hand side of~\eqref{eq:rewrite_as_convex_sum} converges to 0 in the limit as $L \to \infty$,
and so
\[
    \limsup_{L \to \infty}
    \abs*{\frac
    {\displaystyle \int_{-L}^{L} F \, \e^{V+U}}
    {\displaystyle \int_{-L}^{L} \e^{V+U}} - A_{\real}}
    =
    \limsup_{L \to \infty}
    \abs*{\frac
        {\displaystyle \int_{[-L, -\ell) \cup (\ell,L]} (F - A_{\real}) \, \e^{V+U}}
        {\displaystyle \int_{[-L, -\ell) \cup (\ell,L]} \e^{V+U}}}
    \leq \sup_{\abs{x} \geq \ell} \abs{F(x) - A_{\real}}.
\]
Since $\lim_{\abs{x} \to \infty} F(x) = A_{\real}$ by definition of~$A_{\real}$ in~\eqref{eq:constant_A_real},
the right-hand side of this equation can be made arbitrarily small by taking~$\ell$ sufficiently large,
and so the limit~\eqref{eq:limit} follows.

\paragraph{Connection between $A^*_{\real}$ and $A^*_{\torus}$.}
The constant $A_{\real}^* = A_{\real}$ coincides with
\begin{equation}
    \label{eq:limit_median}
    \lim_{\ell \to \infty} \sup \left\{ A \in \real : \int_{-\ell}^{\ell} \sgn (F-A) \geq 0 \right\}.
\end{equation}
Indeed, since $\lim_{|x| \to \infty} F(x) = A_{\real}$,
there exists for any $\varepsilon > 0$ a constant $\ell_{\varepsilon} > 0$ such that
\[
    \forall \ell \geq \ell_{\varepsilon}, \qquad
    \int_{-\ell}^{\ell} \sgn (F - A_{\real} + \varepsilon) > 0 \qquad  \text{ and } \qquad
    \int_{-\ell}^{\ell} \sgn (F - A_{\real} - \varepsilon) < 0.
\]
Therefore, for all $\ell \geq \ell_{\varepsilon}$,
the supremum in~\eqref{eq:limit_median} is contained in the interval $[A_{\real} -\varepsilon, A_{\real} + \varepsilon]$.
Since~$\varepsilon$ was arbitrary, the claim is proved.

\section{Second variation of the asymptotic variance} % M: I think this works in the main text as opposed to the appendix
\label{sec:second_variation_of_the_asymptotic_variance}
Since the method we propose in \cref{sub:optimal_solution_in_the_multi_dimensional_setting} relies on a steepest descent for the asymptotic variance
viewed as a functional of~$U$,
it is natural to wonder whether this functional is convex,
in order to provide guarantees on the convergence of the method.
We provide a partial answer to this question in~\cref{proposition:second_variation_asymptotic_variance} and~\cref{remark:non_convexity_torus} below.
Specifically, we prove that the asymptotic variance is convex when the domain is the one-dimensional real line but possibly non-convex when the domain is $\torus$. %, for any $d \in \nat$.
We have not managed to prove or rule out the convexity of the asymptotic variance in the multi-dimensional setting.

We emphasize that the convexity of the asymptotic variance in the case where the domain is~$\real$
does not imply the uniqueness (up to an additive constant) of the minimizer.
The most straightforward example is that of the constant observable,
in which case the asymptotic variance is equal to 0 for any smooth biasing potential~$U$.

\begin{proposition}
    \label{proposition:second_variation_asymptotic_variance}
    Suppose that~\cref{assumption:as1} is satisfied and let $\phi_U$ be the solution to the Poisson equation as in~\cref{lemma:asymptotic_variance}.
    Then, for all $\delta U, \delta W \in \smoothcompact(\domain^d)$,
    it holds that
    \begin{align}
        \nonumber
        &\frac{1}{2} \d (\d\sigma^2_f[U] \cdot \delta U) \cdot \delta W
        = \frac{\Z{U}}{Z^2}\int_{\domain^d} \delta U_0 \delta W_0 \left( \abs*{\nabla\phi_U}^2 + \int_{\domain^d} \abs*{\nabla\phi_U}^2 \, \d \mu_{U} \right) \e^{-V-U} \\
        \label{eq:statement_second_variation}
        &\qquad -\frac{2 \Z{U}}{Z^2} \int_{\domain^d} \bigl( \nabla \psi_{U,\delta U_0} - \delta U_0 \nabla \phi_U \bigr) \cdot  \bigl( \nabla \psi_{U,\delta W_0} - \delta W_0 \nabla \phi_U \bigr) \e^{-V-U},
    \end{align}
    where, for a perturbation $\delta X \in \{\delta U, \delta W\}$,
    \[
        \delta X_0 :=  \delta X - \mu_U(\delta X) , \qquad \mu_U(\delta X) := \int_{\domain^d} \delta X \, \d \mu_{U}.
    \]
    and $\psi_{U,\delta X_0} \in H^1(\mu_U) \cap L^2_0(\mu_U)$ is the solution to~\eqref{eq:pert2} with $\delta U = \delta X_0$.
    In addition, the second term in~\eqref{eq:statement_second_variation} is zero in dimension 1 when $\domain = \real$,
    and so the asymptotic variance~$\sigma^2_f[U]$ is a convex functional in this case.
\end{proposition}
\begin{proof}
We begin by rewriting the expression \eqref{eq:funcder} as
\begin{align*}
    \frac{1}{2}\d\sigma^2_f[U] \cdot \delta U
    &= \frac{1}{Z^2}\int_{\domain^d} \left(\Z{U} \delta U - \int_{\domain^d} \delta U \, \e^{-V-U} \right) \abs*{\nabla\phi_U}^2 \e^{-V-U} \\
    &= \frac{Z[U]}{Z^2}\int_{\domain^d} \delta U \abs*{\nabla\phi_U}^2 \e^{-V-U}
    - \left(\int_{\domain^d} \delta U \, \e^{-V-U}\right) \frac{\sigma^2_f[U]}{2\Z{U}}
    =: T_1[U;\delta U] + T_2[U;\delta U].
\end{align*}
Using the chain rule, we have
\begin{align*}
    \d T_1[U;\delta U] \cdot \delta W
    =& -\frac{1}{Z^2}\int_{\domain^d} \delta W  \e^{-V-U}\int_{\domain^d}\delta U\abs*{\nabla\phi_U}^2  \e^{-V-U} \nonumber \\
    & + \lim_{\varepsilon \to 0} \frac{\Z{U}}{\varepsilon Z^2}\int_{\domain^d}\delta U \left(\abs*{\nabla \phi_{U + \varepsilon \delta W}}^2 - \abs*{\nabla \phi_{U}}^2\right)  \e^{-V-U}
    - \frac{\Z{U}}{Z^2}\int_{\domain^d}\delta U \delta W \abs*{\nabla\phi_U}^2  \e^{-V-U}.
\end{align*}
Similarly, for the second term we obtain
\begin{align*}
    \d T_2[U;\delta U] \cdot \delta W
    = \frac{1}{Z^2}\int_{\domain^d} \delta U \delta W  \e^{-V-U} \int_{\domain^d} \abs*{\nabla\phi_U}^2  \e^{-V-U} \nonumber
    - \int_{\domain^d} \delta U  \e^{-V-U} \d\bigg(\frac{\sigma^2_f[U]}{2\Z{U}}\bigg) \cdot \delta W.
\end{align*}
The functional derivative in the last term on the right-hand side is calculated as in the proof of~\cref{proposition:functional_derivative_asym_var};
specifically,
\[
    \d\bigg(\frac{\sigma^2_f[U]}{2\Z{U}}\bigg) \cdot \delta W
    = \d\left(\frac{1}{Z^2} \int_{\domain^d} \phi_U (f-I) \e^{-V} \right) \cdot \delta W
    = \frac{1}{Z^2} \int_{\domain^d} \delta W \lvert \grad \phi_U \rvert^2 \e^{-V-U}.
\]
By~\cref{lemma:stability_gradient_sol_poisson} and the fact that $\delta U \in \smoothcompact(\domain^d)$,
we have that
\begin{align*}
    &\frac{1}{\varepsilon}
    \int_{\domain^d}\delta U \left(\abs*{\nabla \phi_{U + \varepsilon \delta W}}^2 - \abs*{\nabla \phi_{U}}^2\right)  \e^{-V-U} \\
    &\qquad \qquad =
    \int_{\domain^d}\delta U
    \left(\frac{\nabla \phi_{U + \varepsilon \delta W} - \nabla \phi_{U}}{\varepsilon}\right)
    \cdot \left(\nabla \phi_{U + \varepsilon \delta W} + \nabla \phi_{U}\right)
    \e^{-V-U} \\
    &\qquad \qquad \xrightarrow[\varepsilon \to 0]{}
    2\int_{\domain^d}\delta U \nabla \psi_{U,\delta W} \cdot \nabla \phi_{U} \, \e^{-V-U}.
\end{align*}
Collecting all the terms,
we obtain
\begin{align}
    \notag
    \frac{1}{2} \d (\d\sigma^2_f[U] \cdot \delta U) \cdot \delta W
    &= -\frac{1}{Z^2}\int_{\domain^d} \delta U  \e^{-V-U}\int_{\domain^d}\delta W\abs*{\nabla\phi_U}^2  \e^{-V-U} \\
    \notag
    &\quad -\frac{1}{Z^2}\int_{\domain^d} \delta W \e^{-V-U}\int_{\domain^d}\delta U\abs*{\nabla\phi_U}^2  \e^{-V-U} \\
    &\qquad + \frac{2\Z{U}}{Z^2}\int_{\domain^d}\delta U \nabla \psi_{U,\delta W} \cdot \nabla \phi_U  \e^{-V-U} \\
    \notag
    &\qquad + \frac{\Z{U}}{Z^2}\int_{\domain^d}\delta U \delta W \left( - \abs*{\nabla\phi_U}^2 + \int_{\domain^d} \abs*{\nabla\phi_U}^2 \, \d \mu_{U} \right) \e^{-V-U}.
\end{align}
By rewriting the last term on the right-hand side as
\[
    \frac{\Z{U}}{Z^2}\int_{\domain^d}\delta U \delta W \left( \abs*{\nabla\phi_U}^2 + \int_{\domain^d} \abs*{\nabla\phi_U}^2 \, \d \mu_{U} \right) \e^{-V-U}
    - \frac{2\Z{U}}{Z^2}\int_{\domain^d}\delta U \delta W \abs*{\nabla\phi_U}^2 \e^{-V-U},
\]
and substituting $\delta U \delta W = \delta U_0 \delta W_0 + \delta U \mu_U(\delta W) + \delta W \mu_U(\delta U) - \mu_U(\delta U) \mu_U(\delta W)$ in the first term of the latter expression,
the second variation may be further simplified to
\begin{align}
    \notag
    \frac{1}{2} \d (\d\sigma^2_f[U] \cdot \delta U) \cdot \delta W
    &= \frac{\Z{U}}{Z^2}\int_{\domain^d} \delta U_0 \delta W_0 \left( \abs*{\nabla\phi_U}^2 + \int_{\domain^d} \abs*{\nabla\phi_U}^2 \, \d \mu_{U} \right) \e^{-V-U} \\
    \label{eq:rearranged_second_variation}
    &\qquad + \frac{2 \Z{U}}{Z^2} \int_{\domain^d}\delta U \nabla \phi_U \cdot (\nabla \psi_{U,\delta W} - \delta W \nabla \phi_U) \, \e^{-V-U}.
\end{align}
Using~\eqref{eq:auxiliary_func_der},
both for $\psi_{U,\delta U}$ and $\psi_{U,\delta W}$,
we obtain
\begin{align*}
    &\int_{\domain^d}\delta U \nabla \phi_U \cdot \bigl(\nabla \psi_{U,\delta W} - \delta W \nabla \phi_U\bigr)  \e^{-V-U} \\
    &\qquad = - \int_{\domain^d} \bigl( \nabla \psi_{U,\delta U} - \delta U \nabla \phi_U \bigr) \cdot  \bigl( \nabla \psi_{U,\delta W} - \delta W \nabla \phi_U \bigr) \e^{-V-U}.
\end{align*}
From~\eqref{eq:pert2},
it is simple to see that $\nabla \psi_{U,\delta U} = \nabla \psi_{U,\delta U_0} + \nabla \psi_{U, \mu_U(\delta U)} = \nabla \psi_{U,\delta U_0} + \mu_U(\delta U) \nabla \phi_U$.
Similarly, $\nabla \psi_{U,\delta W} = \nabla \psi_{U,\delta W_0} + \mu_U(\delta W) \nabla \phi_U$.
Substituting these expressions in~\eqref{eq:rearranged_second_variation} leads to the claimed result~\eqref{eq:statement_second_variation}.

\paragraph{One-dimensional setting.}
In dimension 1 when $\domain = \real$,
it holds that $\psi_{U,\delta W}' = \delta W \phi_U'$ by~\eqref{eq:psi_prime},
% the formula \eqref{eq:derivative_of_phi} can be used to obtain directly~\eqref{eq:psi_prime},
% \begin{equation}\label{eq:sec2}
%     \psi_{U,\delta U}' = - \delta U (F+A_{\real}) \e^{V+U} = \delta U \phi_U',
% \end{equation}
and so the second term in~\eqref{eq:rearranged_second_variation} cancels out,
which proves the last part of the statement.
% Note that~\eqref{eq:sec2} is also~\eqref{eq:psi_prime} in the specific case where
\end{proof}
\begin{remark}
Since all the terms on the right-hand side of~\eqref{eq:statement_second_variation} depend only on $\delta U_0$,
the second variation is invariant under vertical shift of $\delta U$,
in the sense that, formally,
\[
    \forall C \in \real, \qquad
    \d\bigl(\d\sigma^2_f[U] \cdot (\delta U + C)\bigr) \cdot (\delta U + C) = \d(\d\sigma^2_f[U] \cdot \delta U) \cdot \delta U.
\]
This property had to hold a priori because $\sigma^2_f[U]$ is itself invariant under addition of constants to~$U$,
and so we could have assumed that $\mu_U(\delta U) = 0$ from the beginning of the proof without loss of generality.
\end{remark}

\begin{remark}
    The optimal biasing potential is known explicitly by \cref{lemma:asymvar_in_1d} in the one-dimensional setting,
    so~\cref{proposition:second_variation_asymptotic_variance} is of little direct importance in this case.
    Nonetheless, the result provides understanding for the numerical experiments using the formula of the directional derivative.
\end{remark}

\begin{remark}
    [One-dimensional case with $\domain = \torus$]
    \label{remark:non_convexity_torus}
    The asymptotic variance is not a convex functional when $\domain = \torus$ and $d = 1$.
    Indeed, we construct in this remark a potential $V$, a smooth function~$\phi$,
    and a direction $\delta U$ such that
    the second variation of the asymptotic variance~$\sigma^2_f$ for the observable $f = -\mathcal L \phi$
    (with $\mathcal L$ the generator~$\mathcal L_U$ given in~\eqref{eq:generator} with~$U = 0$)
    in the direction~$\delta U$ is negative when evaluated at the biasing potential~$U = 0$.
    In the setting we consider, since $\phi$ is the solution to~\eqref{eq:poisson} when~$U = 0$,
    it holds by~\cref{remark:psi_in_dim_one} that
    \[
        \psi_{U=0,\delta U_0}' = \delta U_0 \phi'
        - \left( \frac{\int_{\torus} \delta U_0 \phi'}  {\int_{\torus}\e^{V}} \right)\e^{V}.
    \]
    Here $\delta U_0 := \delta U - \mu(\delta U)$.
    Therefore, by substitution in~\eqref{eq:statement_second_variation} we have that
    \[
        \frac{1}{2} \d \bigl(\d\sigma^2_f[0] \cdot \delta U\bigr) \cdot \delta U
        = \frac{1}{Z}
        \left( \int_{\torus} \delta U_0^2 \left( \abs*{\phi'}^2 + \int_{\torus} \abs*{\phi'}^2\d\mu \right) \e^{-V}
        - 2 \frac{ \left(\int_{\torus} \delta U_0 \phi' \right)^2}{\int_{\torus} \e^{V}} \right),
    \]
    The right-hand side of this equation is not always positive.
    In order to show this,
    consider the case where~%
    \(
        \delta U_0 = \phi' \e^{V}.
    \)
    Note that $\delta U_0$ indeed has average~0 with respect to~$\mu$
    since
    \(
        \int_{\torus} \delta U_0 \e^{-V} = \int_{\torus} \phi' = 0.
    \)
    Then, we have
    \[
        \frac{Z}{2} \d \bigl(\d\sigma^2_f[0] \cdot \delta U\bigr) \cdot \delta U
        = \int_{\torus} \abs{\phi'}^4 \e^{V} + \frac{\int_{\torus} \abs*{\phi'}^2 \e^{-V}}{\int_{\torus} \e^{-V}} \int_{\torus} \abs*{\phi'}^2 \e^{V}
        - 2 \frac{ \left(\int_{\torus} \abs{\phi'}^2 \e^{V} \right)^2}{\int_{\torus} \e^{V}}.
    \]
    Assume that $\phi = \varrho_{\varepsilon} \star h + C$ is a regularization of a hat function~$h\colon \torus\to \real$ given on the interval~$[-\pi, \pi]$,
    which we identify with its image under the quotient map~$\real \to \torus$,
    by
    \[
        h(x) :=
        \begin{cases}
            1 - \abs{x}, \qquad &\abs{x} < 1, \\
            0 & \text{otherwise},
        \end{cases}
    \]
    with $\varrho_{\varepsilon}$ the standard mollifier~\eqref{eq:mollification} and~$C\in \real$ the constant such that $\phi$ has average 0 with respect to~$\mu$.
    Then, letting $\nu$ denote the probability measure with Lebesgue density proportional to $\e^{V}$
    and~$\mathcal I = [-1, 1]$,
    we obtain that, in the limit as $\varepsilon \to 0$,
    \[
        \int_{\torus} \abs{\phi'}^4 \d \nu \to \nu(\mathcal  I),
        \qquad
        \int_{\torus} \abs*{\phi'}^2 \d \mu \to \mu (\mathcal  I),
        \qquad
        \int_{\torus} \abs*{\phi'}^2 \d \nu \to \nu (\mathcal  I).
    \]
    Therefore, it holds in this limit that
    \[
        \frac{Z}{2} \d (\d\sigma^2_f[0] \cdot \delta U) \cdot \delta U
        \to \left( \nu(\mathcal  I) + \mu(\mathcal  I) \nu(\mathcal I) - 2 \nu(\mathcal I)^2 \right) \int_{\torus} \e^{V} .
    \]
    Now let $V(x) = K \cos(x)$ for all~$x$.
    In the limit as $K \to \infty$, it holds that $\mu(\mathcal I) \to 0$ and $\nu(\mathcal I) \to 1$.
    We conclude that, for sufficiently large~$K$ and sufficiently small $\varepsilon$,
    the second variation of the asymptotic variance in direction $\delta U$ is negative.
\end{remark}

\section{Numerical discretization of the Poisson equation}
\label{sec:numerical_discretization_of_the_Poisson_equation}
We consider here the case where the domain is $\torus^2$ for simplicity,
noting that the method may be generalized to any spatial dimension.
In order to numerically solve the Poisson equation~\eqref{eq:poisson},
we use a finite difference approach on a grid of size $N \times N$.
For a given~$\delta>0$, the discretization nodes are arranged linearly according to
\begin{equation}
    \label{eq:discretization_nodes}
    \vect x_{\ell} := (- \pi + i \step, - \pi + j \step) \in \torus^2,
    \qquad j = \left\lfloor \frac{\ell - 1}{N} \right\rfloor, \qquad i = \ell - 1 - j N,
    \qquad \step = \frac{2\pi}{N},
\end{equation}
for $\ell \in \{1, \dotsc, N^2\}$.
Note that the indices~$i$ and~$j$ each run from $0$ to $N-1$;
the largest value of either coordinate over the set of discretization nodes is $\pi - \step$,
which is sufficient given that $-\pi$ and $\pi$ coincide under the quotient map~$\real^2 \to \torus^2$.

Before we present the method,
we introduce additional notation.
We denote by $\Pi_N$ the discretization operator which associates to a function its values at the grid points~\eqref{eq:discretization_nodes},
and for a function~$h\colon \torus^2 \to \real$,
we write~$\vect h = \Pi_N h \in \real^{N^2}$.
The notation~$\exp.(\vect h)$ refers the vector obtained by applying the exponential function element-wise to~$\vect h$,
and $\diag(\vect h)$ refers to the diagonal matrix with diagonal entries given by~$\vect h$.
The notation $\vect 1 \in \real^{N^2}$ refers to a column vector containing only ones.
We also introduce the one-dimensional backward and forward difference operators,
which act on vectors in $\real^N$:
\begin{align*}
    \mat D_{\rm B} =
    \frac{1}{\step}
    \begin{pmatrix}
        1 & & & & -1 \\
          -1 & 1 \\
             & -1 & 1  \\
             & & \ddots & \ddots  \\
          & & & -1 & 1 \\
    \end{pmatrix},
    \qquad
    \mat D_{\rm F} =
    \frac{1}{\step}
    \begin{pmatrix}
        -1 & 1 \\
          & -1 & 1 \\
          & & \ddots & \ddots \\
          & & & -1 & 1 \\
          1  & & & & -1 \\
    \end{pmatrix}.
\end{align*}
From these operators,
we construct difference operators along the~$x$ and~$y$ directions by taking Kronecker products with the $\real^{N\times N}$ identity matrix~$\matid_N$:
\[
    \mat D_{\rm B}^x = \matid_N \otimes \mat D_{\rm B},
    \qquad
    \mat D_{\rm B}^y = \mat D_{\rm B} \otimes \matid_N,
    \qquad
    \mat D_{\rm F}^x = \matid_N \otimes \mat D_{\rm F},
    \qquad
    \mat D_{\rm F}^y = \mat D_{\rm F} \otimes \matid_N.
\]
We recall that, for two matrices $\mat A, \mat B \in \real^{N \times N}$,
the Kronecker product $A \otimes B$ is defined as
\[
    \mat A \otimes\mat B = \begin{pmatrix}
        a_{11} \mat B & \cdots & a_{1N}\mat B \\
        \vdots & \ddots &           \vdots \\
        a_{N1} \mat B & \cdots & a_{NN} \mat B
    \end{pmatrix}.
\]
We denote by $\nabla_{\rm F} \vect h$ the $N^2 \times 2$ matrix
\(
    \nabla_{\rm F} \vect h =
    \begin{pmatrix}
        \mat D_{\rm F}^x \vect h & \mat D_{\rm F}^y \vect h
    \end{pmatrix}.
\)
For a weight function $w\colon \torus^2 \to \real$,
we introduce the weighted inner product~$\ip{\placeholder,\placeholder}_{w}\colon \real^{N^2}\times\real^{N^2} \to \real$ given for~$\vect g,\vect h\in\real^{N^2}$ by
\begin{equation}
    \label{eq:weighted_inner_product}
    \ip{\vect g, \vect h}_{w}
    = \step^2 \, \vect g^\t \diag\bigl(\vect w\bigr) \vect h
    = \step^2 \sum_{\ell=1}^{N^2} \vect g_\ell \vect h_\ell \, w(\vect x_{\ell}),
\end{equation}
with corresponding norm $\norm{\placeholder}_{w}$.
We include the factor $\step^2$ in this definition so that,
if~$\vect g$ and~$\vect h$ contain the values taken by continuous functions $g$ and $h$ when evaluated at the discretization points
and $w$ is continuous,
then
\[
    \ip{\vect g, \vect h}_w \xrightarrow[N \to \infty]{} \int_{\torus^2} g(x) h(x) w(x) \, \d x.
\]
Finally, let $\norm{\nabla_{\rm F} \vect h}_{w}^2 = \norm{\mat D_{\rm F}^x \vect h}_{w}^2 + \norm{\mat D_{\rm F}^y \vect h}_{w}^2$
and let~$\abs{\nabla_{\rm F} \vect h}^2$ denote the $N^2 \times 1$ column vector obtained by taking the squared Euclidean norm of each row of~$\nabla_{\rm F} \vect h$.
In the remainder of this section,
the notation~\eqref{eq:weighted_inner_product} and corresponding norm are usually employed with the weight function~$w = \e^{-V-U}$ and so,
in order to simplify notation, we omit the subscript in this case.
We are now ready to write the discrete formulation of the Poisson equation~\eqref{eq:poisson}.

\begin{proposition}
    \label{proposition:discrete_poisson}
    For~$V,U,f\colon \torus^2 \to \real$, there exists a unique solution~$(\vect \phi_N, I_N) \in \real^{N^2} \times \real$ to
    \begin{equation}
        \label{eq:discrete_poisson_equation}
        -\widetilde {\mat L}
        \begin{pmatrix}
            \vect \phi_N \\
            I_N
        \end{pmatrix}
        :=
        \begin{pmatrix}
            -\mat L & \exp.(\vect U) \\
            \step^2\exp.(- \vect V - \vect U)^\t & 0
        \end{pmatrix}
        \begin{pmatrix}
            \vect \phi_N \\
            I_N
        \end{pmatrix}
        =
        \begin{pmatrix}
            \diag\bigl(\exp.(\vect U)\bigr) \vect f \\
            0
        \end{pmatrix},
    \end{equation}
    where
    \begin{align}
        \nonumber
        \mat L &= \diag\bigl(\exp.(\vect V + \vect U)\bigr) \mat D_{\rm B}^x \diag\bigl(\exp.(-\vect V - \vect U)\bigr) \mat D_{\rm F}^x \\
        \label{eq:definition_of_matrix_L}
               &\qquad +\diag\bigl(\exp.(\vect V + \vect U)\bigr) \mat D_{\rm B}^y \diag\bigl(\exp.(-\vect V - \vect U)\bigr) \mat D_{\rm F}^y.
    \end{align}
\end{proposition}

\begin{remark}
    The first rows in~\eqref{eq:discrete_poisson_equation} may be rewritten as
    \[
        - \mat L \vect \phi_N = \diag\bigl(\exp.(\vect U)\bigr) (\vect f - I_N \vect 1),
    \]
    which resembles the Poisson equation~\eqref{eq:poisson}.
    The last row in~\eqref{eq:discrete_poisson_equation} may be rewritten as
    \[
        \step^2\exp.(- \vect V - \vect U)^\t \vect \phi_N = \ip{\vect 1, \vect \phi_N} = 0.
    \]
    It expresses the requirement that the vector $\vect \phi_N$ should be mean-zero with respect to the discrete measure $\exp.(-\vect V - \vect U)$.
\end{remark}

\begin{remark}
    We use the notation~$I_N$ for the scalar unknown in~\eqref{eq:discrete_poisson_equation}
    because solving~\eqref{eq:discrete_poisson_equation} yields
    both an approximate solution to the Poisson equation and an approximation of~$I=\mu(f)$.
\end{remark}
\begin{proof}
    In order to prove the statement,
    it is sufficient to show that the homogeneous equation
    \begin{equation}
        \label{eq:homogeneous}
        \begin{pmatrix}
            - \mat L & \exp.(\vect U) \\
            \step^2\exp.(- \vect V - \vect U)^\t & 0
        \end{pmatrix}
        \begin{pmatrix}
            \vect \gamma \\
            \sigma
        \end{pmatrix}
        =
        \begin{pmatrix}
            \vect 0 \\
            0
        \end{pmatrix}
    \end{equation}
    admits only the trivial solution $ = (\vect 0, 0)$.
    We assume by contradiction that $(\vect \gamma, \sigma)$ is a nonzero solution.
    Then
    \[
        - \mat L \vect \gamma + \sigma \exp.(\vect U) = \vect 0,
    \]
    implying that
    \begin{equation}
        \label{eq:inner_product}
        - \ip{\mat L \vect \gamma, \vect 1}_{} + \sigma \ip{\exp.(\vect U), \vect 1}_{} = 0.
    \end{equation}
    The linear operator on $\real^{N^2 \times N^2}$ induced by $\mat L$ is self-adjoint for the inner product~$\ip{\placeholder,\placeholder}_{}$,
    because
    \begin{align}
        \nonumber
        -\ip{\vect g, \mat L \vect h}_{}
        &= - \step^2 \vect g^\t
        \left( \mat D_{\rm B}^x \diag\bigl(\exp.(-\vect V - \vect U)\bigr) \mat D_{\rm F}^x + \mat D_{\rm B}^y \diag\bigl(\exp.(-\vect V - \vect U)\bigr) \mat D_{\rm F}^y \right) \vect h \\
        \label{eq:discrete_bilinear}
        &= \ip{\mat D_{\rm F}^x \vect g, \mat D_{\rm F}^x \vect h}_{} + \ip{\mat D_{\rm F}^y \vect g, \mat D_{\rm F}^y \vect h}_{},
    \end{align}
    where we used the relation $\mat D_{\rm B}^\t = - \mat D_{\rm F}$.
    Therefore, going back to~\eqref{eq:inner_product},
    we deduce that
    \[
        0 = - \ip{\vect \gamma, \mat L \vect 1}_{} + \sigma \ip{\exp.(\vect U), \vect 1}_{}
        = \sigma \ip{\exp.(\vect U), \vect 1}_{}
        = \step^2\sigma \vect 1^\t \exp.(-\vect V).
    \]
    Therefore $\sigma = 0$,
    but then $\mat L \vect \gamma = \vect 0$ by%the first equation in
    ~\eqref{eq:homogeneous} and so $\ip{\vect \gamma, \mat L \vect \gamma}_{} = 0$.
    By the relation~\eqref{eq:discrete_bilinear},
    this implies that~$D_{\rm F}^x \vect \gamma = D_{\rm F}^y \vect \gamma = 0$,
    so the vector $\vect \gamma$ is constant.
    The last equation in~\eqref{eq:homogeneous} then implies that~$\vect \gamma = \vect 0$.
    Note that~\eqref{eq:discrete_bilinear} implies that
    \(
        \kernel(L^\t) = \Span \bigl\{ \exp.(-\vect V- \vect U) \bigr\},
    \)
    which will be useful in the proof of \cref{lemma:stability_finite_difference_disretization}.
\end{proof}

It is possible to prove the convergence of the solution to~\eqref{eq:discrete_poisson_equation}
to the exact solution of the Poisson equation~\eqref{eq:poisson} in the limit as~$N \to \infty$.
To this end,
we begin by showing the following Poincaré-like inequality.
\begin{lemma}
    [Discrete Poincaré inequality]
    \label{lemma:discrete_poincare}
    Assume that $V+U\colon \torus^2 \to \real$ is uniformly bounded.
    Then there exists a constant $R_{\rm disc}[U] > 0$ independent of $N$ such that
    \begin{align}
        \label{eq:discrete_poincaré}
        \forall \vect g \in \left\{\vect h \in \real^{N^2} : \vect h^\t \exp.(- \vect V - \vect U) = 0 \right\}, \qquad
        \norm{\nabla_{\rm F} \vect g}_{}^2 \geq  R_{\rm disc}[U] \norm{\vect g}_{}^2.
    \end{align}
\end{lemma}
\begin{proof}
    It is sufficient to show~\eqref{eq:discrete_poincaré} for $V+U = 0$.
    Indeed, assuming that the inequality holds in this particular case and denoting by $C$ a constant
    which depends only on $V + U$ and is allowed to change from line to line,
    we have that
    \begin{align*}
        \forall \vect g \in \real^{N^2}, \qquad
        &\ip{\mat D_{\rm F}^x \vect g, \mat D_{\rm F}^x \vect g}_{} + \ip{\mat D_{\rm F}^y \vect g, \mat D_{\rm F}^y \vect g}_{} \\
        &\qquad \geq C \Bigl( \ip{\mat D_{\rm F}^x \vect g, \mat D_{\rm F}^x \vect g}_1 + \ip{\mat D_{\rm F}^y \vect g, \mat D_{\rm F}^y \vect g}_1 \Bigr) \\
        &\qquad \geq C \ip{\vect g - \widetilde {\vect g}, \vect g - \widetilde {\vect g}}_1,
        \qquad \widetilde {\vect g} =  \frac{\ip{\vect g, \vect 1}_1}{\ip{\vect 1, \vect 1}_1}  \vect 1,
    \end{align*}
    where we employed the equivalence between $\ip{\placeholder, \placeholder}_1$,
    which is given by~\eqref{eq:weighted_inner_product} in the particular case where~$V + U = 0$,
    and $\ip{\placeholder, \placeholder}_{\e^{-V-U}}$,
    noting that both constants in this equivalence can be fixed independently of~$N$.
    Using this equivalence in the other direction,
    we obtain
    \begin{align*}
        \ip{\vect g - \widetilde {\vect g}, \vect g - \widetilde {\vect g}}_1
        &\geq C \ip{\vect g - \widetilde {\vect g}, \vect g - \widetilde {\vect g}}_{}
        \geq C \inf_{s \in \real} \ip{\vect g - s \vect 1, \vect g - s \vect 1}_{} \\
        &= C \ip{\vect g - \overline {\vect g}, \vect g - \overline{\vect g}}_{},
        \qquad \overline {\vect g} =  \frac{\ip{\vect g, \vect 1}_{}}{\ip{\vect 1, \vect 1}_{}}  \vect 1.
    \end{align*}
    Finally, equation~\eqref{eq:discrete_poincaré} when $V + U = 0$ follows from its one-dimensional counterpart by using a standard tensorization argument (as for the proof of~\cite[Proposition 2.6]{MR3509213} for instance).
    It only remains to show the one-dimensional inequality
    \begin{equation}
        \label{eq:one_dimensional_poincare}
        \forall \vect g \in \real^N, \qquad
        \ip{\mat D_{\rm F} \vect g, \mat D_{\rm F} \vect g}_1
        \geq  R_{\rm disc}[U] \ip{\vect g - \widetilde {\vect g}, \vect g - \widetilde {\vect g}}_1,
        \qquad \widetilde {\vect g} =  \frac{\ip{\vect g, \vect 1}_1}{\ip{\vect 1, \vect 1}_1}  \vect 1,
    \end{equation}
    for a constant $R_{\rm disc}[U]$ independent of~$N$ for sufficiently large~$N$.
    To this end,
    we notice that
    \[
        \ip[\big]{\mat D_{\rm F} \vect g, \mat D_{\rm F} \vect g}_1
        = \ip[\big]{\mat D_{\rm B} \mat D_{\rm F} (\vect g - \widetilde {\vect g}), \vect g - \widetilde {\vect g}}_1.
        % \geq \lambda_{*} (\mat D_{\rm B} \mat D_{\rm F}) \ip[\big]{\vect g - \widetilde {\vect g}, \vect g - \widetilde {\vect g}}_1.
    \]
    The matrix $D_{\rm B} \mat D_{\rm F}$ is given by
    \[
        D_{\rm B} \mat D_{\rm F} =
        \frac{1}{\step^2}
        \begin{pmatrix}
            2 & -1 & & & -1\\
            -1 & 2 & -1 \\
               & -1 & \ddots & \ddots \\
               & & \ddots & \ddots & -1 \\
            -1 & & & -1 & 2
        \end{pmatrix}.
    \]
    This is a circulant matrix~\cite{MR543191} with explicit eigenvalues given by
    \[
        \lambda_k = \frac{4}{\delta^2} \sin^2 \left(\frac{\pi k}{N}\right) = \frac{N^2}{\pi^2} \sin^2 \left(\frac{\pi k}{N}\right), \qquad k= 0, 1, \dotsc, N-1.
    \]
    The minimum eigenvalue of this matrix is $\lambda_0 = 0$,
    and the associated eigenvector is $\vect 1$,
    to which $\vect g - \widetilde{\vect g}$ is orthogonal.
    Therefore, equation~\eqref{eq:one_dimensional_poincare} implies that
    \[
        \forall \vect g \in \real^N, \qquad
        \ip[\big]{\mat D_{\rm F} \vect g, \mat D_{\rm F} \vect g}_1
        \geq \lambda_1 \ip[\big]{\vect g - \widetilde {\vect g}, \vect g - \widetilde {\vect g}}_1,
    \]
    which implies that, for fixed~$N$,
    equation~\eqref{eq:one_dimensional_poincare} holds with constant $R_{\rm disc}(N) = (N/\pi)^2 \sin^2 (\pi/N)$,
    which converges to $1$ in the limit $N \to \infty$.
    % see, for example, \cite[Proposition~4.5.5]{MR3155209}.
    See also~\cite[Lemma~12.2]{MR2265914} for a Poincaré inequality for discrete functions on a bounded interval that are zero at the endpoints.
\end{proof}

\iffalse
\begin{remark}
    It would be interesting to show that~\eqref{eq:discrete_poincaré} is satisfied for a constant $R_{\rm disc}[U] = R_{\rm disc}[U](N)$
    which converges to the Poincaré constant of $\mu_U$ in the limit as $N \to \infty$.
    Our approach does not enable to prove such a statement,
    and we do not further investigate this question.
\end{remark}
\fi

\Cref{proposition:discrete_poisson} implies that $\widetilde {\mat L}$ is invertible.
Using~\cref{lemma:discrete_poincare},
we show that $\widetilde {\mat L}^{-1}$ does not diverge in the limit as $N \to \infty$,
in an appropriate norm.
\begin{lemma}
    \label{lemma:stability_finite_difference_disretization}
    Assume that~$V,U\colon \torus^2 \to \real$ are continuous.
    Then the matrix $\widetilde{\mat L}^{-1}$ is bounded uniformly in~$N \in\mathbb{N}$,
    for the operator norm induced by the following norm on $\real^{N^2} \times \real$:
    \begin{equation}
        \label{eq:norm_discrete}
        (\vect \gamma, \sigma) \mapsto \norm{\vect \gamma}_{\e^{-U-V}} + \abs{\sigma}.
    \end{equation}
\end{lemma}
\begin{proof}
    % For simplicity,
    % we use the notation $\ip{\placeholder,\placeholder}$ and $\norm{\placeholder}$
    % in place of $\ip{\placeholder,\placeholder}_{}$ and $\norm{\placeholder}_{}$ throughout the proof.
    The strategy of proof is similar to that used in~\cite[Section 2.2]{MR3865558}.
    Since by the Fredholm alternative $\range(\mat L) = \kernel(\mat L^{\t})^{\perp}$,
    the range of $\mat L$ is given by  $\left\{\vect h \in \real^{N^2} : \vect h^\t \exp.(- \vect V - \vect U) = 0 \right\}$.
    It then follows by~\cref{lemma:discrete_poincare} that the following inequality holds for all $\vect g \in \range(\mat L)$:
    \[
        \norm{\mat L \vect g} \norm{\vect g} \geq
        - \ip{\mat L \vect g, \vect g}
        = \norm{\mat D_{\rm F}^x \vect g}^2 + \norm{\mat D_{\rm F}^y \vect g}^2
        \geq R_{\rm disc}[U] \norm{\vect g}^2,
    \]
    and so we deduce that that $\norm{\mat L^{-1}} \leq \frac{1}{R_{\rm disc}[U]}$ over $\range(\mat L)$.
    Denote by $(\vect \gamma, \sigma)$ the solution to
    \begin{equation}
        \label{eq:conditioning}
        \begin{pmatrix}
            -\mat L & \exp.(\vect U) \\
            \step^2\exp.(- \vect V - \vect U)^\t & 0
        \end{pmatrix}
        \begin{pmatrix}
            \vect \gamma \\
            \sigma
        \end{pmatrix}
        =
        \begin{pmatrix}
            \vect g \\
            s
        \end{pmatrix}.
    \end{equation}
    This solution  satisfies
    \begin{equation}
        \label{eq:first_equation_expanded}
        - \ip{\mat L \vect \gamma, \vect 1} + \sigma \ip{\exp.(\vect U), \vect 1} = \ip{\vect g, \vect 1},
    \end{equation}
    and since $\ip{\mat L \vect \gamma, \vect 1} = \ip{\vect \gamma, \mat L \vect 1} = 0$,
    this implies
    \[
        \lvert \sigma \rvert = \left\lvert \frac{\ip{\vect g, \vect 1}}{\ip{\exp.(\vect U), \vect 1}} \right\rvert
        \leq \frac{\norm{\vect g} \norm{\vect1}}{\ip{\exp.(\vect U), \vect 1}}.
    \]
    We then deduce that
    \[
        \vect \gamma - \overline {\vect{\gamma}} = -\mat L^{-1} \bigl( \vect g - \sigma \exp.(\vect U) \bigr),
        \qquad
        \overline {\vect \gamma} =  \frac{\ip{\vect \gamma,\vect 1}}{\ip{\vect 1,\vect 1}}  \vect 1,
    \]
    and then the last equation in~\eqref{eq:conditioning} gives~$\overline{\vect \gamma} = s \vect 1 / \norm{\vect 1}^2$.
    This leads to the bound
    \begin{align*}
        \norm{\vect \gamma}
        &\leq \norm{\vect \gamma - \overline{\vect \gamma}} + \norm{\overline{\vect \gamma}} \\
        &\leq \frac{1}{R_{\rm disc}[U]} \bigl(\norm{\vect g} + \abs{\sigma} \norm{\exp.(\vect U)} \bigr)
        + \abs{s} \frac{\norm{\vect 1}}{\norm{\vect 1}^2}
        \leq \frac{\norm{\vect g}}{R_{\rm disc}[U]} \left(1 + \frac{\norm{\exp.(\vect U)} \norm{\vect1}}{\ip{\exp.(\vect U), \vect 1}}  \right) + \abs{s}  \frac{\norm{\vect 1}}{\norm{\vect 1}^2}.
    \end{align*}
    In the limit $N \to \infty$, it holds that
    \[
        \norm{\vect 1} \to \sqrt{\int_{\torus^2} \e^{-V-U}},
        \qquad \norm{\exp.(\vect U)} \to \sqrt{\int_{\torus^2} \e^{-V+U}} ,
        \qquad \ip{\exp.(\vect U), \vect 1} \to \int_{\torus^2} \e^{-V},
   \]
   which enables to conclude the proof.
\end{proof}

We are now ready to prove the convergence of the solution of the discretized Poisson equation~\eqref{eq:discrete_poisson_equation} in the limit~$N \to \infty$.
\begin{proposition}
    Suppose that \cref{assumption:as1} is satisfied.
    Let $\phi$ denote the exact solution to~\eqref{eq:poisson} and let~$I = \mu(f)$.
    Let also $(\vect \phi_N, I_N)$ denote the solution to the discretized equation~\eqref{eq:discrete_poisson_equation}.
    %In the limit as $N \to \infty$,
    Then it holds that
    \[
        \norm{\vect \phi_N - \Pi_N \phi}_{} \xrightarrow[N \to \infty]{} 0,
        \qquad I_N \xrightarrow[N \to \infty]{} I.
    \]
\end{proposition}
\begin{proof}
    The proof is an application of the standard Lax equivalence theorem.
    We denote the matrix of the linear system~\eqref{eq:discrete_poisson_equation} by $\widetilde {\mat L}_N$ to emphasize its dependence on $N$.
    Convergence follows from the usual argument:
    \[
        \norm*{
            \begin{pmatrix}
                \Pi_{N} \phi - \vect \phi_N \\
                I - I_N
            \end{pmatrix}
        }
        \leq
        C
        \norm*{%
            {\widetilde {\mat L}_N }
            \begin{pmatrix}
                \Pi_{N} \phi - \vect \phi_N \\
                I - I_N
            \end{pmatrix}
        }%
        =
        \norm*{%
            {\widetilde {\mat L}_N }
            \begin{pmatrix}
                \Pi_{N} \phi \\
                I
            \end{pmatrix}
            -
            \begin{pmatrix}
                \Pi_N\left(f \e^U\right) \\
                0
            \end{pmatrix}
        }%
        \xrightarrow[N \to \infty]{} 0,
    \]
    where the norm in this equation is that defined in~\eqref{eq:norm_discrete}.
    The first inequality follows from the stability statement of~\cref{lemma:stability_finite_difference_disretization},
    while the limit follows from the consistency of the discretization,
    which is simple to check given that $\phi$ is a smooth function under~\cref{assumption:as1},
    and relying on the presence of the factor~$\step^2$ in the definition~\eqref{eq:weighted_inner_product}.
\end{proof}

The main interest of the discretization~\eqref{eq:discrete_poisson_equation} lies in the following statement,
which may be viewed as a result on the commutation of the discretization and derivative operators.
In order to be more precise,
we denote by $(\vect \phi_N, I_N)$ the solution to~\eqref{eq:discrete_poisson_equation} and let
\begin{equation}
    \label{eq:discretized_asym_var}
    \sigma^2_{f,N}[\vect U] = \frac{2\ZN{\vect U}}{Z_N^2}  \norm{\nabla_{\rm F} \vect \phi_N}_{}^2,
    \qquad
    \text{ where }
    \left\{
    \begin{aligned}
        Z_N &:= \step^2 \vect 1^\t \exp.(-\vect V), \qquad \\
        \ZN{\vect U} &:= \step^2 \vect 1^\t \exp.(-\vect V - \vect U).
    \end{aligned}
    \right.
\end{equation}
The following statement shows that the functional derivative of $\widehat \sigma^2_f[\vect U]$
has a structure similar to that of~$\sigma^2_f[U]$ given in~\eqref{eq:asym_var};
it may be viewed as a discretization thereof.

\begin{proposition}
    [Functional derivative of $\sigma^2_{f,N}$]
    \label{proposition:functional_derivative_asym_var_discrete}
    Suppose that~$V, U \colon \torus^2 \to \real$ are uniformly bounded.
    The functional derivative with respect to~$\vect U$ of~$\sigma^2_{f,N}$ is given by
    \begin{equation}
        \label{eq:funcder_discrete}
        \frac{1}{2} \d \sigma^2_{f,N}[\vect U] \cdot \vect {\delta U}
            = \frac{\ZN{\vect U}}{Z_N^2}
            \ip*{\vect {\delta U}, \abs{\nabla_{\rm F} \vect \phi_N}^2 - \overline {\abs{\nabla_{\rm F} \vect \phi_N}^2}}_{},
            \qquad
            \overline {\abs{\nabla_{\rm F} \vect \phi_N}^2} := \frac{\norm{\nabla_{\rm F} \vect \phi_N}_{}^2}{\ZN{\vect U}}.
    \end{equation}
\end{proposition}
Notice that~\eqref{eq:funcder_discrete} is very similar to the formula~\eqref{eq:funcder} for the functional derivative of~$\sigma^2_f[U]$.
\begin{proof}
    The proof mirrors that of~\cref{proposition:functional_derivative_asym_var}.
    In view of~\eqref{eq:discrete_bilinear}, we first rewrite
    \begin{align*}
        \sigma^2_{f,N}[\vect U]
        &= \frac{2\ZN{\vect U}}{Z_N^2}  \norm{\nabla_{\rm F} \vect \phi_N}_{}^2
        = - \frac{2\ZN{\vect U}}{Z_N^2} \ip{\vect \phi_N, \mat L \vect \phi_N}_{} \\
        &= \frac{2\ZN{\vect U}}{Z_N^2} \ip{\vect \phi_N, \diag\bigl(\exp.(\vect U)\bigr) (\vect f - I_N \vect 1)}_{}
        = \frac{2\ZN{\vect U}}{Z_N^2} \ip{\vect \phi_N, \vect f - I_N\vect 1}_{\e^{-V}}.
    \end{align*}
    Let $(\vect \phi_N^{\varepsilon}, I_N^{\varepsilon})$ denote the solution to~\eqref{eq:discrete_poisson_equation} with~$\vect U + \varepsilon \vect {\delta U}$ in place of~$\vect U$ everywhere
    and~$\mat L^\varepsilon$ the corresponding matrix~\eqref{eq:definition_of_matrix_L}.
    It is simple to check, using a reasoning similar to~\eqref{eq:first_equation_expanded} as well as the equation~$\ip{L^\varepsilon\vect \phi_N^{\varepsilon},\vect 1} = 0$,
    that the scalar term $I^{\varepsilon}_N = I_N$ is fact independent of the potential $\vect U$.
    Therefore,
    we obtain that
    \[
        -\mat L^{\varepsilon} \vect \phi_N^{\varepsilon}
        = \diag\bigl(\exp.(\vect U + \varepsilon \vect {\delta U})\bigr) (\vect f - I_N \vect 1).
    \]
    By definition of the functional derivative,
    we then have
    \begin{align}
        \label{eq:def_func_der_discrete}
        \frac{1}{2} \d \sigma^2_{f,N}[\vect U] \cdot \vect{\delta U}
        &= \frac{\d \ZN{\vect U} \cdot \vect{\delta U}}{Z_N^2} \norm{\nabla_{\rm F} \vect \phi_N}_{}^2
        + \lim_{\varepsilon\rightarrow 0} \frac{ \ZN{\vect U}}{\varepsilon Z_N^2} \ip{\vect \phi^{\varepsilon}_N - \vect \phi_N, \vect f - I_N \vect 1}_{\e^{-V}}.
    \end{align}
    The functional derivative in the first term is given by
    \begin{equation}
        \label{eq:func_der_ztilde_discrete}
        \d\ZN{\vect U} \cdot {\vect \delta U} = - \ip{\vect {\delta U}, \vect 1}_{}.
    \end{equation}
    For the second term in~\eqref{eq:def_func_der_discrete},
    we obtain
    \begin{align*}
        \ip{\vect \phi^{\varepsilon}_N - \vect \phi_N, \vect f - I_N \vect 1}_{\e^{-V}}
        &= \ip*{\vect \phi^{\varepsilon}_N - \vect \phi_N, \diag\bigl(\exp.(\vect U) \bigr)(\vect f - I_N\vect 1)}_{\e^{-V-U}}  \\
        &= - \ip{\vect \phi^{\varepsilon}_N - \vect \phi_N, \mat L \vect \phi_N}_{\e^{-V-U}}\\
        &= - \ip{\mat L(\vect \phi^{\varepsilon}_N - \vect \phi_N), \vect \phi_N}_{\e^{-V-U}} \\
        &= - \ip{\mat L^{\varepsilon} \vect\phi^{\varepsilon}_N, \vect \phi_N}_{\e^{-V-U -\varepsilon \delta U}}
        + \ip{\mat L \vect \phi_N, \vect \phi_N}_{\e^{-V-U}} \\
        & \qquad - \ip{\mat L \vect \phi^{\varepsilon}_N, \vect \phi_N}_{\e^{-V-U}}
        + \ip{\mat L^{\varepsilon} \vect \phi^{\varepsilon}_N, \vect \phi_N}_{\e^{-V-U -\varepsilon \delta U}}.
    \end{align*}
    The first two terms in the last expression cancel out,
    and after substituting the expressions of~$\mat L$ and~$\mat L^{\varepsilon}$
    given in~\eqref{eq:definition_of_matrix_L} in the other two terms,
    we obtain
    \begin{align*}
        \ip{\vect \phi^{\varepsilon}_N - \vect \phi_N, \vect f - I_N\vect 1}_{\e^{-V}}
        &=
        - \ip{\mat D_{\rm B}^x (\mat M - \mat M^{\varepsilon}) \mat D_{\rm F}^x \vect \phi^{\varepsilon}_N, \vect \phi_N}_{1}
        - \ip{\mat D_{\rm B}^y (\mat M - \mat M^{\varepsilon}) \mat D_{\rm F}^y \vect \phi^{\varepsilon}_N, \vect \phi_N}_{1},
    \end{align*}
    where
    \[
        \mat M - \mat M^{\varepsilon} :=
        \diag\bigl(\exp.(-\vect V - \vect U)\bigr)
        - \diag\bigl(\exp.(-\vect V - \vect U - \varepsilon \vect {\delta U})\bigr).
    \]
    Noting that
    \[
        \lim_{\varepsilon\rightarrow 0} \frac{\mat M - \mat M^{\varepsilon}}{\varepsilon} = \diag\bigl(\exp.(-\vect V - \vect U)\bigr) \diag(\vect {\delta U}),
    \]
    we deduce that
    \[
    \lim_{\varepsilon\rightarrow 0} \frac{\ZN{\vect U}}{\varepsilon Z_N^2} \ip{\vect \phi^{\varepsilon}_N - \vect \phi_N, \vect f - I_N \vect 1}_{\e^{-V}}
    = \frac{\ZN{\vect U}}{Z_N^2} \ip*{\vect {\delta U}, \abs{\nabla_{\rm F} \vect \phi_N}^2}_{\e^{-V-U}}.
    \]
    Combining this equation with~\eqref{eq:def_func_der_discrete} and~\eqref{eq:func_der_ztilde_discrete},
    we deduce~\eqref{eq:funcder_discrete}.
\end{proof}

\paragraph{Acknowledgements.}
We are grateful to Andrew Duncan and Grigorios Pavliotis for useful discussions and for sharing with us their preliminary calculations on this problem.
The work of MC was supported by the Agence Nationale de la Recherche under grant ANR-20-CE40-0022 (SWIDIMS).
The work of TL, GS and UV was partially funded by the European Research Council (ERC) under the European Union's Horizon 2020 research and innovation programme (grant agreement No 810367),
and by the Agence Nationale de la Recherche (ANR) under grants ANR-21-CE40-0006 (SINEQ) and ANR-19-CE40-0010 (QuAMProcs).

\printbibliography

\end{document}